%% file: arxiv.tex
\documentclass{amsart}
\input{preamble}

\DeclareMathOperator{\tr}{tr}

\newcommand{\maf}[1]{\mathfrak{#1}}

\renewcommand{\R}{\mathbb{R}} 
\renewcommand{\C}{\mathbb{C}} 
\newcommand{\eps}{\epsilon}
\newcommand{\lad}{\lambda}
\newcommand{\si}{\sigma}

\newcommand{\rd}{\mathrm{d}}
\newcommand{\re}[1]{{\color{black} #1}}
\newcommand{\rre}[1]{{\color{black} #1}}

\newcommand{\magg}[1]{{\color{black} #1}}
\def \ww {\omega}
\def \d {\delta}

\def \w {\widetilde}
\def \q {\quad}

\def  \mi {{\bf 1}}
\def \bh {\mc{B}(\mc{H})}
\def \l {\langle}
\def \r {\rangle}
\def \si {\sigma}
\def \dd {\cdot}
\def \pa {\partial}

\usepackage{etoolbox}
\patchcmd{\thebibliography}
  {\settowidth}
  {\setlength{\itemsep}{0pt plus 0.1pt}\settowidth}
  {}{}
\apptocmd{\thebibliography}
  {\small}
  {}{}

\numberwithin{equation}{section}

\newcommand{\DeptMath}{Department of Mathematics, University of California, Berkeley, CA 94720, USA}
\newcommand{\LBLMath}{Applied Mathematics and Computational Research Division, Lawrence Berkeley National Laboratory, Berkeley, CA 94720, USA}

\newcommand{\cityMath}{Department of Mathematics,  City University of Hong Kong, Kowloon Tong, Hong Kong SAR}

\usepackage{makecell} 
\setcellgapes{2pt}

\begin{document}

\title[Quantum Gibbs samplers with KMS detailed balance condition]{Efficient quantum Gibbs samplers with Kubo--Martin--Schwinger detailed balance condition}
\date{Latest revision: \today}

\author[Z. Ding]{Zhiyan Ding}
\address{\DeptMath}
\email{zding.m@math.berkeley.edu}
\author[B. Li]{Bowen Li}
\address{\cityMath}
\email{bowen.li@cityu.edu.hk}
\author[L. Lin]{Lin Lin}
\address{\DeptMath; \LBLMath}
\email{linlin@math.berkeley.edu}

\begin{abstract}
Lindblad dynamics and other open-system dynamics provide a promising path towards efficient Gibbs sampling on quantum computers.  In these proposals, the Lindbladian is obtained via an algorithmic construction akin to designing an artificial thermostat in classical Monte Carlo or molecular dynamics methods, rather than treated as an approximation to weakly coupled system-bath unitary dynamics. Recently, Chen, Kastoryano, and Gily\'en (arXiv:2311.09207) introduced the first efficiently implementable Lindbladian satisfying the  Kubo--Martin--Schwinger (KMS) detailed balance condition, which ensures that the Gibbs state is a fixed point of the dynamics and is applicable to non-commuting Hamiltonians. This Gibbs sampler uses a continuously parameterized set of jump operators, and the energy resolution required for implementing each jump operator depends only logarithmically on the precision and the mixing time. In this work, we build upon the structural characterization of KMS detailed balanced 
Lindbladians by Fagnola and Umanit\`a, and develop a family of efficient quantum Gibbs samplers using a finite set of jump operators (the number can be as few as one), \re{akin to the classical Markov chain-based sampling algorithm. Compared to the existing works, 
our quantum Gibbs samplers have a comparable quantum simulation cost but with greater design flexibility and a much simpler implementation and error analysis.} Moreover, it encompasses the construction of Chen, Kastoryano, and Gily\'en as a special instance.

\end{abstract}

\maketitle

\section{Introduction}

For a given quantum Hamiltonian $H\in\CC^{N\times N}$, preparing the associated Gibbs state
$\sigma_\beta=e^{-\beta H}/\mc{Z}_\beta$  (also called quantum Gibbs sampling) has a wide range of applications in condensed matter physics, quantum chemistry, and optimization.
Here $N=2^n$ is the dimension of the underlying Hilbert space, $\beta$ is the inverse temperature, and 
$\mc{Z}_\beta=\tr(e^{-\beta H})$ is the partition function. \re{In this work,} we assume efficient quantum access to the Hamiltonian simulation $\exp(-itH),t\in\RR$, \re{and we measure} the cost of a quantum algorithm \re{in terms of} the total Hamiltonian simulation time of $H$. 

\re{Recently, there has been a revival of interest in designing quantum Gibbs samplers based on the Lindblad dynamics \cite{MozgunovLidar2020,RallWangWocjan2023,chen2021fast,ChenKastoryanoBrandaoEtAl2023,ChenKastoryanoGilyen2023,WocjanTemme2023}
\begin{equation*}
\partial_t\rho_t = \mc{L}^{\dag} \rho_t\,,
\end{equation*}
which is a Markovian semigroup on quantum states. For ease of exposition, we shall refer to both the generator $\mc{L}^\dag$ and its adjoint $\mc{L}$ as Lindbladian, which are 
related to each other via
\begin{equation} \label{eq:adjointl}
    \l  X, \mc{L}(Y) \r = \l  \mc{L}^\dag (X), Y\r\,, \quad \text{for $X, Y\in\CC^{N\times N}$}\,.
\end{equation}
Here $\l \dd, \dd \r$ is the Hilbert--Schmidt inner product for matrices, and we follow the convention that the adjoint $\mc{L}^{\dag}$ denotes the Lindbladian 
in the Schr\"{o}dinger picture acting on a quantum state. 
A common choice of Lindbladian for quantum Gibbs sampling is the Davies generator, which is originally derived from the weak-coupling limit of open quantum dynamics \cite{Davies1974,Davies1976,davies1979generators}:  
\begin{equation} \label{eqq:davies}
    \mc{L}^\dag_{\rm Davies}(\rho) = - i [H + H_{\rm LS}, \rho] + \sum_{a,\,\ww} \gamma_a(\ww) \left(\widehat{A}^a(\ww) \rho  \widehat{A}^a(\ww)^\dag - \frac{1}{2}\left\{\widehat{A}^a(\ww)^\dag  \widehat{A}^a(\ww), \rho \right\}  \right)\,, 
\end{equation}
with the functions $\gamma_a(\omega)$ and  operators $\widehat{A}^a(\ww)$ satisfying 
\begin{equation*}
    \gamma_a(-\ww) = e^{\beta \ww} \gamma_a(\ww)\,,\q \si_\beta \widehat{A}^a(\ww) \si_\beta^{-1} = e^{-\beta \ww} \widehat{A}^a(\ww)\,,
\end{equation*}
where $\{X, Y\} = XY + YX$ stands the anticommutator and the sum is over $a$ and $\ww$ from some discrete index sets, representing the number of couplings to the environment and Bohr frequencies, respectively. Here the commutator term involving $H + H_{\rm LS}$ is called the coherent part, and $H_{\rm LS}$ is a correction Hamiltonian called the Lamb shift satisfying $[H_{\rm LS}, H] = 0$. 
A key property for using the Davies generator as a quantum Gibbs sampler \cite{RallWangWocjan2023,WocjanTemme2023} is that the associated dynamics $e^{t \mc{L}^\dag_{\rm Davies}}$ admits the Gibbs state as a fixed point, equivalently, $\mc{L}^\dag_{\rm Davies}(\sigma_\beta)=0$, which can be directly checked from \eqref{eqq:davies} (see \cref{sec:daveisgns} for a more detailed discussion).}

\re{In general, for a fixed Gibbs state $\si_\beta$,  \rre{it is not an easy task to design a Lindbladian satisfying the fixed-point property $\mc{L}^\dag(\si_\beta) = 0$.} A \rre{sufficient condition guaranteeing this} is the quantum \rre{generalization of the} detailed balance condition (DBC) \cite{kossakowski1977quantum,fagnola2007generators,fagnola2010generators} (see \cref{def:sidbc}). The quantum DBC is a set of criteria for the equilibria of quantum systems, and is crucial in understanding the dynamics and thermodynamics of such systems.} 
Moreover, for a given precision $\epsilon>0$, we define the mixing time of a Lindblad dynamic with $\si_\beta$ as a unique fixed point by 
\begin{equation} \label{eqn:mixtime}
t_{\mathrm{mix}} := \inf\left\{t\ge 0\,;\  \left\|e^{t \mathcal{L^{\dag}} }\left(\rho\right)-\sigma_\beta\right\|_{1} \leq \epsilon, \  \forall \, \text{\re{quantum states} $\rho$}\right\}\,,
\end{equation}
\re{which quantifies the convergence speed of the dynamics,} where $\norm{\cdot}_1$ denotes the trace norm. The quantum DBC is also instrumental for proving the finite mixing time (if it is the case), and its scaling with respect to $\beta$ and the system size~\cite{TemmeKastoryanoRuskaiEtAl2010,KastoryanoTemme2013,BardetCapelGaoEtAl2023,rouz2024}.
However, due to the non-commutativity of operators in quantum mechanics, there is no unique definition of
quantum DBC~\cite{TemmeKastoryanoRuskaiEtAl2010,CarlenMaas2017,CarlenMaas2020}.


  
\subsection{Related works}

The most widely studied form of quantum DBC is in the sense of Gelfand--Naimark--Segal (GNS). The seminal result by Alicki~\cite{Alicki1976} characterizes the class of Lindbladians satisfying the GNS DBC, which turns out to have the same dissipative part as that of the Davies generator; see \cref{sec:daveisgns}. 
\re{The Davies generator $ \mc{L}_{\rm Davies}$ is typically derived as a simplified representation of weakly interacting system-bath models, following the Born-Markov-Secular\footnote{The secular approximation is also referred to as the rotating wave approximation (RWA).} approximation route \cite{breuer2002theory,Lidar2019}, and its applicability range seems to be constrained by the limitations of these ideal approximations. However, for the purpose of Gibbs sampling, we may also view the Davies generator as a purely algorithmic process, without mimicking any specific system-bath unitary dynamics in nature. A much more severe issue is that the implementation of $\widehat{A}^a(\omega)$ requires the resolution of all Bohr frequencies, which is very challenging for most Hamiltonians (see also \cref{rem:ineffi}). For this reason, earlier works on the quantum thermal state preparation based on Davies generators often suffer from some unphysical assumptions, such as the rounding promise \cite{RallWangWocjan2023,WocjanTemme2023}.
 

In what follows, $\mc{A}$ denotes a discrete index set.} The key object in the work \cite{ChenKastoryanoBrandaoEtAl2023} is the following frequency-dependent jump operator:
\begin{equation}\label{eqn:freq_jump}
  \widehat{A}^a_f(\omega) := \int^\infty_{-\infty}f(t)e^{-i\omega t} e^{iHt}A^ae^{-iHt} \ud t, \quad \omega\in\RR\,, \quad a\in\mc{A}\,,
\end{equation}
which is a $f$-weighted Fourier transform of the Heisenberg evolution $A^a(t) := e^{iHt}A^{a}e^{-iHt}$ of $A^a$.
Here $f:\RR\to \mathbb{R}$ \re{is a \emph{real} integrable   filtering function,} and 
$\{A^a\}_{a\in\mc{A}}$ is a set of (frequency-independent) coupling operators provided by the user, which represent the coupling between the system and the fictitious environment, akin to designing an artificial thermostat in the classical Monte Carlo or molecular dynamics methods~\cite{FrenkelSmit2002}. The choice of $\{A^a\}_{a\in\mc{A}}$ can be flexible and relatively simple  (such as Pauli operators). The algorithmic Lindbladian in the Schr\"{o}dinger picture then takes the form
\begin{equation}\label{eqn:lindblad_filter}
  \mc{L}^{\dag}(\rho) = 
  \sum_{a\in\mathcal{A}}\int^\infty_{-\infty}\gamma(\omega)\left(\widehat{A}_f^a(\omega)\rho \widehat{A}_f^a(\omega)^\dagger-\frac{1}{2}\Big\{\widehat{A}_f^a(\omega)^\dagger \widehat{A}_f^a(\omega),\rho\Big\}\right)\ud \omega\,,
\end{equation}
\re{where $\gamma(\ww): \RR \to [0,1]$ is the transition weight function satisfying $ \gamma(-\ww) = e^{\beta \ww} \gamma(\ww)$.
One may easily see from \eqref{eqq:davies} that the Davies generator can formally correspond to taking the filtering function $f(t)\equiv 1$ in \eqref{eqn:lindblad_filter}, up to a coherent part.}
Unfortunately, in this case, since $f$ does not decay in $t$, the exact implementation of $\widehat{A}^a_f(\omega)$ in 
\eqref{eqn:freq_jump} requires simulating the Heisenberg evolution $A^a(t)$ for an infinitely long period. This implies that 
in the frequency space, the energy levels of $H$ have to be distinguished to infinite precision, which cannot be achieved 
in general, except for some special systems such as Hamiltonians with commuting terms. \re{A natural idea is} to select certain decaying filtering functions $f$ in \eqref{eqn:lindblad_filter} such that $\mc{L}$ approximates a Davies generator,  \re{whereas in this case, the generators can only satisfy the \emph{approximate}} GNS DBC and their fixed points are not known \emph{a priori}. 
Consequently, estimating the deviation of the fixed point of the approximate dynamics from $\sigma_\beta$ involves tracking the accumulated error along the dynamic trajectory. \re{It was proved in 
\cite[Theorems I.1, I.3]{ChenKastoryanoBrandaoEtAl2023} 
that} to approximate the Gibbs state to precision $\epsilon$, the dynamics should distinguish the energy levels of $H$ to precision $\poly(\beta^{-1} t_{\mathrm{mix}}^{-1} \epsilon)$, where $t_{\rm mix}$ is given in \cref{eqn:mixtime}. 
It means that the integral in \cref{eqn:freq_jump} could be truncated to $T=\poly(\beta t_{\mathrm{mix}} \epsilon^{-1})$. However, this cost may still be prohibitively high for practical applications. 

 

In the \rre{follow-up} work \cite{ChenKastoryanoGilyen2023}, Chen et al. introduced a new quantum Gibbs sampler based on \eqref{eqn:lindblad_filter}, \re{which requires a lower energy resolution in simulating  $\widehat{A}^a_f(\omega)$ and hence has the reduced computational complexity, by constructing a coherent term $- i [G, \rho]$ to $\mc{L}^\dag$ in \eqref{eqn:lindblad_filter} to make the total Lindbladian $\w{\mc{L}}(X) =i [G, X] + \mc{L}(X)$} exactly satisfies a less stringent version of DBC called the Kubo--Martin--Schwinger (KMS) DBC. \re{To be specific, \cite[Section II.C]{ChenKastoryanoGilyen2023} showed that the transition part of $\mc{L}^\dag$ in \eqref{eqn:lindblad_filter}:
\begin{equation} \label{eq:transi}
    \mc{T}(\rho) = \sum_{a\in\mathcal{A}}\int^\infty_{-\infty}\gamma(\omega) \widehat{A}_f^a(\omega)\rho \widehat{A}_f^a(\omega)^\dagger \ud \omega
\end{equation}
satisfies the $\sigma$-KMS DBC if $f$ and $\gamma$ are Gaussians with some parameter constraints. Furthermore, if $f(t)$ is chosen to be a real function, then the Gaussian function turns out to be the unique choice (up to a linear combination of $\gamma$, see \cite[Appendix D]{ChenKastoryanoGilyen2023}). On the other hand, \cite[Lemma II.1]{ChenKastoryanoGilyen2023} showed one can always design a coherent term $G$ such that the superoperator $\mc{R}(\rho) - i[G,\rho]$ is KMS detailed balanced, where $\mc{R}(\dd)$ is the decay part in \eqref{eqn:lindblad_filter}:
\begin{equation} \label{eq:decay}
    \mc{R}(\rho) = 
  -\frac{1}{2}\sum_{a\in\mathcal{A}}\int^\infty_{-\infty}\gamma(\omega)\Big\{\widehat{A}_f^a(\omega)^\dagger \widehat{A}_f^a(\omega),\rho\Big\}\ud \omega\,.
\end{equation} 
  Thanks to the \emph{exact} KMS DBC, the Gibbs state $\sigma_\beta$ remains a fixed point of the Lindblad dynamics, and one no longer needs to keep track of the accumulated deviation and distinguish the energy levels of $H$ to very high precision.}
\re{In this setup, it was proved \cite[Theorem I.2]{ChenKastoryanoGilyen2023} that} 
to prepare $\sigma_\beta$ to precision $\epsilon$, the integral in \cref{eqn:freq_jump} can be truncated to $T=\Or(\beta \log(t_{\mathrm{mix}}/\epsilon))$ and the cost can be significantly reduced. 

\rre{It is worth pointing out that for each jump proposal $A^a$,} the integral with respect to $\omega$ in \cref{eqn:lindblad_filter} involves a continuously parameterized set of jump operators 
if $\gamma(\omega)$ is a continuous function. Although one can discretize 
such an integral using a quadrature scheme, the algorithm must be
meticulously designed to efficiently simulate the resulting Lindblad
dynamics \cite{CW17,ChenKastoryanoBrandaoEtAl2023,
ChenKastoryanoGilyen2023}. To the best of our knowledge,  high-order
Lindblad simulators designed for a finite number of jump operators,
which allow for simpler implementations \cite{LiWang2023,DingLiLin2024}, are not suitable for this task.

\re{
Along with \cite{ChenKastoryanoBrandaoEtAl2023,ChenKastoryanoGilyen2023}, there are other algorithms \cite{Rall_thermal_22} for the thermal state preparation also aiming to approximate the GNS detailed balanced Lindbladian (i.e., the Davies generator). 
Similar to \cite{ChenKastoryanoBrandaoEtAl2023}, the truncation time of these methods often depends polynomially on  $1/\epsilon$, resulting in the total cost scales polynomially with $1/\epsilon$. \cref{tab:comparison} summarizes the total costs of the algorithms in \cite{ChenKastoryanoBrandaoEtAl2023,ChenKastoryanoGilyen2023} and some earlier related works. A more comprehensive overview of quantum Gibbs samplers can be found in \cite{ChenKastoryanoBrandaoEtAl2023} and references therein.}

In parallel to Alicki's characterization of Lindblad generators satisfying the GNS DBC, Fagnola, and Umanit\`a~\cite[Theorems 7.2 and 7.3]{fagnola2007generators} have prescribed the necessary and sufficient conditions for a Lindblad dynamics to satisfy the KMS DBC.
This leads to a set of conditions on the jump operators and the Hamiltonian for its corresponding Lindblad generator; see also \cite{fagnola2010generators,AmorimCarlen2021}. \re{ 
However, it is unclear from the results in \cite{fagnola2007generators} how to choose suitable jump operators for the quantum thermal state preparation, and they also did not provide an explicit form of the coherent term required for the KMS detailed balance. In this sense,  \cite[Lemma II.1]{ChenKastoryanoGilyen2023} shares some commonalities to the results in \cite{fagnola2007generators} and provides a more explicit form of the Lindbladian.} 

\subsection{Contribution and outline}
\re{In this work, we first provide a detailed introduction to the canonical forms of Lindbladians with various detailed balance conditions in \cref{sec:dbclindblad}. In particular, in \cref{sec:kmsdbc}, we revisit the arguments in \cite{fagnola2007generators}, along the lines of the recent work \cite{AmorimCarlen2021}, and provide a self-contained construction for the $\si_\beta$-KMS detailed balanced Lindbladian with explicit formulations of the jumps and the coherent term (see \cref{prop:KMS_DBC} and \cref{eqn:V_first}). This is the starting point of our quantum Gibbs
samplers constructed and analyzed in \cref{sec:QGS}.}

\re{The idea underlying our quantum Gibbs sampler is analogous to the classical Markov chain-based sampling. In the classical context, a state undergoes a candidate jump according to a transition matrix. The proposed new state is then accepted with a certain probability to ensure that the resulting modified transition matrix satisfies a detailed balance. A natural question that arises when designing efficient quantum Gibbs samplers is:
\begin{quote}
    \emph{How to modify a predefined set of jump proposals (i.e., coupling operators) such that the resulting Lindblad dynamics satisfies the quantum detailed balance for the target Gibbs state and meanwhile can be efficiently simulated on a quantum computer?}
\end{quote}
Based on the structural characterization in \cref{prop:KMS_DBC}, we provide an affirmative answer to this question. Our approach enables the use of general filter functions to modify a set of predefined coupling operators, resulting in dynamics that can be efficiently simulated and exactly satisfy the KMS detailed balance condition.


In contrast to \cite{ChenKastoryanoBrandaoEtAl2023,ChenKastoryanoGilyen2023}, which can be regarded as smoothened approximations of Davies generators, our new framework provides significantly more flexibility in designing quantum Gibbs samplers. Recall from  \cite[Appendix D]{ChenKastoryanoGilyen2023} that to make the transition part \eqref{eq:transi} satisfy the KMS DBC, the filtering function $f$ in \eqref{eqn:freq_jump} has to be a Gaussian if it is assumed to be real, and then the outer integration in \eqref{eq:transi} with respect to a continuous function $\gamma(\ww)$ is necessary. Intuitively, this is because when $f$ is real, its Fourier transform must be conjugate symmetric. We need an additional integration to induce the biased energy transition. However, if the main goal is to ensure the KMS detailed balance (instead of approximating Davies generators), by considering a general class of filtering functions with a mild symmetry condition (see \eqref{eq:addjump}--\eqref{eq:integ_jump} below), we find that $f$ can be relaxed to be \emph{complex} and the outer integration in $\omega$ can be removed. 

To be specific, in \cref{sec:gibbsampler}, we show that the admissible jump $L$ for KMS detailed balance is of the form: for a self-adjoint operator $A$ with $A_\nu = \sum_{\lambda_i - \lambda_j = \nu} P_i A P_j$,
\begin{equation*}
    L = \sum_{\nu \in B_H} e^{- \beta \nu/4} A_\nu\,,
\end{equation*}
where $B_H$ is the set of Bohr frequencies (see \eqref{def:Bohr}--\eqref{eqn:A_v} for related definitions). We may overparameterize it 
by adding a \emph{weighting function} $q(\nu)$:
\begin{equation}  \label{eq:addjump}
       L = \sum_{\nu \in B_H} \widehat{f}(\nu) A_\nu\,, \quad \widehat{f}(\nu)=  e^{- \beta \nu/4} q(\nu)\,,
\end{equation}
with $q$ being conjugate symmetric $q(-\nu) = \overline{q(\nu)}$, and then, without loss of generality, we can assume the operator norm $\|A\| \le 1$. Here, intuitively, the Fourier transform of the filtering function $\widehat{f}$ helps modify the transition probability between eigenstates with an energy difference $\nu$ to satisfy the KMS detailed balance. This jump operator can be efficiently simulated in the time domain: 
\begin{equation} \label{eq:integ_jump}
    L = \int_{-\infty}^{\infty} f(t)e^{iHt}Ae^{-iHt}\, \mathrm{d}t \q \text{with}\q f(t) = \frac{1}{2 \pi} \int_{-\infty}^\infty \widehat{f}(\nu) e^{- i \nu t} \ud \nu\,,
\end{equation}
when $f(t)$ decays rapidly as $\abs{t}\to \infty$ (see also \eqref{defaf}--\eqref{eq:fcalculus}). This is the main motivation for us to consider the \emph{Gevrey function} class with compact support in \cref{assumption:q} as weighting functions $q(\nu)$ \cite{Adwan_2017,Gus_2019}, whose Fourier transform has the desired sub-exponential decay as $|t| \to \infty$ (see Lemmas  \ref{lem:hat_Gev_decay} and \ref{lem:property_f_g}), so that one could approximate the integral \eqref{eq:integ_jump} accurately with a moderate truncation time $\wt{\Or}(\beta\mathrm{log}^{1+o(1)}(t_{\rm mix}/\epsilon))$ (see \cref{lem:discretization_L_G} and \cref{thm:simulation}). 
Once we design a family of efficient jumps $\{L_a\}$ via $\{(A_a,f^a)\}$, the coherent term can be readily defined via \cref{eqn:G_formula_KMS} in \cref{prop:KMS_DBC}, admitting a similar time integral formulation as \eqref{eq:integ_jump} (see \cref{eqn:G_KMS,def:hll}). 
The quantum algorithm for the proposed quantum Gibbs sampler is discussed in detail in \cref{sec:simulation}. In \cref{sec:recover}, we elaborate on how to fit the construction of~\cite{ChenKastoryanoGilyen2023} into our algorithmic framework as a special instance. }


\re{In addition, we emphasize that compared to the algorithm in \cite{ChenKastoryanoGilyen2023}, our Gibbs sampler \rre{only uses one jump operator per jump proposal $A^a$}. In particular, we can directly achieve a Metropolis-type transition using our framework while exactly satisfying the KMS DBC (see \cref{fig:f}). This greatly simplifies the quantum simulation process. Thanks to \cref{lem:discretization_L_G} built on the technical \cref{assumption:q} via the Poisson summation formula,} our jump operators can be constructed using the standard linear combination of
unitaries (LCU) routine~\cite{BerryChildsCleveEtAl2014,
GilyenSuLowEtAl2019}. As a result, our Lindblad dynamics can be efficiently simulated using \emph{any} high-order simulation algorithms, including those in \cite{LiWang2023,DingLiLin2024}. \re{We show in \cref{thm:simulation} that our approach and that in \cite{ChenKastoryanoGilyen2023} have a comparable computational cost when the best
available Lindblad simulation algorithms with near-optimal costs are used. 
One should note that \cref{assumption:q} (Gevrey-type weighting functions) is mainly for simplifying the discretization error analysis and not essential for the design of our quantum Gibbs sampler. We believe that a broader class of smooth functions can also serve as weighting functions without affecting the total quantum algorithm cost. }






\begin{table}[htbp!]
\begin{adjustbox}{width=\textwidth} 
\makegapedcells
\centering
\begin{tabular}{c|cccc|c}
\hline
\hline
\textbf{Algorithms} & \multicolumn{4}{c|}{\textbf{Properties}} & \textbf{Remark} \\
 & Detailed & Truncation     & \# Jump          & Total &  \\
 & balance  & time & \rre{per $A^a$}     & cost  &  \\ \hline
\cite{chen2021fast} & $\approx$ GNS & N/A & $\infty$ & $\poly(\beta\epsilon^{-1} t_{\mathrm{mix}})$ & \begin{tabular}[c]{@{}c@{}} Weak coupling\\Refreshable bath
\end{tabular} 
\\ \hline
\cite[Theorem 1]{RallWangWocjan2023} & $\approx$ GNS & $\wt{\Or}(\beta \epsilon^{-2})$ & $\infty$ & $\wt{\Or}(\beta^3 t_{\mathrm{mix}}\epsilon^{-7})$ & Rounding promise
\\ \hline
\cite[Theorem I.1]{ChenKastoryanoBrandaoEtAl2023} & $\approx$ GNS & $\wt{\Or}(\beta t_{\mathrm{mix}}^2 \epsilon^{-2})$ & $\infty$ & $\wt{\Or}(\beta t_{\mathrm{mix}}^3 \epsilon^{-2})$ & Rectangular filter
\\ \hline
\cite[Theorem I.3]{ChenKastoryanoBrandaoEtAl2023} & $\approx$ GNS & $\wt{\Or}(\beta t_{\mathrm{mix}} \epsilon^{-1})$ & $\infty$ & $\wt{\Or}(\beta t_{\mathrm{mix}}^2  \epsilon^{-1})$ & Gaussian filter
\\ \hline
\cite[Theorem I.2]{ChenKastoryanoGilyen2023} & KMS & $\wt{\Or}(\beta \log(t_{\mathrm{mix}}/\epsilon))$ & $\infty$ & $\wt{\Or}(\beta t_{\mathrm{mix}}  \mathrm{polylog}(1/\eps))$ & \begin{tabular}[c]{@{}c@{}} Gaussian / \\ Metropolis filter
\end{tabular}
\\ \hline
This work [\cref{thm:simulation}]& KMS & $\wt{\Or}(\beta\mathrm{log}^{1+o(1)}(t_{\rm mix}/\epsilon))$ & \rre{$1$} & \re{$\wt{\Or}( \beta t_{\mathrm{mix}} \mathrm{polylog}(1/\epsilon))$} & A family of filters
\\ \hline\hline 
\end{tabular}
\end{adjustbox}
\vspace{0.5em}
\caption{\re{A comparison of some quantum Gibbs samplers using techniques related to Lindblad dynamics \rre{with $\beta > 1$}. Here the truncation time $T$ is used to truncate the integral in \cref{eqn:freq_jump} to $[-T, T]$ in simulation. \rre{For each jump proposal $A^a$}, the number of jump operators for the Lindbladian is denoted by $\infty$ if $\gamma(\omega)$ is a continuous function in \cref{eqn:lindblad_filter}.
The total cost refers to the total Hamiltonian simulation time using the best available Lindblad simulation algorithm. The $o(1)$  factor in the truncation time of this work stems from our choice of the filtering function and $o(1)$ can be chosen to be arbitrarily small without much added cost. The mixing time 
$t_{\mathrm{mix}}$ is method-dependent, and its value is generally difficult to compare with each other in theory.
The weak coupling assumption may be unphysical for large quantum systems. The rounding promise-based method prepares an ensemble of density operators, and the total cost for preparing each density operator in the ensemble may be improved to 
$\wt{\Or}(\beta t_{\rm mix} \epsilon^{-2})$ by the improved Lindbladian
simulation technique from \cite{ChenKastoryanoBrandaoEtAl2023}
(see \cite[Remark of Theorem 1]{RallWangWocjan2023}).}}
\label{tab:comparison}
\end{table}


\subsection{Discussion and open questions}

\rre{Despite their importance, the computational complexity of preparing quantum Gibbs states, and the associated task of computing partition functions, is a topic that is being actively debated in the literature~\cite{bravyi2021complexity,bravyi2024quantum,BakshiLiuMoitraEtAl2024}. It is particularly challenging to establish quantum computational advantage for Gibbs state preparation at a constant temperature. This is because on one hand, at high enough (constant) temperatures, there exist polynomial-time classical algorithms to sample from Gibbs states and to estimate partition functions~\cite{bravyi2021complexity,MannHelmuth2021,YinLucas2023,BakshiLiuMoitraEtAl2024}. On the other hand, in the low-temperature regime, preparing classical Gibbs states is already \NP-hard in the worst case~\cite{Barahona1982,Sly2010}. 
When the temperature is sufficiently low, the quantum Gibbs state can have high overlap with the ground state, and the task of cooling to these temperatures should be \QMA-hard in the worst case. So we do not expect generically efficient quantum algorithms in these cases. In this line, recent developments~\cite{rouz2024,bergamaschi2024quantum,rajakumar2024gibbs} have used Gibbs samplers to provide new insights from the complexity theory perspective. Specifically, let $n$ be the number of qubits. \cite{bergamaschi2024quantum} proved that there exists a family of $\Or(\log \log n)$-local Hamiltonians such that sampling from their associated Gibbs states at constant temperatures can be achieved in time $T = \poly(n)$ by the Davies generator~\cite{Davies1974} with the block-encoding framework~\cite{ChenKastoryanoBrandaoEtAl2023}, but is classically intractable assuming no collapse of the polynomial hierarchy. 
Under the same assumptions, \cite{rajakumar2024gibbs} recently improved this result to $k$-local Hamiltonians at a constant temperature lower than the classically simulatable threshold. Furthermore, \cite{rouz2024} showed that simulating the Lindbladian in \cite{ChenKastoryanoGilyen2023} at $\beta = \Omega(\log(n))$ to runtime $T=\Or(\poly(n))$ for a $k$-local Hamiltonian is \BQP-complete.
Consequently, the development of practical quantum Gibbs samplers, which is the focus of our work, emerges as a critical endeavor in realizing these potential quantum advantages.}

There exist a series of algorithms~\cite{PoulinWocjan2009,ChowdhurySomma2017,VanApeldoornGilyenGriblingEtAl2017,GilyenSuLowEtAl2019,an2023quantum} that require only quantum access to $H$ without additional information (such as coupling operators). The cost of these algorithms is deterministic and scales as $\Or(\sqrt{N/\mc{Z}_\beta}\poly(\beta,\log\epsilon^{-1}))$. \re{These algorithms can perform efficiently in the extremely high-temperature regime, where $\beta$ approaches to zero and $\sqrt{N/\mc{Z}_\beta}\sim 1$ (assuming the smallest eigenvalue of $H$ is zero).} However, they become significantly less efficient in the low-temperature regime, where $\beta$ is large, and $\sqrt{N/\mc{Z}_\beta}\sim \sqrt{N}$. 
Furthermore, the factor $\sqrt{N/\mc{Z}_\beta}$ is explicitly present in the algorithm, and the average-case complexity is not very different from the worst-case scenario.
On the other hand,  the computational cost of open-system quantum dynamics is primarily determined by the mixing time, which can vary significantly across different systems. Besides the Lindblad dynamics, alternative open-system dynamics formalisms are also viable~\cite{Temme_2011,Man_2012,shtanko2023preparing,cubitt2023dissipative} for Gibbs state preparation.
We may anticipate that for certain classes of physical Hamiltonians, even at low temperatures, Gibbs sampling could be executed efficiently.  This possibility does not contradict the statement that preparing the ground state of $H$ (when $\beta=\infty$) remains QMA-hard in the worst-case scenario~\cite{KitaevShenVyalyi2002,AharonovGottesmanEtAl2009}.

\rre{We emphasize that in this work, an \emph{efficient} quantum Gibbs sampler means that the proposed Gibbs sampler can be efficiently simulated on a quantum computer, while in classical MCMC literature, \emph{efficient} often indicates rapid mixing. The rigorous bounding of the mixing time has been achieved for Davies generators for some general families of local commuting Hamiltonians~\cite{KastoryanoBrandao2016,BardetCapelGaoEtAl2023,kochanowski2024rapid},} whereas establishing the mixing time for non-commuting Hamiltonians at moderate or even low temperatures presents a substantial theoretical challenge. There are two interesting works along this line. The first is that Rouz\'e et al~\cite{rouz2024} established the spectral gap of KMS detailed balanced Lindbladians for certain $k$-local Hamiltonians at high temperatures using Lieb-Robinson estimates. 
The analysis in \cite{rouz2024} may be applicable in our setting and we plan to investigate this in detail in a future work. 
On the other hand, Bakshi et al~\cite{BakshiLiuMoitraEtAl2024} demonstrated that at the high temperature, the Gibbs State of certain $k$-local Hamiltonian becomes a linear combination of tensor products of stabilizer states, can be prepared in polynomial time using randomized classical algorithms. This suggests that exploring the relationship between the complexity of Gibbs states and mixing times could be a fruitful avenue for future research. 



It is also noteworthy that introducing a nontrivial coherent term $i[G, \dd]$ with $[G,\sigma_\beta] = 0$ to any
$\si_\beta$-KMS detailed balanced Lindbladian \rre{potentially disrupts} the detailed balance condition, but the Gibbs state remains a fixed point. The influence of the coherent term on the mixing time may be significant and its characterization \re{remains a largely open question. We refer interested readers to \cite{fang2024mixing,li2024quantum} for the recent progress on analyzing the coherent term effect via hypocoercivity.} 
Finally, the mixing time $t_{\rm mix}$ may be very different across different quantum Gibbs samplers. Both theoretical and numerical evidence are needed in order to quantify the mixing time and to compare the efficiency of quantum Gibbs samplers for physical systems of interest.

\subsection{Notation}


We denote by $\mc{H}$ a finite-dimensional Hilbert space with dimension $N = 2^n$, and by $\mc{B}(\mc{H})$ the space of bounded operators.  For simplicity, we usually write $A \ge 0$ (resp., $A > 0$) for a positive semidefinite (resp., definite) operator. The identity element in $\bh$ is denoted by $\mi$. Moreover, we denote by $\mc{D}(\mc{H})$ the set of quantum states (i.e., $\rho \ge 0$ with $\tr(\rho) = 1$), and $\mc{D}_+(\mc{H})$ the subset of full-rank states. Let $X^\dag$ be the adjoint operator of $X$. We denote by $\l\dd,\dd\r$ the Hilbert--Schmidt inner product on $\bh$: $\l X, Y\r := \tr (X^\dag Y)$. Then, with some slight abuse of notation, the adjoint of a superoperator $\Phi: \mc{B}(\mc{H}) \to \mc{B}(\mc{H})$ with respect to $\l \dd,\dd \r $ is also denoted by $\Phi^\dag$. Unless specified otherwise, $\norm{X}$ denotes the operator norm for $X \in \mc{B}(\mc{H})$, while $\norm{x}_s := (\sum_{j} \abs{x_j}^s)^{1/s}$ denotes the $s$-norm of the vector $x \in \C^N$ ($s \ge 1$). \rre{The diamond norm of a superoperator $\mc{E}$ on $\mc{B}(\mc{H})$ is defined by $\norm{\mc{E}}_{\Diamond} = \norm{\mc{E} \otimes {\rm id}}_1$, where ${\rm id}$ is the identity map on $\mc{B}(\mc{H})$.}

We adopt the following asymptotic notations beside the usual big $\Or$ one. We write $f=\Omega(g)$ if $g=\Or(f)$; $f=\Theta(g)$ if $f=\Or(g)$ and $g=\Or(f)$. The notations $\wt{\Or}$, $\wt{\Omega}$, $\wt{\Theta}$ are used to suppress subdominant polylogarithmic factors. Specifically, $f = \wt{\Or}(g)$ if $f = \Or(g\operatorname{polylog}(g))$; $f = \wt{\Omega}(g)$ if $f = \Omega(g\operatorname{polylog}(g))$; $f = \wt{\Theta}(g)$ if $f = \Theta(g\operatorname{polylog}(g))$. Note that these tilde notations do not remove or suppress dominant polylogarithmic factors. For instance, if $f=\Or(\log g \log\log g)$, then we write $f=\wt{\Or}(\log g)$ instead of $f=\wt{\Or}(1)$.   

\subsection*{Acknowledgement}
This material is based upon work supported by the Challenge Institute for Quantum Computation (CIQC) funded by National Science Foundation (NSF) through grant number OMA-2016245 (Z.D.),
National Science Foundation (NSF) award under grant number DMS-2012286 and CHE-2037263 (B.L.), the  Applied Mathematics Program of the US Department of Energy (DOE) Office of Advanced Scientific Computing Research under contract number DE-AC02-05CH1123, and a Google Quantum Research Award (L.L.). L.L. is a Simons investigator in Mathematics. We thank Anthony Chen, \re{Andr\'as Gily\'en}, Li Gao, Marius Junge and Jianfeng Lu for insightful discussions.   


\vspace{1em}
\noindent \textit{Note:} In completing this work, we became aware of the concurrent research by Chen, Doriguello, and Gily\'en~\cite{gilyen2024quantumgeneralizations}, which similarly aims to develop efficiently simulable quantum Gibbs samplers with a finite number of jump operators by coherent reweighing in both continuous-time and discrete-time settings, in analog with the classical Metropolis sampling.

\section{Structures of detailed balanced Lindbladians} \label{sec:dbclindblad}

In this section, we present the canonical forms of the Lindbladians with detailed balance conditions and discuss their feasibility for implementation on a quantum computer. 


We first recall that a quantum channel $\Phi: \bh \to \bh$ is a completely positive trace preserving (CPTP) map, and a quantum Markov semigroup (QMS)\footnote{\re{With some abuse of terminology, we will refer to both $\mc{P}_t$ and its adjoint dynamics $\mc{P}^\dag_t$ as a QMS.}} $(\mc{P}_t)_{t \ge 0}: \mc{B}(\mc{H}) \to \mc{B}(\mc{H})$, also called Lindblad dynamics, is defined as a $C_0$-semigroup of completely positive, unital maps. The generator $$\mc{L}(X): = \lim_{\re{t \to 0^+}} t^{-1}(\mc{P}_t (X) - X)$$ is usually referred to as the Lindbladian, which has the following GKSL form \cite{Lindblad1976,GoriniKossakowskiSudarshan1976}.
\begin{lem} \label{lem:lgksform}
For any generator $\mc{L}$ of a QMS $\mc{P}_t$, there exist operators $L_j, K \in \bh$ such that
\begin{equation}\label{GKSL_1}
    \mc{L} (X) = \Psi(X) + K^\dag X + X K\,,
\end{equation}
where $\Psi(\dd)$ is completely positive with the Kraus representation:
\begin{equation} \label{eq:kraus}
    \Psi(X) = \sum_{j \in \mc{J}} L_j^\dag X L_j\,,
\end{equation}
with $\mc{J}$ being an index set with cardinality $|\mc{J}| \le N^2$. \rre{In particular, letting $G:= \frac{K^\dag - K}{2 i}$, we have the Lindbladian form:
\begin{equation}\label{GKSL_2}
     \mc{L} (X) = i [G, X] + \sum_{j \in \mc{J}} \left( L_j^\dag X L_j - \frac{1}{2}\left\{L_j^\dag L_j, X  \right\} \right)\,,
\end{equation}
where $i [G, X]$ and $\mc{L} (X) -  i [G, X]$ are the coherent and dissipative parts of dynamics, respectively.} 
\end{lem}

The operators $L_j$ in \eqref{GKSL_2} are called jump operators, which are non-unique for a given Lindbladian. \rre{The formula \eqref{GKSL_2} is a direct consequence of \eqref{GKSL_1}. In fact, from $\mc{L}(\mi) = 0$ with \eqref{GKSL_1}, the operator $K$ can be written as $K = V - i G$ with $V = - \frac{1}{2} \sum_{j \in \mc{J}} L_j^\dag L_j$, where $V := \frac{K^\dag + K}{2}$ and $G := \frac{K^\dag - K}{2 i}$ are self-adjoint operators. Plugging this into \eqref{GKSL_1} readily gives \eqref{GKSL_2}.}

For the purpose of quantum state preparation, we are interested in those QMS converging to a given full-rank state $\si > 0$, i.e., 
\begin{align} \label{eq:conver_qms}
    \lim_{t \to \infty} \mc{P}^\dag_t(\rho) = \si\,,\q \forall \rho \in \mc{D}(\mc{H})\,, 
\end{align}
equivalently, the irreducible QMS $\mc{P}^\dag_t$ \cite[Proposition 7.5]{wolf5quantum}. 
For  readers' convenience, we recall the definition of irreducibility and some equivalent conditions. We say that a quantum channel $\Phi$ is irreducible if all the orthogonal projections $P$ satisfying $\Phi(P \mc{B}(\mc{H}) P) \subset P \mc{B}(\mc{H}) P$ are trivial, i.e., zero or identity. The following results are adapted from \cite{wolf5quantum,zhang2023criteria}.

\begin{lem}
A QMS $\mc{P}^\dag_t = e^{t \mc{L}^\dag}$ is irreducible if and only if one of the following conditions holds:
\begin{itemize}
    \item $\mc{P}^\dag_t$ (as a quantum channel) is irreducible for some $t_0 > 0$. 
    \item There exists a unique full-rank invariant state $\si$, i.e., $\mc{L}^\dag(\si) = 0$. 
    \item The $\C$-algebra generated by the jump operators $\{L_j\}_{j \in \mc{J}}$ and $K = - \frac{1}{2}\sum_j L_j^\dag L_j - i G$ gives the whole algebra $\bh$. 
     \item The operators $\{L_j\}_{j \in \mc{J}}$ and $K$ have no trivial common invariant subspace. 
\end{itemize}
\end{lem}
\begin{lem}[{\cite[Theorem 7.2]{wolf5quantum}}]\label{lem:converg}
If the QMS $\mc{P}^\dag_t$ admits a full-rank invariant state, then 
\begin{equation*}
    \{G, L_j, L_j^\dag\}' = \ker(\mc{L})\,,
\end{equation*}
where the commutant $\{G, L_j, L_j^\dag\}'$ is defined by all the operators commuting with $L_j$, $L_j^\dag$ and $G$. It follows that in this case, the irreducibility is also equivalent to 
\begin{equation}  \label{eq:irredu}
    \{G, L_j, L_j^\dag\}' = \{z \mi\,;\ z \in \CC\}\,.
\end{equation}
\end{lem}

We next discuss the quantum detailed balance condition (DBC), which provides a sufficient criterion to guarantee steady states of a Lindbladian.
For a given full-rank state 
$\si \in \mc{D}_+(\mc{H})$, we define the modular operator: 
\begin{equation*} 
    \Delta_{\si}(X) = \si X \si^{-1}:\ \mc{B}(\mc{H}) \to \mc{B}(\mc{H})\,,
\end{equation*}
and the weighting operator:
\begin{align*} 
    \Gamma_\si X = \si^{\frac{1}{2}}X\si^{\frac{1}{2}}:\ \mc{B}(\mc{H}) \to \mc{B}(\mc{H})\,.
\end{align*} 
We also let $L_{\si}(X) = \si X$ and $R_{\si}(X) = X \si$ be the left and right multiplication operators, respectively. 
Then, for any $f: (0,\infty) \to (0,\infty)$ satisfying $f(1) = 1$ and $\si \in \mc{D}_+(\mc{H})$, we define the operator: 
\begin{align} \label{def:operator_kernel}
    J_\si^f := R_\si f(\Delta_\si): \bh \to \bh\,,
\end{align} 
and the associated inner product:
\begin{align} \label{eq:general_inner}
    \l X, Y\r_{\si,f} := \l X, J_\si^f (Y)\r\,.
\end{align}
In particular, for $f = x^{1-s}$ with $s \in \R$, the above inner product gives 
\begin{align} \label{def:s-inner}
    \l X, Y\r_{\si,s} := \tr(\si^s X^\dag \si^{1-s} Y)\,, \q \forall X, Y \in \bh\,,
\end{align}
where $\l \dd, \dd\r_{\si,1}$ and $\l \dd, \dd \r_{\si,1/2}$ are the Gelfand-Naimark-Segal (GNS) and Kubo-Martin-Schwinger (KMS) inner products, respectively. 

\begin{defn} \label{def:sidbc}
A QMS $\mc{P}_t= e^{t \mc{L}}$ satisfies the $J_\si^f$-DBC for some $\si \in \mc{D}_+(\mc{H})$ if the Lindbladian $\mc{L}$ is self-adjoint with respect to the inner product $\l\dd, \dd\r_{\si,f}$, equivalently, 
\begin{align*}
    J_\si^f  \mc{L} = \mc{L}^\dag  J_\si^f\,.
\end{align*}
In the cases of $f = 1$ and $x^{1/2}$, it is called $\si$-{\rm GNS DBC} and $\si$-{\rm KMS DBC}, respectively.
\end{defn}
By the above definition, we find that if $\mc{P}_t$ satisfies the $J_\si^f$-DBC, there holds
\begin{equation*}
  0 = \l X, \mc{L}(\mi)\r_{\si,f} =  \l \mc{L}(X), J_\si^f(\mi)\r = \l X, \mc{L}^\dag(\si)\r\,,\q \forall X \in \bh\,,
\end{equation*}
which gives $\mc{L}^\dag(\si) = 0$, namely, $\si$ is an invariant state of $\mc{P}_t^\dag$. The following lemma relates different concepts of detailed balance conditions; see \cite[Lemma 2.5 and Theorem 2.9]{CarlenMaas2017}. 

\begin{lem} \label{lem:self_adjoint}
Let $\si \in \mc{D}_+(\mc{H})$ be a full-rank quantum state and $\mc{P}_t= e^{t \mc{L}}$ be a QMS. Then, 
\begin{itemize}
    \item If $\mc{P}_t$ satisfies the $\si$-{\rm GNS DBC}, then it satisfies the $J_\si^f$-DBC for any $f$ and the generator $\mc{L}$ commutes with the modular operator $\Delta_\si$. 
    \item If $\mc{P}_t$ satisfies the $J_\si^f$-DBC for $f = x^{1-s}$, $s \in [0,1]\backslash \{\frac{1}{2}\}$, then it also satisfies $\si$-{\rm GNS DBC}. 
\end{itemize}
\end{lem}
The above lemma means that the quantum DBC for the inner products $\l \dd, \dd \r_{\si,s}$ with $s \in [0,1]\backslash \{\frac{1}{2}\}$ are all equivalent, and they are stronger notions than $\si$-{\rm KMS DBC} (i.e., $s = \frac{1}{2}$). In fact, one can show that the class of QMS with $\si$-{\rm KMS DBC} is strictly larger than the class of QMS satisfying $\si$-GNS DBC \cite[Appendix B]{CarlenMaas2017}. These properties underscore the special roles played by the KMS and GNS detailed balance when analyzing Lindblad dynamics. 

\subsection{Davies generator and GNS-detailed balance}\label{sec:daveisgns}

Let $H$ be a quantum Hamiltonian on the Hilbert space $\mc{H}$ with the eigendecomposition:
\begin{align}  \label{eq:hamihh}
    H = \sum_{i} \lad_ i P_i\,,
\end{align}
where $P_i$ is the orthogonal projector to the eigenspace associated with the energy $\lad_i$. Given an inverse temperature $\beta > 0$, the corresponding Gibbs state $\si_\beta$ is defined by 
\begin{equation} \label{eq:gibbstate}
 \si_\beta : = e^{-\beta H}/\mc{Z}_\beta\,,
\end{equation} 
with $\mc{Z}_\beta = \tr (e^{-\beta H})$ being the normalization constant called the partition function. It is easy to see that any full-rank quantum state can be written as a Gibbs state $\si = e^{-h}$ with $h = - \log(\si)$. 

Recall that our main aim is to design an efficient quantum Gibbs sampler via QMS. An important class of Lindbladians for this purpose are Davies generators, which describe the weak coupling limit of a system coupled to a large thermal bath \cite{Davies1976,davies1979generators}. It has natural applications in thermal state preparations but with inherent difficulties from the energy-time uncertainty principle \cite{MozgunovLidar2020,RallWangWocjan2023,ChenKastoryanoBrandaoEtAl2023}; see \cref{rem:ineffi} below. 
We next review the canonical form of Davies generators and show that they essentially characterize the Lindbladians with GNS-DBC \cite{kossakowski1977quantum}. 

For the Hamiltonian \eqref{eq:hamihh}, we define the set of Bohr frequencies by 
\begin{equation} \label{def:Bohr}
    B_H = \left\{\nu = \lad_i - \lad_j\,;\ \lad_i,\lad_j \in {\rm Spec}(H) \right\}\,,
\end{equation}
which is a sequence of real numbers symmetric with respect to $0$. Here, ${\rm Spec}(H)$ denotes the spectral set of $H$. Then, for any bounded operator $A \in \mc{B}(\mc{H})$, one can write 
\begin{align} \label{eq:jumpdecom}
    A = \sum_{\lad_i,\lad_j \in {\rm Spec}(H)} P_i A P_j = \sum_{\nu \in B_H} A_\nu\,, 
\end{align}
where 
\begin{equation}\label{eqn:A_v}
    A_\nu := \sum_{\lad_i - \lad_j = \nu}  P_i A P_j\,, \quad \text{with} \quad (A_{\nu})^\dag = (A^\dag)_{-\nu}\,,
\end{equation}
is an eigenstate of the modular
operator $\Delta_{\sigma_\beta}$ (see \cref{eq:rela2} below).
Such a decomposition \eqref{eq:jumpdecom} naturally relates to the Heisenberg evolution of $A$:
\begin{align} \label{eq:Heisenbergevo}
    A(t): =  e^{iH t} A e^{-iH t} & = \sum_{\nu\in B_H} A_{\nu}e^{i\nu t}\,, \quad \text{equivalently}, \quad [H,A_{\nu}] = \nu A_{\nu}\,.
\end{align}
\re{For our later use, we introduce a weighted integral $A_f$ of the Heisenberg evolution $A(t)$ with respect to an $L^1$-integrable function $f$: for $A \in \bh$, 
\begin{equation} \label{defaf}
    A_{f}:= \int_{-\infty}^{\infty} A(t) f(t)\ud t
 = \sum_{\nu \in B_H} A_{\nu} \widehat{f}(\nu)\,,
\end{equation}
where $\widehat{f}(\nu)$ is the inverse Fourier transform of $f(t)$:
\begin{equation*}
    \widehat{f}(\nu) = \int_{-\infty}^\infty f(t) e^{i \nu t} \ud t\,.
\end{equation*}
Moreover, by a direct computation, we have 
\begin{equation} \label{eq:fcalculus}
     A_{f} = \widehat{f}(-\beta^{-1}\log \Delta_{\si_\beta}) A\,,
\end{equation}
where the operator $\widehat{f}(-\beta^{-1}\log \Delta_{\si_\beta})$ is defined by functional calculus based on the modular operator $\Delta_\si$, noting that the Fourier transform of an $L^1$ function is bounded and uniformly continuous.}




The Davies Lindbladian is generally of the form: 
\begin{align} \label{eq:davies}
    \mc{L}_\beta(X) := i [H, X] + \sum_{a \in \mc{A}} \sum_{\nu \in B_H} \mc{L}_{a,\nu}(X)\,, \q X \in \mc{B}(\mc{H})\,,
\end{align}
where the dissipative generators are given by  
\begin{align} \label{eq:dissipative}
    \mc{L}_{a,\nu}(X)  = \gamma_a(\nu) \left( (A_{\nu}^a)^\dag X A_{\nu}^a - \frac{1}{2}\left\{(A_{\nu}^a)^\dag A_{\nu}^a, X\right\} \right)\,.
\end{align}
Here, the index $a \in \mc{A}$ sums over all the coupling operators $A^a \in \mc{B}(\mc{H})$ to the environment that satisfy $\{A^a\}_{a \in \mc{A}} = \{(A^a)^{\dag}\}_{a \in \mc{A}}$, and $\gamma_a(\cdot)$ are the Fourier transforms of 
the bath correlation functions, which are nonnegative and bounded. The jump operators $\{A_{\nu}^a\}$ associated with a coupling $A^a \in \mc{B}(\mc{H})$ are defined by \cref{eq:jumpdecom,eqn:A_v}:
\begin{equation}
    A^a = \sum_{\lad_i,\lad_j \in {\rm Spec}(H)} P_i A^a P_j = \sum_{\nu \in B_H} A_{\nu}^a,
\end{equation}
which gives the transitions from the eigenvectors of $H$ with energy $E$ to those with $E+\nu$. In addition, the following relations hold, for any $a \in \mc{A}$ and $\nu$,
    \begin{equation} \label{eq:rela1}
        \gamma_{a} (- \nu) = e^{\beta \nu} \gamma_a(\nu)\,, 
    \end{equation}
    and    
    \begin{align} \label{eq:rela2}
        \Delta_{\si_\beta} (A_{\nu}^a) = e^{-\beta \nu} A_{\nu}^a\,,
    \end{align}
    that is, $A_{\nu}^a$ is an eigenvector of $\Delta_{\si_\beta}$ with the eigenvalue $e^{-\beta \nu}$. 
    Note that the condition \eqref{eq:rela2} holds by the definition of $A_{\nu}^a$, while \re{the other one} \eqref{eq:rela1} is often referred to as KMS condition\footnote{The KMS condition should not be confused with the KMS detailed balance condition. These two terms are mathematically unrelated.} \cite{kossakowski1977quantum}. 
    \re{These two conditions,} as we shall see below, ensure the GNS-reversibility of the Lindbladian.   
    
The following canonical form for the QMS that satisfies the $\si$-{\rm GNS DBC} is due to Alicki \cite{Alicki1976}.
\begin{lem}\label{lem:qms_gns}
For a Lindbladian $\mc{L}$ satisfying $\si_\beta$-{\rm GNS DBC}, there holds 
 \begin{align} \label{eq:structure}
        \mc{L}(X) = \sum_{j \in \mc{J}} \left(e^{-\omega_j/2}L_j^\dag[X,L_j] + e^{\omega_j/2}[L_j,X]L^\dag_j\right)\,,
    \end{align}
with $\ww_j \in \R$, where \re{$\mc{J}$ is an index set with cardinality $|\mc{J}| \le N^2 -1$}, and 
$L_j \in \mc{B}(\mc{H})$ satisfies 
\begin{align} \label{eq:eigmodular}
    \Delta_{\si_\beta}(L_j) =  e^{-\omega_j}L_j\,,\quad \tr(L^\dag_j L_k) = c_j\d_{j,k}\,,\quad \tr(L_j) = 0\,,
\end{align}
for normalization constants $c_j > 0$, and for each $j$, there exists $j' \in \mc{J}$ such that 
\begin{align} \label{eq:adjoint_index}
    L_j^\dag = L_{j'}\,,\quad \ww_j = - \ww_{j'}\,.
\end{align}
\end{lem}
It is easy to see that the Davies semigroup \eqref{eq:davies} is exactly the class of QMS with GNS-DBC, up to the coherent term $i [H, \cdot]$. Indeed, we define $g_a(\nu) := e^{\beta \nu/2} \gamma_a(\nu)$ and find $g_a(\nu) = g_a(-\nu)$ by the KMS condition \eqref{eq:rela1}. Then, letting $L_{a,\nu} := g_a(\nu)^{1/2} A_{\nu}^a$, it follows from \eqref{eq:davies} and \eqref{eq:dissipative} that 
\begin{align*}
     \mc{L}_\beta(X) & = i [H, X] + \sum_{a \in \mc{A}} \sum_{\nu \in B_H} e^{-\beta \nu/2} \left( L_{a,\nu}^\dag X L_{a,\nu} - \frac{1}{2}\left\{L_{a,\nu}^\dag L_{a,\nu}, X\right\} \right) \\
     & =  i [H, X] + \frac{1}{2}\sum_{a \in \mc{A}} \sum_{\nu \in B_H} e^{-\beta \nu/2} L_{a,\nu}^\dag [X, L_{a,\nu}] + e^{\beta \nu/2} [L_{a,\nu}, X] L_{a,\nu}^\dag\,,
\end{align*}
where the dissipative part exactly satisfies the conditions in \cref{lem:qms_gns} by re-indexing $j = (a,\nu)$.

\begin{rem} \label{rem:ineffi}
\re{As reviewed in the introduction, the  Davies generator $\mc{L}_\beta$ in \eqref{eq:davies} has been adopted for the thermal state preparation  \cite{RallWangWocjan2023,ChenKastoryanoBrandaoEtAl2023,ChenKastoryanoBrandaoEtAl2023}. 
However, implementing it with a small error
requires accurately resolving 
all the Bohr frequencies $\nu$, while the minimal gap between two Bohr frequencies $\nu, \nu'$ could be exponentially small in the worst case.}
By the energy-time uncertainty principle, this
means an impractically long Hamiltonian simulation time and is a key obstacle in directly leveraging the Davies semigroup as an \emph{efficient} quantum Gibbs sampler \cite{ChenKastoryanoBrandaoEtAl2023}. 

    It is also worth mentioning that the sum over Bohr frequencies \eqref{eq:davies} is derived from a secular approximation \cite{breuer2002theory}, which may be regarded as theoretical evidence that the GNS detailed balance is an idealized construction and thus is difficult to \re{be efficiently implemented with desired accuracy.} 
\end{rem}

\subsection{KMS-detailed balanced generators} \label{sec:kmsdbc}

Recalling that the KMS DBC is a weaker property compared to the GNS one, but can still guarantee the Gibbs state as a fixed point of the dynamics, 
one may expect that the KMS-detailed balanced QMS can provide a more efficient class of Gibbs state preparation algorithms. \re{The first complete characterization of the $\sigma$-KMS detailed balanced Lindbladian was given in \cite[Theorems 7.2 and 7.3]{fagnola2007generators}, where some equivalent conditions on the coherent term and jump operators were proposed. In this section, we revisit their arguments along the lines of the recent work \cite{AmorimCarlen2021} and 
introduce the canonical form of the QMS satisfying $\si_\beta$-KMS DBC with a self-contained proof, which are the starting points of our quantum Gibbs samplers.}



Let $\maf{H}: = \mc{B}(\bh)$ be the space of superoperators. We define the subspace $\maf{H}_S$ consisting of $\Phi \in \maf{H}$ of the form: for some $X, Y \in \mc{B}(\mc{H})$,  
\begin{equation} \label{def:Phixy}
    \Phi(A) = XA + A Y\,,
\end{equation}
and denote by $\maf{H}_S^\perp$ its orthogonal complement \re{for the normalized Hilbert--Schmidt inner product on superoperators, which is defined as follows: for $\Phi, \Psi \in \maf{H}$, 
\begin{equation*}
    \l \Phi, \Psi\r_\maf{H} = \frac{1}{N^2} \sum_{i,j = 1}^N \l \Phi(X_{i,j}), \Phi(X_{i,j})\r\,,
\end{equation*}
with $\{X_{i,j}\}_{i,j = 1}^N$ being any orthonormal basis of $\bh$ satisfying}
\begin{equation*}
 \re {\l X_{i,j}, X_{k,l}\r =  \tr(X_{i,j}^\dag X_{k,l}) = \d_{(i,j),(k,l)}\,.}
\end{equation*}
The following useful lemma characterizes the freedom of $X, Y$ in \eqref{def:Phixy}\footnote{This result is from \cite[Lemma 3.10 and Remark 3.11]{AmorimCarlen2021}, which include some typos in the arguments. We provide a short proof here for the reader's convenience.}.

\begin{lem}\label{lem:unixy}
    Let $\Phi$ be a superoperator with the representation \eqref{def:Phixy}. If some $X', Y' \in \mc{B}(\mc{H})$ gives the same $\Phi$ via \eqref{def:Phixy}, then $X' = X + \eta {\bf 1}$ and $Y' = Y - \eta {\bf 1}$ for some $\eta \in \CC$. 
\end{lem}
\begin{proof}
  Let $F_\alpha$ with $\alpha = (i,j)$ and $1 \le i,j \le N$ be a basis of $\mc{B}(\mc{H})$ satisfying 
$\l F_\alpha, F_\beta\r/N = \d_{\alpha\beta}$,
   $F_{(1,1)} = \mi$, and $F^\dag_{(i,j)} = F_{(j,i)}$, and let $E_{i,j}$ be \re{another basis defined by $E_{i,j} = \ket{u_i}\bra{u_j}$}, where $\{u_j\}$ is an orthonormal basis of $\mc{H}$. Suppose that $\Phi(A) = X A + A Y$ for some $X, Y$. \re{Note that $\Phi_{\alpha,\beta}(X) = F_\alpha^\dag A F_\beta$ is an orthonormal basis of $\maf{H}$ by computing
   \begin{equation*}
       \l \Phi_{\alpha,\beta}, \Phi_{\gamma,\eta} \r_{\maf{H}} = \frac{1}{N}\l F_\gamma, F_\alpha \r \frac{1}{N} \l F_\beta, F_\eta \r = \d_{(\alpha,\beta),(\gamma,\eta)}\,.
   \end{equation*}
   We then} compute the coefficients for the \re{orthogonal} expansion of $\Phi(A) = \sum_{\alpha,\beta} (C_\Phi)_{\alpha,\beta} F_\alpha^\dag A F_\beta$: 
\begin{align*}
    (C_\Phi)_{\alpha,\beta} & = \re{\frac{1}{N^2}} \sum_{i,j = 1}^N \tr \left[ (F_\alpha^\dag E_{i,j} F_\beta)^\dag \Phi(E_{i,j}) \right] \\
    & = \frac{1}{N^2} \left( \tr [F_\beta^\dag] \tr [F_\alpha X] + \tr [F_\alpha] \tr [F_\beta^\dag Y] \right) = \frac{1}{N} \left( \d_{\beta,(1,1)} \tr [F_\alpha X] + \d_{\alpha,(1,1)} \tr [F_\beta^\dag Y] \right)\,,
\end{align*}
which is zero if $\alpha, \beta \neq (1,1)$.  It follows that 
\begin{align*}
    \Phi(A) = (C_\Phi)_{(1,1),(1,1)} A + \sum_{\alpha \neq (1,1)} (C_\Phi)_{\alpha,(1,1)} F_\alpha^\dag A + \sum_{\beta \neq (1,1)} (C_\Phi)_{(1,1),\beta} A F_\beta\,.
\end{align*}
Therefore, any $X', Y'$ such that $\Phi(A) = X' A + A Y'$ satisfy 
\begin{align*}
    X' = \sum_{\alpha \neq (1,1)} (C_\Phi)_{\alpha,(1,1)} F_\alpha^\dag + a \mi\,,\q Y' = \sum_{\beta \neq (1,1)} (C_\Phi)_{(1,1),\beta} F_\beta + b \mi\,,
\end{align*}
for some $a,b \in \C$ with $a + b = (C_\Phi)_{(1,1),(1,1)}$. The proof is complete. 
\end{proof}


We next discuss the structure of QMS satisfying $\si_\beta$-KMS DBC for some Gibbs state $\si_\beta$ \eqref{eq:gibbstate}.

\begin{lem} \label{lem:kmsform}
A Lindbladian $\mc{L}$ satisfies $\si_\beta$-KMS DBC if and only if $\mc{L}$ has the form: 
\begin{align} \label{eq:lgksform2}
    \mc{L} (X) = \Psi(X) + \Phi(X)\,, \q 
\end{align}  
with the CP operator $\Psi(\dd)$ admitting the Kraus representation \eqref{eq:kraus} and the operator 
\begin{equation} \label{eq:phii}
    \Phi(X) := K^\dag X + X K\,,\quad \text{for some $K \in \bh$}\,,     
\end{equation}
and both $\Psi$ and $\Phi$ are self-adjoint with respect to the KMS inner product. In this case, there exist the jump operators $\{L_j\}_{j \in \mc{J}}$ and the operator $K$ in \cref{eq:phii} satisfying
\begin{equation} \label{eq:adjointlj}
    \Delta_{\si_\beta}^{-1/2} L_j = L_j^\dag\,,
\end{equation}
and 
\begin{equation} \label{eq:constK}
 \Delta_{\si_{\beta}}^{-1/2} K =  K^\dag \,.
\end{equation}
 \end{lem}

 \begin{proof}

It suffices to prove the \emph{only if} part. Recalling the structure of a Lindbladian in \cref{lem:lgksform}, without loss of generality, we assume $\tr(L_j) = 0$ by replacing $L_j$ with $L_j - \tr(L_j) {\bf 1}$ and  $K$ with $K + \sum_j \tr(L_j^\dag) L_j - \frac{1}{2}|\tr(L_j)|^2$. Then, by \cite[Lemma 3.12]{AmorimCarlen2021}, there holds $\Psi(\dd) \in \maf{H}_S^\perp$. According to \cite[Lemma 3.13]{AmorimCarlen2021}, the subspaces $\maf{H}_S$ and $\maf{H}_S^\perp$ are invariant under the adjoint with respect to the KMS inner product. By $\Phi\in\maf{H}_S$ and $\Psi\in \maf{H}_S^\perp$, it holds that adjoints $\Phi^\dagger_{{\rm KMS}} \in \maf{H}_S$ and $\Psi^\dagger_{{\rm KMS}} \in \maf{H}_S$, where $\Phi^\dagger_{{\rm KMS}}$ and $\Psi^\dagger_{{\rm KMS}}$ are adjoints of $\Phi$ and $\Psi$ for the KMS inner product. Thus, the self-adjointness $\Psi^\dagger_{{\rm KMS}}+\Phi^\dagger_{{\rm KMS}} = \Psi + \Phi$ implies  
$\Psi^\dagger_{{\rm KMS}}=\Psi$ and $\Phi^\dagger_{{\rm KMS}}=\Phi$. 

    
    
    Next, since $\Psi$ is a KMS-detailed balanced CP map, \eqref{eq:adjointlj} is implied by the structure result~\cite[Theorem 4.1]{AmorimCarlen2021}. To show \eqref{eq:constK}, by the invariance of $\Phi = K^\dag X + X K$ for adding a pure imaginary $i c \mi$ ($c \in \R$) to $K$, without loss of generality, we can assume $\tr(K) \in \R$, which further implies $\tr(\Delta_{\si_\beta}^{-1/2}K) \in \R$. It follows from $\si_\beta$-KMS DBC of $\Phi$ that $ \Gamma_{\si_\beta} \Phi = \Phi^\dag   \Gamma_{\si_\beta}$, equivalently, 
    \begin{equation*}
        K^\dag X + X K  = (\Delta_{\si_{\beta}}^{-1/2} K)   X   +   X (\Delta_{\si_{\beta}}^{1/2}  K^\dag)\,.
    \end{equation*}
    Then, by \cref{lem:unixy}, we derive $K^\dag = \Delta_{\si_{\beta}}^{-1/2} K + \eta {\bf 1}$ for some $\eta \in \C$, where $\eta$ must be zero, thanks to $\tr(K^\dag) = \tr(K)=\tr(\Delta_{\si_{\beta}}^{-1/2} K) \in \R$. The proof is complete. 
     \end{proof}



We proceed to derive an explicit formula for the operator $K$. Recalling the decomposition $K = V - iG$ with $V =  - \frac{1}{2} \sum_{j \in \mc{J}} L_j^\dag L_j$ (\rre{see after \cref{lem:lgksform}}), it suffices to find an expression for $G$. To do so, we reformulate the constraint \eqref{eq:constK} above as a Lyapunov equation:  
\begin{align*}
    G \si_\beta^{1/2} + \si_\beta^{1/2} G = i (\si_\beta^{1/2} V - V \si_\beta^{1/2})\,.
\end{align*}
It can be uniquely solved as 
\begin{equation} \label{eq:hamil}
        G = i \int_0^\infty e^{- t \si_\beta^{1/2}} (\si_\beta^{1/2} V - V \si_\beta^{1/2}) e^{- t \si_\beta^{1/2}} \ud t \,.
\end{equation}
To simplify the formula, we note that for any $\lad,\mu > 0$, 
\begin{equation*}
    \int_0^\infty e^{-t \lad^{1/2}}  e^{- t \mu^{1/2}} \ud t  = \frac{1}{\lad^{1/2} + \mu^{1/2}}\,,
\end{equation*}
and $\tanh (\log(x^{1/4})) = \frac{x^{1/2} - 1}{x^{1/2} + 1}$. Then, by functional calculus and \eqref{eq:hamil}, \re{it holds that}
\begin{equation} \label{eq: hamil2}
     G = i  \frac{L_{\si_\beta}^{1/2} - R_{\si_\beta}^{1/2}}{L_{\si_\beta}^{1/2} + R_{\si_\beta}^{1/2}}(V) = i  \frac{\Delta_{\si_\beta}^{1/2} - I}{\Delta_{\si_\beta}^{1/2} + I}(V) = i \tanh \circ \log (\Delta_{\si_\beta}^{1/4})(V)\,.
\end{equation}
 We summarize the above discussion in the following proposition.
 
 
\begin{thm}\label{prop:KMS_DBC}
    A Lindbladian $\mc{L}$ satisfies $\si_{\beta}$-KMS DBC if and only if there exist linear operators $L_j, G \in \bh$ such that  
    \begin{align} \label{eq:kmsformcan}
     \mc{L} (X) = i [G, X] + \sum_{j \in \mc{J}} \left( L_j^\dag X L_j - \frac{1}{2}\left\{L_j^\dag L_j, X  \right\} \right)\,,
\end{align}
with $L_j$ \re{satisfying 
\begin{equation*} 
    \Delta_{\si_\beta}^{-1/2} L_j = L_j^\dag\,,
\end{equation*}
and} $G$ being self-adjoint and given by
\begin{equation}\label{eqn:G_formula_KMS}
     G :=  - i \tanh \circ \log (\Delta_{\si_\beta}^{1/4})\left( \frac{1}{2} \sum_{j \in \mc{J}}  L_j^\dag L_j\right)\,.
\end{equation}
\end{thm}
    
We have shown in \cref{lem:self_adjoint} that the Lindbladians with $\si_\beta$-GNS DBC is a subclass of those with $\si_\beta$-KMS DBC, but this property cannot be easily seen from the corresponding structural results (cf.\, \cref{lem:qms_gns} and \cref{prop:KMS_DBC}), noting that an eigenvector $L$ of the operator $\Delta_{\si_{\beta}}$ is generally not a solution to \eqref{eq:adjointlj}. To fill this gap, we next show that the canonical form \eqref{eq:structure} of a QMS with GNS DBC can be  indeed reformulated as the one \eqref{eq:kmsformcan} for KMS DBC. 

\begin{cor} \label{cor:gns_kms}
    Let $\mc{L}$ be a Lindbladian with $\si_\beta$-{\rm GNS DBC} of the form:
    \begin{equation} \label{eq:auxlind}
        \mc{L}(X) = \sum_{j \in \mc{J}} \mc{L}_j(X)\quad \text{with}\quad \mc{L}_j(X) = 2 e^{-\ww_j/2} \left( L_j^\dag X L_j - \frac{1}{2}\left\{L_j^\dag L_j, X  \right\} \right)\,,
    \end{equation}
    where the jumps $\{L_j\}_{j \in \mc{J}}$ satisfy the conditions in \cref{lem:qms_gns}. Then, we can reformulate it in the form of a KMS detailed balanced Lindbladian \eqref{eq:kmsformcan}:
    \begin{align*}
        \mc{L}(X) = \sum_{j \in \mc{J},\, L_j^\dag = L_j} \mc{L}_j(X) + \frac{1}{2} \sum_{j \in \mc{J},\, L_j^\dag \neq L_j} \w{\mc{L}}_j(X) \,,  
    \end{align*}
 with  
    \begin{equation*}
        \w{\mc{L}}_j(X) =  \w{L}_{j,1}^\dag X \w{L}_{j,1}  + \w{L}_{j,2}^\dag X \w{L}_{j,2} - \frac{1}{2}\left\{\w{L}_{j,1}^\dag \w{L}_{j,1}, X  \right\} - \frac{1}{2}\left\{\w{L}_{j,2}^\dag \w{L}_{j,2}, X  \right\}\,,
    \end{equation*}
where 
\begin{align*}
    \w{L}_{j,1} :=  e^{-\ww_j/4} L_j + e^{\ww_j/4} L_j^\dag\,,\q  \w{L}_{j,2} :=  i(-  e^{-\ww_j/4} L_j + e^{\ww_j/4} L_j^\dag) 
\end{align*} 
satisfy the constraint \eqref{eq:adjointlj}. 
\end{cor}

\begin{proof}
    If $L_j$ is self-adjoint,  there hold $\ww_j = 0$ and
    $\Delta_{\si_\beta}^{-1/2} L_j = L_j = L_j^\dag$. By the formula \eqref{eq: hamil2} of $G$, the associated Hamiltonian is given by
    \begin{equation*}
        G_j :=  -\frac{i}{2} \frac{\Delta_{\si_\beta}^{1/2} - I}{\Delta_{\si_\beta}^{1/2} + I}(L_j^2) = 0\,.
    \end{equation*}
    Thus, in this case, $\mc{L}_j$ in \eqref{eq:auxlind} satisfies the canonical form \eqref{eq:kmsformcan} in  \cref{prop:KMS_DBC}. 

    We next consider the case where $L_j$ is not self-adjoint. Let $j'$ be the adjoint index for $j$ specified in \eqref{eq:adjoint_index}. It follows that 
    \begin{equation*}
        \mc{L}_j(X) + \mc{L}_{j'}(X) = 2 e^{-\ww_j/2} \left( L_j^\dag X L_j - \frac{1}{2}\left\{L_j^\dag L_j, X  \right\} \right) + 2 e^{\ww_j/2} \left( L_j X L_j^\dag - \frac{1}{2}\left\{L_j L_j^\dag, X  \right\} \right)\,.
    \end{equation*}
We consider the equation $\Delta_{\si_\beta}^{-1/2} \w{L} = \w{L}^\dag$ with ansatz $\w{L} = a L_j + b L_j^\dag$, $a,b \in \mathbb{C}$: 
\begin{equation*}
    \Delta_{\si_\beta}^{-1/2} ( a L_j + b L_j^\dag) = a e^{\ww_j/2} L_j + b e^{-\ww_j/2} L_j^\dag =  \overline{a} L_j^\dag + \overline{b} L_j\,.
\end{equation*}
It is easy to check that the coefficients $(a,b)$ being real linear combinations of vectors  $(e^{-\ww_j/2},1)$ and $(-e^{-\ww_j/2}i,i)$ satisfy $a = \overline{b} e^{-\ww_j/2}$, and  the corresponding $\w{\mc{L}}$ solves $\Delta_{\si_\beta}^{-1/2} \w{L} = \w{L}^\dag$. We define 
\begin{align} \label{auxeqq:Lj}
    \w{L}_1 =  e^{-\ww_j/4} L_j + e^{\ww_j/4} L_j^\dag\,,\q  \w{L}_2 =  i(-  e^{-\ww_j/4} L_j + e^{\ww_j/4} L_j^\dag)\,.
\end{align}
A direct computation gives, for $X \in \mc{B}(\mc{H})$, 
\begin{equation*}
    \w{L}_1^\dag X \w{L}_1 + \w{L}_2^\dag X \w{L}_2 = 2 e^{-\ww_j/2} L_j^\dag X L_j + 2 e^{\ww_j/2}  L_j X L_j^\dag\,,
\end{equation*}
and 
\begin{align*}
    G & =  - \frac{i}{2} \frac{\Delta_{\si_\beta}^{1/2} - I}{\Delta_{\si_\beta}^{1/2} + I}\left(\w{L}_1^\dag \w{L}_1 + \w{L}^\dag_2 \w{L}_2\right) \\
    & =  - i \frac{\Delta_{\si_\beta}^{1/2} - I}{\Delta_{\si_\beta}^{1/2} + I}\left( e^{-\ww_j/2} L_j^\dag  L_j +  e^{\ww_j/2}  L_j L_j^\dag \right) = 0\,,
\end{align*}
by noting that $L_j^\dag L_j$ and $L_j L_j^\dag$ are eigenvectors of $\Delta_{\si_\beta}$ associated with eigenvalue one. Therefore, for non-self-adjoint $L_j$, the Lindbladian $\mc{L}_j + \mc{L}_{j'}$ also matches with the form \eqref{eq:kmsformcan}. The proof is complete by the linearity of Lindbladians. 
\end{proof}



\section{A family of efficient quantum Gibbs samplers}\label{sec:QGS}

In this section, we present a general framework for designing \re{efficiently simulable} quantum Gibbs samplers via Lindblad dynamics satisfying $\sigma_\beta$-KMS DBC, \re{based on the representation \cref{prop:KMS_DBC}}. 


\subsection{Quantum Gibbs samplers via KMS-detailed balanced Lindbladian} \label{sec:gibbsampler}

By \cref{prop:KMS_DBC}, the class of KMS detailed balanced Lindbladians can be parameterized by
\begin{itemize}
    \item a set of jump operators for the Lindbladian $\{L_j\}_{j \in \mc{J}}$ satisfying \eqref{eq:adjointlj}: $\Delta_{\si_\beta}^{-1/2} L_j = L_j^\dagger$;
    

    
\item a coherent term $G$ defined as in \eqref{eqn:G_formula_KMS} via $\{L_j\}_{j \in \mc{J}}$. 


\end{itemize}

\re{We first reformulate the condition \eqref{eq:adjointlj} into a more convenient form.} 
Note that \eqref{eq:adjointlj} is equivalent to $\Delta_{\si_\beta}^{-1/4} L_j = \Delta_{\si_\beta}^{1/4} L_j^\dagger$, namely, $\Delta_{\si_\beta}^{-1/4} L_j$ is self-adjoint. Then the admissible set of jump operators can be given by
\begin{equation*} 
\{L \in \mc{B}(\mc{H})\,;\ L = \Delta_{\si_\beta}^{1/4} \w{A}\ \  \text{with}\ \  \w{A} = \w{A}^\dagger\}\,.    
\end{equation*}
From the eigendecomposition of $H$ in \eqref{eq:hamihh}, we have
\begin{equation}\label{eqn:V_first}
  L =  \Delta_{\si_\beta}^{1/4} \w{A} = \sum_{i,j} e^{-\beta (\lad_i - \lad_j)/4} P_i \w{A} P_j = \sum_{\nu\in B_H} e^{-\beta \nu/4}  \w{A}_{\nu}\,,
\end{equation}  
where $\w{A}_{\nu}$ is defined by \eqref{eqn:A_v} for some self-adjoint $\w{A}$. 

Suppose that we are given a set of self-adjoint coupling operators $\{A^a\}_{a \in \mc{A}}$ \re{as predefined jump proposals (coupling operators), where $\mc{A}$ is a discrete index set}. For each $a \in \mc{A}$, we choose a weighting function $q^a(\nu):\mathbb{R}\rightarrow\mathbb{C}$ such that
\begin{equation}\label{eqn:q_sym}
q^{a}(-\nu)=\overline{q^a(\nu)}, \quad \widehat{f^a}(\nu):=q^a(\nu)e^{-\beta\nu/4}\in L^1(\RR)\,.
\end{equation}
\re{Here $q^a$ is chosen to be a symmetric function so that $L_a$ is an admissible jump (i.e., satisfies the condition \eqref{eq:adjointlj}), see~\cref{eqn:A_tilde_hermitian}.}

We then construct the jump operator \re{associated with the pair $(A^a,q^a)$:} 
\begin{equation}\label{eqn:L_a_formula}
L_a := \sum_{\nu\in B_H}q^a(\nu)e^{-\beta\nu/4}A^a_\nu = 
\sum_{\nu\in B_H} \widehat{f^a}(\nu) A^a_\nu = \int^\infty_{-\infty} f^a(t) A^a(t) \mathrm{d}t\,,
\end{equation}
\re{as a weighted integral of $A^a$ (see \cref{defaf}),} where the Fourier transform of $\widehat{f^a}$:
\begin{equation}\label{eqn:f_a}
f^a(t)=\frac{1}{2\pi}\int^\infty_{-\infty} q^a(\nu)e^{-\beta \nu/4}e^{-it \nu} \ud \nu 
\end{equation}
gives the filtering function in the time domain. \re{To verify that $L_a$ satisfies the condition \eqref{eq:adjointlj}, it suffices to note that $\w{A}^a_\nu := q^a(\nu)A^a_\nu$ corresponds to a self-adjoint operator defined by $\wt{A}^a := \sum_{\nu\in B_H} \w{A}^a_{\nu}$:} 
\begin{equation}\label{eqn:A_tilde_hermitian}
(\wt{A}^a)^{\dag}=\sum_{\nu\in B_H} (\wt{A}^a)^{\dag}_{\nu}=\sum_{\nu\in B_H} \overline{q^a(\nu)} (A^a_\nu)^{\dag}=\sum_{\nu\in B_H} q^a(-\nu) A^a_{-\nu} = \wt{A}^a\,.
\end{equation}  
We see from \eqref{eqn:L_a_formula} that
each $L_a$ is a linear combination of the Heisenberg evolution $A^a(t) =  e^{iH t} A^a e^{-iH t}$, and \re{by \eqref{eq:fcalculus}, $L_a$ has an alternative representation:} 
\begin{align*}
  \re{ L_a = \widehat{f^a}(-\beta^{-1}\log \Delta_{\si_\beta}) A^a\,.}
\end{align*}
\re{Here, the filtering function $\widehat{f^a}(\nu)$ is used to modify the transition between eigenvectors of $H$ with an energy difference of 
$\nu$, similar to the rejection step in the classical sampling, and the associated operator $\widehat{f^a}(-\beta^{-1}\log \Delta_{\si_\beta})$ maps the predefined coupling operators $A^a$ to the desired jumps $L_a$ for the $\si_\beta$-KMS detailed balance. Moreover, due to the introduction of $q^a$ in $L_a$, without loss of generality, we can let $\norm{A^a} \le 1$ by rescaling the norm.}

\re{We proceed to construct the coherent part via the formula \eqref{eqn:G_formula_KMS}: 
\begin{align*}
    G &:= -i \tanh \circ \log (\Delta_{\si_\beta}^{1/4})\left( \frac{1}{2} \sum_{a\in\mathcal{A}} L_a^\dag L_a\right)\\
    &=- \frac{i}{2}\sum_{a \in \mc{A}} \sum_{\nu \in B_H} \tanh\left(-\frac{\beta\nu}{4}\right)\left(L_a^\dag L_a\right)_\nu\,,
\end{align*}
where $\left( L_a^\dag L_a\right)_\nu$ is defined as in \eqref{eqn:A_v}. However, the function $\tanh\left(-\beta\nu/4\right)$ is not $L^1$-integrable, and hence the useful time integral formulation as the one \eqref{eqn:L_a_formula} for the jump cannot readily follow. To remedy this issue, thanks to $|\nu|\leq 2\|H\|$ for a Bohr frequency, for any smooth and compactly supported function $\kappa(\dd)$ on $\R$ such that 
\begin{equation} \label{eq:condkappa}
 \kappa(\nu)=1 \q \text{for}\q \nu \in [-2\|H\|, 2\|H\|]\,,
\end{equation}
we define the filtering function 
\begin{align} \label{eqn:gk}
    \widehat{g}(\nu) := -\frac{i}{2} \tanh\left(-\frac{\beta\nu}{4}\right) \kappa(\nu)\,,
\end{align}
and then have 
\begin{equation}\label{eqn:g}
G = \sum_{a \in \mc{A}} \sum_{\nu \in B_H} \hat{g}(\nu) \left( L_a^\dag L_a\right)_\nu\,.
\end{equation}
Since $\widehat{g}(\nu)$ is an $L^1$-integrable function, its Fourier transform is well-defined: 
\begin{equation}\label{eqn:g_hat}
g(t)=\frac{1}{2\pi}\int^\infty_{-\infty}\widehat{g}(\nu)e^{-it\nu}\, \mathrm{d}\nu\,,
\end{equation}
and 
it is direct to compute
\begin{equation}\label{eqn:G_KMS}
G=\sum_{a \in \mc{A}} \sum_{\nu \in B_H} \widehat{g}(\nu)\left(  L_a^\dag L_a\right)_\nu=\int^\infty_{-\infty}g(t) H_L(t) \ud t\,,
\end{equation}
where 
\begin{equation} \label{def:hll}
    H_L(t) := \sum_{a\in\mathcal{A}} e^{i Ht}\left(L_a^\dag L_a\right)e^{-iHt}\,.
\end{equation}
We emphasize that the choices of $q^a$ in \eqref{eqn:f_a} and $\kappa$ in \eqref{eqn:gk} are key components in our algorithm and will be discussed in detail in \cref{sec:choice_q}.}


With $G$ and $\{L_a\}_{a\in\mathcal{A}}$ be constructed in \eqref{eqn:G_KMS} and \eqref{eqn:L_a_formula}, respectively, the Lindbladian 
\begin{align} \label{eq:lindbladkms}
\mc{L} (X) = i [G,X] + \sum_{a\in\mathcal{A}} \left( L_a^\dagger X L_a - \frac{1}{2}\left\{L_a^\dagger L_a, X  \right\} \right)\,,
\end{align}
satisfies $\si_{\beta}$-KMS DBC by \cref{prop:KMS_DBC}. Then the corresponding Lindblad master equation reads
\begin{equation}\label{eqn:Lindblad_master_equation}
\partial_t\rho=\mc{L}^{\dag} (\rho)= -i [G, \rho] + \sum_{a\in\mathcal{A}} \left( L_a \rho L^\dagger_a - \frac{1}{2}\left\{L_a^\dagger L_a, \rho  \right\} \right)\,.
\end{equation} 

We have discussed the case where coupling operators $\{A^a\}_{a \in \mc{A}}$ are self-adjoint. In fact, one can generally consider a set of couplings such that $\{A^a\}_{a\in \mathcal{A}}=\{(A^a)^\dagger\}_{a\in \mathcal{A}}$, and construct the corresponding jump operators $\{L_a\}_{a \in \mc{A}}$, the coherent term $G$, and the Lindbladian $\mc{L}$ as in 
\cref{eqn:L_a_formula,eqn:G_KMS,eq:lindbladkms}, respectively. It is easy to see that the Lindblad dynamics defined in this way still satisfies the KMS detailed balance. Indeed, let $L_{a}$ and $L_{a,{\rm adj}}$ be the jumps associated with some $A^a$ and $(A^a)^\dag$ by \eqref{eqn:L_a_formula} (without loss of generality, $A^a \neq (A^a)^\dag$). We then define self-adjoint operators
\begin{equation*}
    A^a_+ = \frac{A^a + (A^a)^\dag}{\sqrt{2}} \,,\q A^a_- = \frac{A^a - (A^a)^\dag}{\sqrt{2}i}\,, 
\end{equation*}
such that $\sqrt{2} A^a = A_+^a + i A_-^a$ and denote by $L_{a,+}$ and $L_{a,-}$ the associated jumps \eqref{eqn:L_a_formula}. A direct computation by using the time-domain representation \eqref{eqn:L_a_formula} gives 
\begin{align*}
    L_a \rho L_a^\dag + L_{a,{\rm adj}} \rho L_{a,{\rm adj}}^\dag & = \iint_{\R^2} f^a(t)\overline{f^a(t')} \left( A^a(t) \rho (A^a(t'))^\dag  +  (A^a(t))^\dag \rho A^a(t')   \right) \\
    & = \iint_{\R^2} f^a(t)\overline{f^a(t')} \left(A_+^a(t) \rho A_+^a(t') +  A_-^a(t)  \rho A_-^a(t') \right) \\
    & =  L_{a,+} \rho L_{a,+}^\dag + L_{a,-} \rho L_{a,-}^\dag \,,
\end{align*}
thanks to 
\begin{align*}
     & A^a(t) \rho (A^a(t'))^\dag  +  (A^a(t))^\dag \rho A^a(t') \\
    = & \frac{1}{2} \left(A_+^a(t) + i A_-^a(t)\right) \rho \left(A_+^a(t') - i A_-^a(t')\right)  + \frac{1}{2} \left(A_+^a(t) - i A_-^a(t)\right) \rho \left(A_+^a(t') + i A_-^a(t')\right) \\
    = &  A_+^a(t) \rho A_+^a(t') +  A_-^a(t)  \rho A_-^a(t')\,. 
\end{align*}
Here $f^a$ is defined via \eqref{eqn:f_a} with $q^a(\nu)$ satisfies \eqref{eqn:q_sym}. Similarly, one can check 
\begin{align*}
   \{L^\dagger_{a}L_{a}, \rho\} + \{L^\dagger_{a,{\rm adj}}L_{a,{\rm adj}}, \rho\}=\{L^\dagger_{a, +}L_{a, +}, \rho\} + \{L^\dagger_{a, -}L_{a, -}, \rho\}\,,
\end{align*}
and hereby our claim holds.  


\subsection{Choice of weighting function \texorpdfstring{$q(\nu)$ and $\kappa(\nu)$}{Lg}}\label{sec:choice_q}

In order to efficiently implement the jump operators $\{L_a\}_{a\in\mathcal{A}}$ in \eqref{eqn:L_a_formula} and the coherent term $G$ in \eqref{eqn:G_KMS}, we need to approximate the involved \re{time--domain integrals} by a numerical quadrature in a finite region. \re{Specifically, we fix a truncation time $T > 0$ and approximate the integral $\int^\infty_{-\infty}$ by $\int_{-T}^{T}$, which is then further approximated by numerical quadrature. To ensure efficiency, we need a small $T$, meaning that $f^a(t),g(t)$ should be smooth functions that decay rapidly as $\abs{t}\to \infty$. To this end, we shall assume that $q^a$ and $\kappa$ are compactly supported Gevrey functions. We first recall the definition of Gevrey functions below \cite{Adwan_2017}.}


\begin{defn}[Gevrey function] \label{def:gevrey}
Let $\Omega\subseteq \RR^d$ be a domain. A complex-valued $C^\infty$ function $h: \Omega\to \CC$ is a \emph{Gevrey function} of order $s\ge 0$, if there exist constants $C_1,C_2>0$ such that for every $d$-tuple of nonnegative integers $\alpha = (\alpha_1,\alpha_2,\ldots,\alpha_d)$, 
\begin{equation}
\left\|\partial^\alpha h\right\|_{L^\infty(\Omega)}\leq C_1C^{|\alpha|}_2|\alpha|^{|\alpha|s}\,,
\end{equation}
where $|\alpha|=\sum^d_{i=1} |\alpha_i|$. For fixed constants $C_1,C_2,s$, the set of Gevrey functions is denoted by $\mathcal{G}^{s}_{C_1,C_2}(\Omega)$. Furthermore, $\mathcal{G}^s=\bigcup_{C_1,C_2>0}\mathcal{G}^s_{C_1,C_2}$.
\end{defn}

Some useful properties of Gevrey functions are collected in \cref{sec:gevrey}. In particular, the product of two Gevrey functions is a Gevrey function (\cref{lem:product_gevrey}); certain compositions of Gevrey functions are Gevrey functions (\cref{prop:exp_gevrey}); the Fourier transform of compactly supported Gevrey functions satisfies Paley-Wiener type estimates (\cref{lem:hat_Gev_decay}).

\begin{assumption}[Weighting function]
\label{assumption:q}
For $\beta,S>0$, suppose that $q(\nu)$ is a weighting function of the form 
\re{$q(\nu)=u(\beta \nu)w(\nu/S)$} with the following conditions:
\begin{itemize}
    \item  \emph{(Symmetry)} For any $\nu\in\RR$, $u(\nu)=\overline{u(-\nu)},w(\nu)=\overline{w(-\nu)}$. 
    \item \emph{(Compact support)} \re{$w(\nu)=1$ when $|\nu|\leq 1/2$ and $w(\nu)=0$ when $|\nu|\geq 1$.} 
    \item \emph{(Gevrey)} There exists $\xi_q,\xi_u,\xi_w\geq 1$ and $s \ge 1$ such that 
    \begin{equation*}
         u(\nu)e^{-\nu/4}\in\mathcal{G}^{s}_{\xi_q,\xi_u}(\mathbb{R})\,,\q  w\in\mathcal{G}^{s}_{\xi_q,\xi_w}(\mathbb{R})\,.
    \end{equation*}
\end{itemize}
 In addition, we assume $\frac{\rd}{\rd \nu} (u(\nu)e^{-\nu/4}) \in L^1(\R)$ and denote
    \begin{equation*}
        C_{1,u} :=  \Big\| \frac{\rd}{\rd \nu} (u(\nu)e^{-\nu/4}) \Big\|_{L^1(\R)}\,.
    \end{equation*}
\end{assumption}
\begin{rem} \re{The above Gevrey functions with compact support are technical assumptions for our complexity analysis. Specifically, due to the properties of Gevrey functions discussed in \cref{sec:gevrey}, we can demonstrate the rapid decay of $f^a$ and $g$ (see \cref{lem:property_f_g}). Combining this with the compact support property of $\widehat{q}(\nu)$, we can apply the Poisson summation formula to control the quadrature error and normalization constant (see \cref{lem:discretization_L_G}). In practical applications, we believe it is sufficient to choose a weighting function $q(\nu)$ with certain smoothness and decay to ensure efficient simulation of the quantum Gibbs sampler.}
\end{rem}

\re{In the above assumption, $S$ is an adjustable parameter to control the support of $q$, and it is easy to see that $w(\nu/S)$ is supported on $[-S,S]$ and satisfies $w(\nu/S)\in \mathcal{G}^{s}_{\xi_q,\xi_w/S}$. Then, for a weighting function $q$ in \cref{assumption:q}, there holds (\cref{lem:product_gevrey}) 
\begin{equation} \label{eq:relaq1}
    q(\nu)=\overline{q(-\nu)}\,,\q \mathrm{supp}(q)\subset [-S,S]\,,
\end{equation}
and 
\begin{equation} \label{eq:relaq2}
    q(\nu)e^{-\beta \nu/4}\in \mathcal{G}^{s,}_{\xi_q^2,\beta \xi_u+\xi_w/S}(\mathbb{R})\,.
\end{equation}
Intuitively, the functions $u(\beta\nu)$ and $w(\nu/S)$ control the magnitude and support of $q$, respectively, which in turn controls the energy transition induced by jump operators $\{L_a\}_{a \in \mc{A}}$~\eqref{eqn:L_a_formula}.
The existence of a \emph{bump function} $w$ required in \cref{assumption:q} is guaranteed by \cite[Corollary 2.8]{Adwan_2017}: for any $s_w>1$, we can construct $w\in \mathcal{G}^{s_w}_{\xi_q,\xi_w}$ with $\mathrm{supp}(w)\subset [-1,1]$ and $w(\nu)=1$ for $|\nu|\leq 1/2$.}

\re{We next discuss the choice of the weighting function $\kappa$ for the coherent term. Let $\{L_a\}_{a \in \mc{A}}$ be the jump operators defined via \cref{eqn:L_a_formula} with $q^a$ satisfying \cref{assumption:q}. We recall the representation \eqref{eqn:g} of $G$ and compute, for $a \in \mc{A}$, 
\begin{align*}
    L_a^\dag L_a 
    & = \sum_{\nu ,\nu'\in B_H} e^{-\beta(\nu +\nu')/4} q^a(\nu)\overline{q^a(\nu')} \left(A^a_{\nu'}\right)^\dagger A^a_{\nu} \\
    & = \sum_{\lad_j, \lad_{j'} \in {\rm Spec}(H)}  \sum_{\lad_i \in {\rm Spec}(H)} e^{-\beta(2\lad_i - \lad_j - \lad_{j'})/4} q^a(\lad_i - \lad_j)\overline{q^a(\lad_{i}-\lad_{j'})}\, P_{j'} \left(A^a\right)^\dagger P_i A^a P_j\,,
\end{align*}
which, by \eqref{eqn:A_v}, implies that for $\nu \in B_H$, 
\begin{align*}
    (L_a^\dag L_a)_\nu 
    & = \sum_{\lad_{j'} - \lad_j = \nu}  P_{j'} \left( \sum_{\lad_i \in {\rm Spec}(H)} e^{-\beta(2\lad_i - \lad_j - \lad_{j'})/4} q^a(\lad_i - \lad_j)\overline{q^a(\lad_{i}-\lad_{j'})}\,  \left(A^a\right)^\dagger P_i A^a \right) P_j\,.
\end{align*}
By \cref{eq:relaq1} from \cref{assumption:q}, we have that when $|\lad_j - \lad_j'| > 2S$, for any $\ww \in \R$, 
\begin{equation*}
    q^a(\ww - \lad_j)\overline{q^a(\ww - \lad_{j'})} = 0\,.
\end{equation*}
It follows that $(L_a^\dag L_a)_\nu  = 0$ if $|\nu| > 2S$ and we can relax the constraint \eqref{eq:condkappa} to $\kappa(\nu) = 1$ for $|v| \le 2S$. Without loss of generality, we let $\kappa$ be
\begin{equation} \label{assp:kappa}
    \kappa(\nu) = w(\nu/4 S) \in \mathcal{G}^{s}_{\xi_q,\xi_w/4S}\,,
\end{equation}
with $w(\dd)$ given in \cref{assumption:q}.}

We first prove that the filtering functions in the time domain $f^a(t)$ and $g(t)$ associated with $q^a$ and $\kappa$ given above decay rapidly (\cref{lem:property_f_g}). This further allows us to 
show that a simple quadrature scheme (trapezoidal rule) can efficiently approximate $\{L_a\}_{a\in\mathcal{A}}$ and $G$ with high accuracy, \re{as summarized in \cref{lem:discretization_L_G}.} Specifically, given $M=2^{\mathfrak{m}-1}$ with $\mathfrak{m}\in\mathbb{N}_+$ and $\tau>0$, the quadrature points are given by 
\begin{equation}\label{eqn:quad_t}
t_m = -M\tau+m\tau\,, \quad 0\leq m< 2M\,.
\end{equation}
The quadrature error can be controlled as follows, and its proof 
is given in \cref{sec:discretization}. 
\rre{For simplicity, in the main text, we assume that the parameters $\xi_q$, $\xi_u$, $\xi_w$ and $C_{1,u}$ for Gevrey functions in \cref{assumption:q} are universal constants and include them in the asymptotic notations in related statements. We also assume $S \ge 1$ in the following discussion for simplicity. 
The explicit dependence on these parameters will be provided in \cref{sec:discretization} for completeness.}



\begin{prop}[Quadrature error]\label{lem:discretization_L_G} 
Under \cref{assumption:q}, we assume $\beta>0$, \rre{$S \ge 1$}, $\|A^a\|\leq 1$ for any $a\in\mc{A}$. When \re{$\tau\le \frac{\pi}{\norm{H} + 2S}$} and  
\begin{equation} \label{eq:condquad}
  \rre{(M-1)\tau = \Omega \left((\beta + 2) \log (\beta + 2) \right)\,,}
\end{equation} 
it holds that 
\begin{equation*}
\left\|L_a-\sum^{2 M - 1}_{m = 0}f^a(t_m) A^a(t_m) \tau\right\|\leq C_fS\exp\left(-\frac{s ((M-1)\tau)^{1/s}}{2(\beta \xi_u + \xi_w/S)^{1/s}e} \right)\,,
\end{equation*}
with
\begin{align*}
  \rre{C_f = \mathcal{O}\big((\beta + 2)^{1/s}\big)\,,}
\end{align*}
and 
\begin{equation*}
\begin{aligned}
&\left\|G-\sum^{2 M - 1}_{m = 0}g(t_m) H_L(t_m) \tau\right\| \le C_g S |\mathcal{A}|\exp\left(-\frac{s((M-1)\tau)^{1/s}}{2e(\beta+\xi_w/4S)^{1/s}}\right)\,,
\end{aligned}
\end{equation*}
with 
\begin{equation*} 
     \rre{C_g = {\mathcal{O}}\big((\beta + 2)^{1/s} \log^2((\beta + 2)  S)\big)\,.}
\end{equation*}
Here $\Omega(\cdot)$, $\mathcal{O}(\cdot)$ absorbs some constant depending \rre{on $s$, $\xi_q$, $\xi_u$, $\xi_w$ and $C_{1,u}$.} 
\end{prop}

Thanks to \cref{lem:discretization_L_G}, to approximately block encode $\{L_a\}$ and $G$, it suffices to construct block-encodings for the discretized quantities 
\[
\sum^{2M-1}_{m = 0}f^a(t_m)e^{i H t_m}A^ae^{-i H t_m}\tau\,,
\]
and
\[\re{
\sum^{2M-1}_{n= 0} g(t_n) e^{iH t_n}\left(\sum_{a \in \mc{A}} L^\dagger_aL_a\right)e^{-iH t_n}\tau\,.}
\]
In our algorithm, we construct these two block encodings using LCU (see \cref{app:LCU}). This utilizes block encodings of $A^a$~\eqref{eqn:A_block_encoding}, controlled Hamiltonian simulation~\eqref{eqn:U_H}, and prepare oracles for $f$ (\cref{eqn:prep_f,eqn:prep_f_bar}) or $g$ (\cref{eqn:prep_g,eqn:prep_g_bar}).  The detailed constructions are presented in the next subsection (see \eqref{eqn:U_L} and \eqref{eqn:U_G}).

\begin{rem}
Bounding the approximation error for $\{L_a\}_{a\in\mathcal{A}}$ and $G$ in the operator norm is a nontrivial task. 
\cite{ChenKastoryanoBrandaoEtAl2023} introduces a ``rounding Hamiltonian'' technique to bound the quadrature error in the frequency domain. By choosing weighting functions in \cref{assumption:q}, we can use the Poisson summation formula to simplify the quadrature error analysis. In particular, we can bound the quadrature error in the time domain without using the ``rounding Hamiltonian'' technique. 
\end{rem}

Now we provide two specific examples of $q$ satisfying \cref{assumption:q}. \begin{itemize}
\item Metropolis-type:  
\begin{equation}\label{eqn:q_metroplis}
q(\nu)= e^{-\sqrt{1+\beta^2\nu^2}/4}w(\nu/S) \q \text{with} \q u(\nu) = e^{-\sqrt{1+ \nu^2}/4}\,,
\end{equation}
where $u(\nu)e^{-\frac{\nu}{4}} = e^{-\frac{\sqrt{1+\nu^2} + \nu}{4}} \in\mathcal{G}^1_{1,7/2}$ is a \re{Gevrey function of order $1$} and its derivative is $L^1$-integrable; see 
\cref{prop:exp_gevrey}. 
When $\beta\gg 1$, $q(\nu)$ in \cref{eqn:q_metroplis} gives a smoothed version of the Metropolis-type filter (similar to the Glauber-type filter):
\begin{equation}\label{eqn:f_metroplis}
\widehat{f}(\nu)= q(\nu) e^{-\beta\nu/4} \approx 
\min\left\{1,e^{-\beta\nu/2}\right\},\quad \nu\in [-S/2,S/2].
\end{equation}
\item Gaussian-type\footnote{In this case, the Gaussian functions already decay rapidly. The multiplication with a bump function is for purely technical reasons to ensure $\hat{f}(\nu)$ is compactly supported.}:
\begin{equation}\label{eqn:q_gaussian}
q(\nu)= e^{-(\beta\nu)^2/8}w(\nu/S) \q \text{with} \q u(\nu) = e^{-\nu^2/8}\,,
\end{equation}
Here $u(\nu)e^{-\frac{\nu}{4}} = e^{-\frac{(\nu+1)^2 - 1}{8}}$ is a Gevrey function of order $1$ 
by \cite[Proposition B.1]{Gus_2019} and the $L^1$-integrability of its derivative is straightforward.
Setting $S=\mathcal{O}(1/\beta)$,  there holds 
\begin{equation}
\widehat{f}(\nu)=q(\nu)e^{-\nu/4} \propto e^{-(\beta\nu+1)^2/8}\,,
\end{equation}
which is approximately a Gaussian function concentrated at $-\beta^{-1}$ with width $\Or(\beta^{-1})$.
\end{itemize}

A comparison of the shapes of the Metropolis-type and Gaussian-type filtering function $\widehat{f}(\nu)$ is shown in \cref{fig:f}. The support size for the Gaussian choice decreases as $\mathcal{O}(\beta^{-1})$, which can cause inefficiency because the magnitude of a local move in Monte Carlo simulations stays around order 1, regardless of the value of $\beta$. 

\begin{figure}[htbp]
     \subfloat{
         \centering
         \includegraphics[width=0.48\textwidth]{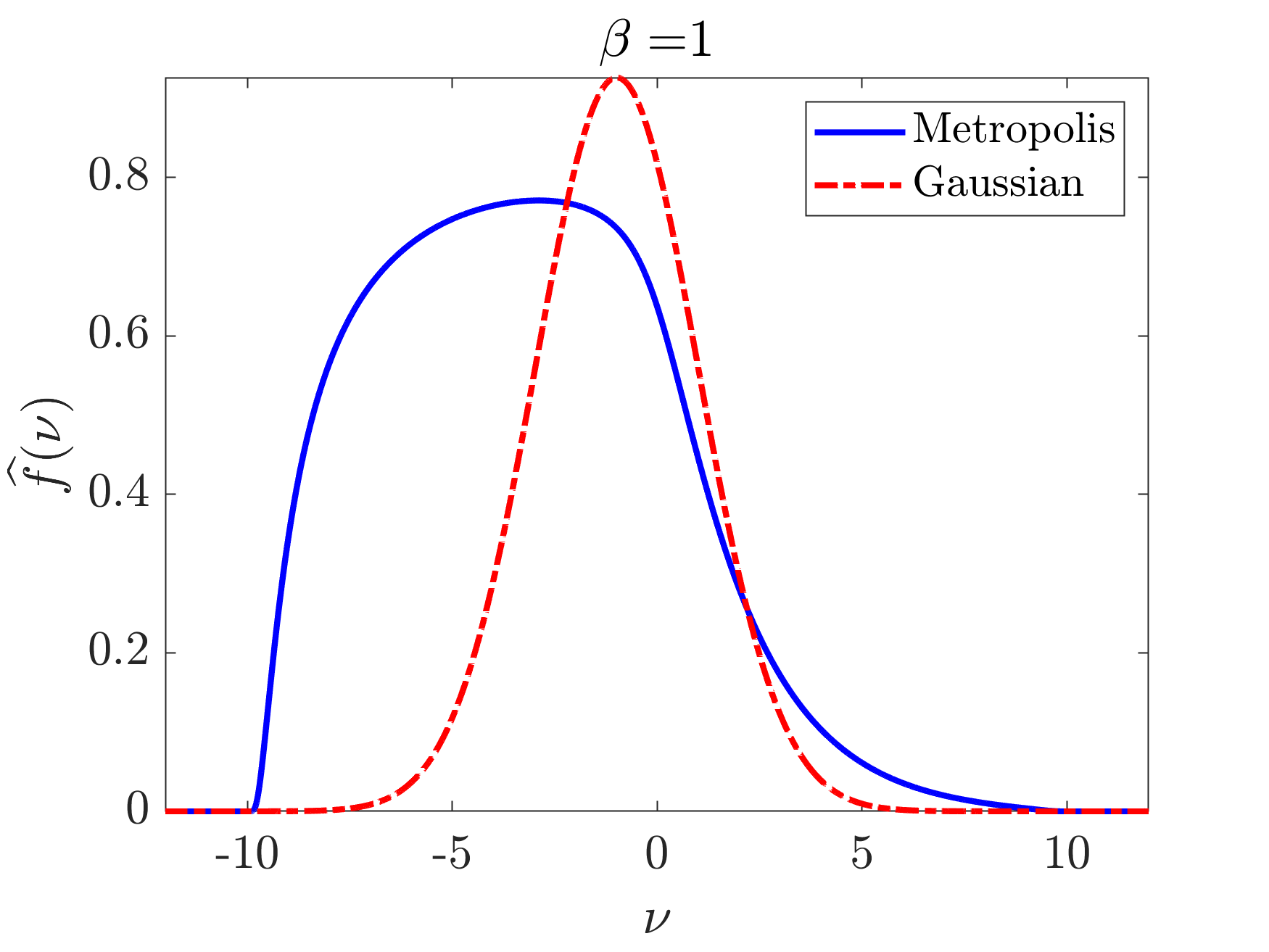}
         \includegraphics[width=0.48\textwidth]{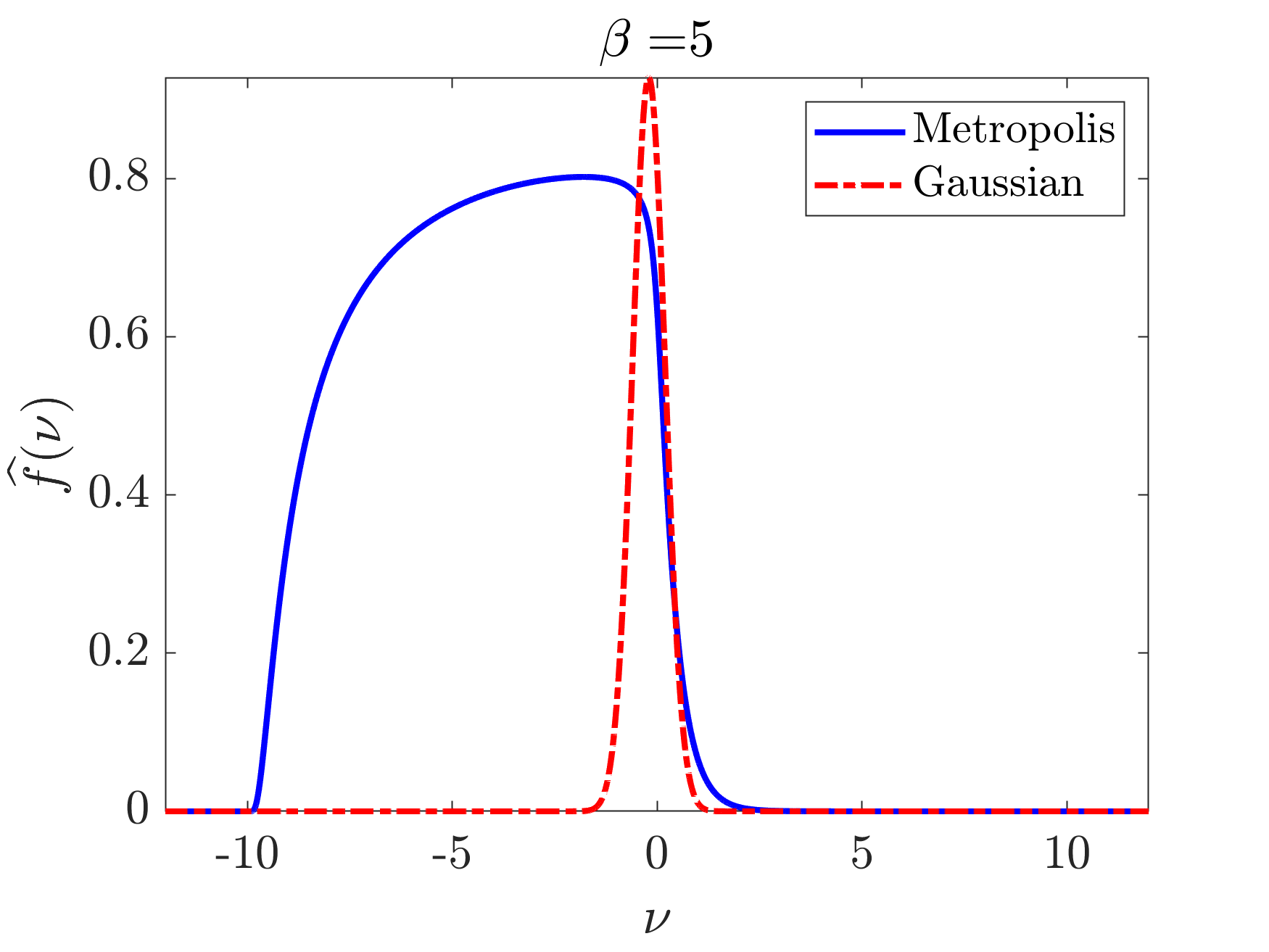}
     }
     \caption{
     \label{fig:f} Comparison between Metropolis-type and Gaussian-type filtering functions $\widehat{f}(\nu)$ in the frequency domain with $\beta=1$ (left) and $\beta=5$ (right). For simplicity, the bump function is chosen to be \magg{$w(\nu/S) = \exp(-20/(10^2-\min\{10, |\/S|\})^2)$ with $S=10$}.
     The approximate support size of $\widehat{f}(\nu)$ remains $\mathcal{O}(1)$ for the Metropolis-type filter as $\beta \to \infty$, while it narrows to $\mathcal{O}(\beta^{-1})$ for the Gaussian-type  filter.
}
\end{figure}

\subsection{Efficient simulation of the Lindblad master equation \texorpdfstring{\eqref{eqn:Lindblad_master_equation}}{Lg}}\label{sec:simulation}

In this section, we discuss the efficient simulation of the Lindblad equation in \eqref{eqn:Lindblad_master_equation}. For simplicity, we assume that \re{$q^a=q$ satisfies \cref{assumption:q} and then
$f^a=f$.} Our construction uses the block encoding input model and the linear combination of unitaries (see \cref{app:block_encoding,app:LCU}).  
We also assume $\max_{a\in\mathcal{A}}\|A^a\|\leq 1$ to ensure that efficient block encodings of $A^a$ are available (see \cref{eqn:A_block_encoding}).


Thanks to our algorithm's use of a discrete set of jump operators for
the Lindbladian, we can directly apply efficient Lindblad simulation
quantum algorithms, such as those in~\cite{CW17,LiWang2023,
ChenKastoryanoGilyen2023,DingLiLin2024}, to prepare the Gibbs state, once the efficient constructions of the block encodings of $\{L_a\}$ and $G$ are available, which will be the focus of the rest of this section.

According to \cref{lem:discretization_L_G}, in the following discussion, we set the integer $\mathfrak{m}$ large enough and 
consider the quadrature points $\{t_m\}^{2M-1}_{m=0}$ ($M=2^{\mathfrak{m}-1}$) as in \cref{eqn:quad_t} such that \cref{eq:condquad} holds. Without loss of generality, we assume $|\mathcal{A}|=2^{\mathfrak{a}}$ with $\mathfrak{a} \in \mathbb{N}_+$.



For the quantum simulation, we assume access to the following oracles:
\begin{itemize}
    \item Block encoding $U_{\mathcal{A}}$  of the coupling operators $\{A^a\}_{a\in\mathcal{A}}$:
    \begin{equation}\label{eqn:A_block_encoding}
     (\mathrm{I}_{\mathfrak{a}}\otimes \bra{0^{\mathfrak{b}}}\otimes \mathrm{I}_{n})\cdot U_\mathcal{A}\cdot (\mathrm{I}_{\mathfrak{a}}\otimes \ket{0^{\mathfrak{b}}}\otimes \mathrm{I}_{n})=\sum_a \ket{a}\bra{a}\otimes A^a/Z_{\mathcal{A}}\,.
    \end{equation}
     We assume that the block encoding factor  $Z_{\mathcal{A}}$ can be chosen to satisfy $\max_{a\in\mathcal{A}} \|A^a\|\leq Z_{\mathcal{A}} \leq 1$. Here $\ket{0^{\mathfrak{b}}}$ represents the ancilla qubits utilized in the block-encoding of $A^a$, $\mathrm{I}_{\mathfrak{a}}$ is the identity matrix acts on the index register. 
    
    \item Controlled Hamiltonian simulations\footnote{The circuit construction is similar to the controlled Hamiltonian simulations in the standard quantum phase estimation, and the total Hamiltonian simulation time required is $\Or(M\tau)$.} for $\{t_m\}^{2M-1}_{m=0}$:
    \begin{equation}\label{eqn:U_H}
    U_{H}=\sum^{2M-1}_{m=0} \ketbra{t_m}\otimes \exp(-i t_m H)\,.
    \end{equation}
    \item Prepare oracle for $a\in\mathcal{A}$, acting on the index register:
    
    
    \begin{equation}\label{eqn:A}
\textbf{Prep}_{\mathcal{A}}=\mathrm{H}^{\otimes \mathfrak{a}}\ket{0^{\mathfrak{a}}}=\frac{1}{\sqrt{|\mathcal{A}|}}\sum_{a\in\mathcal{A}}\ket{a}\,,
    \end{equation}
which is used to implement LCU for the sum over $a$ in $\mathcal{A}$ appearing in $G$. Here $\mathrm{H}^{\otimes \mathfrak{a}}$ are Hadamard gates acting the index register and are self-adjoint.

    \item Prepare oracles for the filtering functions $f$, acting on the time register:
    \begin{equation}\label{eqn:prep_f}
     \textbf{Prep}_{f}: \ket{0^{\mathfrak{m}}}=\frac{1}{\sqrt{Z_f}}\sum^{2M-1}_{m=0}\sqrt{f(t_m)\tau}\ket{t_m}\,,
    \end{equation}
    and
   \begin{equation}\label{eqn:prep_f_bar}
        \textbf{Prep}_{\overline{f}}: \ket{0^{\mathfrak{m}}}=\frac{1}{\sqrt{Z_f}}\sum^{2M-1}_{m=0}\sqrt{\overline{f(t_m)}\tau}\ket{t_m}\,.
    \end{equation}
Here the block encoding factor $ Z_f := \sum^{2M-1}_{m=0}|f(t_m)|\tau$ is bounded by \cref{lem:zfzg} with \cref{eq:inttf}:
    \begin{align} \label{eq:zf}
        Z_f =  \mathcal{O}\left( (\xi_q C_{1,u} + \xi^2_q \xi_w) \log ((\beta \xi_u + \xi_w/S)\max\{S,1\})\right)\,.
    \end{align} 
     \item \re{Prepare oracles for the $g$ function, acting on the frequency register:
     \begin{equation}\label{eqn:prep_g}
     \textbf{Prep}_{g}: \ket{0^\mathfrak{m}}=\frac{1}{\sqrt{Z_g}}\sum^{2M-1}_{m=0}\sqrt{g(t_m)\tau}\ket{t_m}\,,
     \end{equation}
    and
    \begin{equation}\label{eqn:prep_g_bar}
    \textbf{Prep}_{\overline{g}}: \ket{0^\mathfrak{m}}=\frac{1}{\sqrt{Z_g}}\sum^{2M-1}_{m=0}\sqrt{\overline{g(t_m)}\tau}\ket{t_m}\,,
     \end{equation}
 with the block encoding factor $Z_g=\sum^{2M-1}_{m=0}|g(t_m)|\tau$ bounded by \cref{lem:zfzg} with \cref{eq:inttg}}:
\begin{align} \label{eq:zg}
  \re{Z_g = \Or\left(\xi_q(1+\log((\beta +\xi_w/S)\max\{S,1\}))\right)\,.}
\end{align}
\end{itemize}

Using these oracles, according to \cref{lem:discretization_L_G}, we can apply LCU with \eqref{eqn:A_block_encoding}, \eqref{eqn:U_H}, \eqref{eqn:prep_f}, and \eqref{eqn:prep_f_bar} to construct a 
\re{$$\left(Z_fZ_{\mathcal{A}},\mathfrak{m}+\mathfrak{a}+\mathfrak{b}, C_fS|\mc{A}|\exp\left(-\frac{s ((M-1)\tau)^{1/s}}{2(\beta \xi_u + \xi_w/S)^{1/s}e} \right)\right)$$}-block encoding $U_{L}$ of $\sum_a\ket{a}\bra{a}\otimes L_a$:
\begin{equation}\label{eqn:U_L}
U_{L}=\underbrace{\left(\textbf{Prep}^\dagger_{\overline{f}}\otimes \mathrm{I}_{\mathfrak{a}+\mathfrak{b}}\otimes \mathrm{I}_{n}\right)}_{\text{LCU prepare oracle}}\cdot\underbrace{\mathrm{I}_{\mathfrak{a}+\mathfrak{b}}\otimes U_{H}^\dagger}_{e^{iHt_m}} \cdot\underbrace{(\mathrm{I}_{\mathfrak{m}}\otimes U_\mathcal{A})}_{\ket{a}\bra{a}\otimes A^a}\cdot \underbrace{\mathrm{I}_{\mathfrak{a}+\mathfrak{b}}\otimes U_{H}}_{e^{-iHt_m}}\cdot\underbrace{\left(\textbf{Prep}_{f}\otimes \mathrm{I}_{\mathfrak{a}+\mathfrak{b}}\otimes \mathrm{I}_{n}\right)}_{\text{LCU prepare oracle}}\,.
\end{equation}
Here $C_f$ is defined in \cref{lem:discretization_L_G}.
The circuit of $U_{L}$ can be found in \cref{fig:circuit}. In \eqref{eqn:U_L}, the total Hamiltonian simulation time required by one query to $U_{L}$ is $\mathcal{O}\left(M\tau\right)$.

 Next, applying two layers of LCU (see \cref{app:LCU}) with \eqref{eqn:A_block_encoding}, \eqref{eqn:U_H}, \eqref{eqn:A}, \eqref{eqn:prep_g}, and \eqref{eqn:prep_g_bar}, we construct the following
\re{\[\left(Z_gZ^2_{\mathcal{A}}|\mathcal{A}|,3\mathfrak{m}+\mathfrak{a}+2\mathfrak{b}, 
C_g S |\mathcal{A}|\exp\left(-\frac{s((M-1)\tau)^{1/s}}{2e(\beta+\xi_w/4S)^{1/s}}\right)\right)\]}-block encoding $U_G$ of $G$:
\begin{equation}\label{eqn:U_G}
\begin{aligned}
U_G=&\underbrace{\left(\mathrm{I}_{\mathfrak{m}}\otimes \textbf{Prep}_{\mathcal{A}}\otimes \mathrm{I}_{\mathfrak{m}+\mathfrak{b}}\otimes \mathrm{I}_{n}\otimes \mathrm{I}_{\mathfrak{m}+\mathfrak{b}}\right)^\dagger\cdot\left(\textbf{Prep}^\dagger_{\overline{g}}\otimes \mathrm{I}_{\mathfrak{a}+\mathfrak{m}+\mathfrak{b}}\otimes \mathrm{I}_{n}\otimes \mathrm{I}_{\mathfrak{m}+\mathfrak{b}}\right)\cdot}_{\text{Two layers of LCU prepare oracles}}\\
&\underbrace{\mathrm{I}_{\mathfrak{a}+2\mathfrak{m}+\mathfrak{b}}\otimes U_{H}^\dagger\otimes \mathrm{I}_{\mathfrak{m}+\mathfrak{b}}}_{\exp(iHt_n)}\cdot \underbrace{\left(\mathrm{I}_{\mathfrak{m}}\otimes \mathrm{I}_{\mathfrak{m}+\mathfrak{b}}\otimes U^\dagger_{L}\right)\left(\mathrm{I}_{\mathfrak{m}}\otimes U_{L}\otimes \mathrm{I}_{\mathfrak{m}+\mathfrak{b}}\right)}_{\ketbra{a}\otimes L^\dagger_aL_a}\cdot\underbrace{\mathrm{I}_{\mathfrak{a}+2\mathfrak{m}+\mathfrak{b}}\otimes U_{H}\otimes \mathrm{I}_{\mathfrak{m}+\mathfrak{b}}}_{\exp(-iHt_n)}\\
&\underbrace{\cdot\left(\textbf{Prep}_{g}\otimes \mathrm{I}_{\mathfrak{a}+\mathfrak{m}+\mathfrak{b}} \otimes \mathrm{I}_{n}\otimes\mathrm{I}_{\mathfrak{m}+\mathfrak{b}}\right)\cdot\left(\mathrm{I}_{\mathfrak{m}}\otimes \textbf{Prep}_{\mathcal{A}}\otimes \mathrm{I}_{\mathfrak{m}+\mathfrak{b}}\otimes \mathrm{I}_{n}\otimes\mathrm{I}_{\mathfrak{m}+\mathfrak{b}}\right)}_{\text{Two layers of LCU  prepare oracles}}\,.    
\end{aligned}
\end{equation}
Here $C_g$ is defined in \cref{lem:discretization_L_G}. We note that $U_H$ for $\exp(-iHt_m)$ and $\exp(-iHt_n)$ are acting on different time registers. The circuit of $U_G$ is given in \cref{fig:circuit}. In \eqref{eqn:U_G}, one query to $U_G$ still requires $\mathcal{O}\left(M\tau\right)$ total Hamiltonian simulation time.

\begin{figure}[htbp]
     \subfloat[Block encoding $U_{L}$ of jump operators $\sum_{a\in\mathcal{A}}\ket{a}\bra{a}\otimes L_a$.]{
         \centering
         \includegraphics[width=0.9\textwidth]{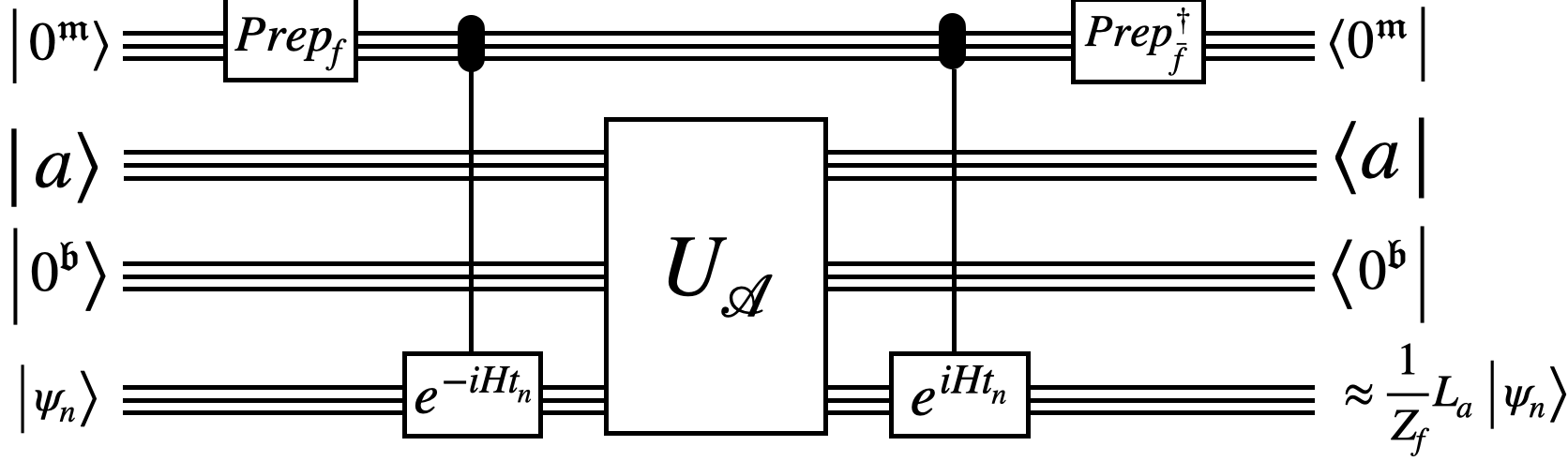}
     }\\
     \subfloat[Block encoding $U_G$ of the coherent term $G$.]{
         \centering
         \includegraphics[width=0.85\textwidth]{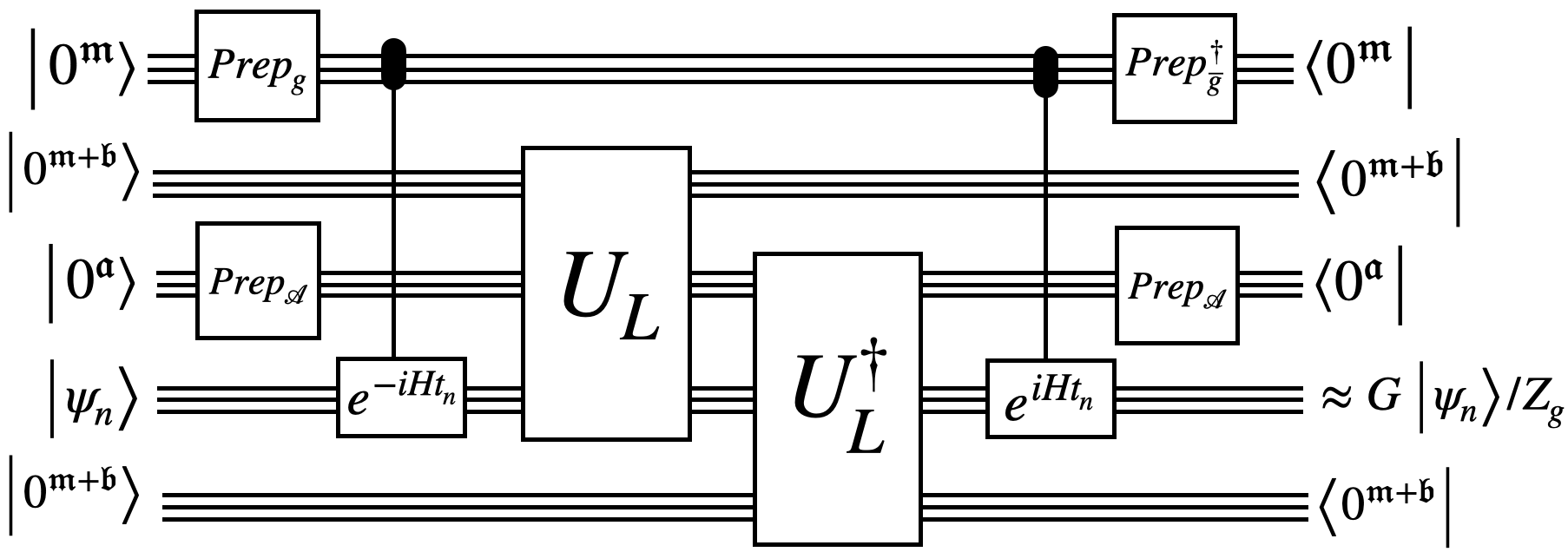}
     }
     \caption{ Quantum circuits for block encodings of $\{L_a\}_{a \in \mc{A}}$ (top) and $G$ (bottom). 
     } \label{fig:circuit}
\end{figure}

After acquiring the block encodings of $\{L_a\}_{a \in \mc{A}}$ and $G$, we can employ the algorithm proposed in \cite{LiWang2023} to simulate \eqref{eqn:Lindblad_master_equation}. The complexity of this algorithm is recalled below.


\begin{thm}[{\cite[Theorem 11]{LiWang2023}}]\label{thm:simulation_LW} Suppose that we are given an $(A_g,\mathfrak{g},\delta)$-block encoding $U_G$ of $G$, and $(A_f,\mathfrak{f},\delta)$-block encodings $\{U_a\}$ for the jumps $\{L_a\}$\footnote{In \cite{LiWang2023}, the authors assume separate access to the block encodings $U_a$ of $L_a$. This is slightly different from our setting assuming the block-encoding $U_L$ of $\ket{a}\bra{a}\otimes L_a$. }. Let $\|\mathcal{L}\|_{\rm be} := A_g + \frac{1}{2}A^2_f|\mathcal{A}|$. For all $t,\epsilon$ with $\delta\leq \epsilon/t\|\mathcal{L}\|_{\rm be}$, there exists a quantum algorithm for simulating \eqref{eqn:Lindblad_master_equation} up to time $t$ with an $\epsilon$-diamond distance using 
\begin{equation} \label{eq:querynumber}
    \mathcal{O}\left(t\|\mathcal{L}\|_{\rm be}\log\left(t\|\mathcal{L}\|_{\rm be}/\epsilon\right)\right)
\end{equation}
queries to $U_G$ and $\{U_{a}\}_{a\in\mathcal{A}}$ and 
\begin{equation*}
\mathcal{O}\left(\log\left(t\|\mathcal{L}\|_{\rm be}/\epsilon\right)\left(\log |\mathcal{A}|+\log\left(t\|\mathcal{L}\|_{\rm be}/\epsilon\right)\right)\right)
\end{equation*}
additional ancilla qubits.

\end{thm}


Now, we are ready to give the simulation cost of our method as follows, \rre{where we have assumed $S \ge 1$ for simplicity as in \cref{lem:discretization_L_G}.
The proof is postponed to \cref{sec:discretization}.}

\begin{thm}\label{thm:simulation} Assume access to weighting functions $\{q^a\}$ satisfying \cref{assumption:q} with $s>1$, block encodings $U_{\mathcal{A}}$ in \cref{eqn:A_block_encoding}, controlled Hamiltonian simulation $U_H$ in \cref{eqn:U_H}, and prepare oracles for filtering functions $\{f^a\}$ and $g$ in \eqref{eqn:prep_f}--\eqref{eqn:prep_g_bar}. The Lindbladian evolution \eqref{eqn:Lindblad_master_equation} \rre{with $\beta > 0$} can be simulated up to time $t_{\rm mix}$ within $\eps$-diamond distance with total Hamiltonian simulation time:
\begin{equation*}
  \widetilde{\mathcal{O}}\left(C_qt_{\rm mix} \rre{(\beta + 1)} |\mc{A}|^2  \log^{1+s}\left(1/\epsilon\right)\right)\,,
\end{equation*}
where the constant $C_q$ is defined as follows:
\begin{align*}
\rre{C_q:= \log^{2 + s}(S) + 1 \,.}
\end{align*}
In addition, the algorithm requires 
\begin{equation*}
\rre{\widetilde{\mathcal{O}}\left(\log(S^2 + S\norm{H}) + \log^2(t_{\rm mix} |\mc{A}|/\epsilon) + \log (\beta + 2) \right)\,.}
\end{equation*}
number of additional ancilla qubits for the prepare oracles and simulation. The $\widetilde{\mathcal{O}}$ absorbs \rre{a constant only depending on $s$, $\xi_q$, $\xi_u$, $\xi_w$, and $C_{1,u}$}, and subdominant polylogarithmic dependencies on parameters $t_{\rm mix}$, $\|H\|$, $|\mathcal{A}|$, $S$, $\xi_q$, $\xi_u$, $\xi_w$, and $\beta$. 
\end{thm}

\begin{rem}
In our simulation algorithm, thanks to the finite number of jump operators, we only need to construct the block encoding of $\sum_{a\in\mathcal{A}}\ket{a}\bra{a}\otimes L_a$ for the efficient simulation. 
When $|\mathcal{A}|\gg 1$, by assuming oracle access to a different form of blocking encodings of $A^a$, we can simultaneously construct block encodings of all jump operators. Specifically, assuming oracle access to the block encoding of all coupling operators in the form $\sum_{a\in\mathcal{A}}\ket{a}\bra{0^{\mathfrak{a}}}\otimes A^a$, we can initially utilize \cite[Appendix B.1 Lemma III.1]{ChenKastoryanoBrandaoEtAl2023} to construct a block encoding $U_{L,\rm all}$ for $\sum_{a\in\mathcal{A}}\ket{a}\bra{0^{\mathfrak{a}}}\otimes L_a$ and a block encoding $U_G$ for $G$. Leveraging $U_{L,\rm all}$, we can implement the weak-measurement scheme proposed in \cite[Section III.1]{ChenKastoryanoGilyen2023} to simulate \eqref{eqn:Lindblad_master_equation} to first-order accuracy. Finally, by applying ``compression'' techniques as outlined in \cite{CW17} to reduce the number of repetitions, the algorithm proposed in \cite[Appendix F]{ChenKastoryanoGilyen2023} achieves optimal scaling in the number of uses of $U_{L,\rm all}$ and $U_G$. 
\end{rem}

\subsection{Quasi-locality}

\rre{
It is worth noting that when $H$ is a local Hamiltonian and the coupling operators $\{A^a\}_{a\in\mc{A}}$ are local, then we can show that at \emph{any} temperature, our jump operators $L_a$ in \eqref{eqn:L_a_formula} are quasi-local, and the coherent term $G$ in \eqref{eqn:G_KMS} is a sum of quasi-local operators.
This can be demonstrated by applying the Lieb-Robinson bounds~\cite{Nach_2006,Hastings_2006,Lieb_1972,PhysRevB.69.104431,doi:10.1137/18M1231511}. Similar to \cite{rouz2024}, for a finite lattice $\Lambda \subset \mathbb{Z}^d$, we consider a $(k,l)$-local Hamiltonian  $H = \sum_{X \subset \Lambda} h_X$, where $h_X$ is a local potential with support on no more than $k$ sites, and each site appears in at most $l$ operators $h_X$. 

\begin{prop}[Quasi-locality of jump and coherent operators] \label{lem:quasi-local}
Let $H = \sum_{X \subset \Lambda} h_X$ be a $(k,l)$-local Hamiltonian defined on a finite lattice $\Lambda$, and let each coupling operator $A^a$ be a local operator. Then, the jump operator $L_a$ constructed via a Gaussian-type weighting function is quasi-local, and the associated coherent operator $G$ is a sum of quasi-local operators. 
\end{prop}
\begin{proof}
For simplicity, assume that each coupling $A^a$ is supported on a single site of the lattice denoted by $u_a$. 
We consider the Gaussian weighting function $q(\nu)=\exp(-(\beta\nu)^2/8)\omega(\nu/S)$ as the Gaussian-type one in \eqref{eqn:q_gaussian} and set $S \ge 4\|H\|$. Define $\widehat{h}(\nu)=\exp(-(\beta\nu)^2/8)$. Recalling the formula for $L_a$ in \eqref{eqn:L_a_formula}, we have
\begin{equation}\label{auxeq1}
    L_a=  
\sum_{\nu\in B_H} \widehat{f}(\nu) A^a_\nu=\sum_{\nu\in B_H} \widehat{h}(\nu) A^a_\nu=\int^\infty_{-\infty} h(t) e^{iHt} A^a e^{-iHt}\, \mathrm{d}t\,,
\end{equation}
where we use $\omega(\nu/S)=1$ when $|\nu|\leq 2\|H\|$ in the second equality. Then, from the Lieb-Robinson bound~\cite[Lemma 5]{doi:10.1137/18M1231511}, there exists a constant $J$ that only depends on $k,l,\sup_{X
}\|h_X\|$ such that
\[
\left\|e^{iHt} A^a e^{-iHt} - e^{i H_{B_{u_a}(r)}  t}  A^a e^{- i H_{B_{u_a}(r)} t}\right\| \leq \|A^a\|\min\left\{\frac{(J|t|)^r}{r!},2\right\},
\]
for any $r > 0$ and $a\in\mc{A}$. Here $B_{u_a}(r) \subset \Lambda$ is a ball centered at $u_a$ with radius $r$, and $H_{B_{u_a}(r)} = \sum_{X \cap B_{u_a}(r) \neq \emptyset} h_X$. This, along with \eqref{auxeq1}, readily implies 
\begin{equation}\label{eqn:error_bound}
\begin{aligned}
&\bigg\| L_a - \underbrace{\int_{-\infty}^{\infty} h(t) e^{iH_{B_{u_a}(r)}t} A^a e^{-iH_{B_{u_a}(r)}t} \, \mathrm{d}t}_{\text{\rm supported on } B_{u_a}(r)} \bigg\| \\
\leq &\frac{J^r \|A^a\|}{r!} \int_{-\infty}^{\infty} |h(t)| |t|^r \, \mathrm{d}t =\|A^a\| \left(\mathcal{O}\left(J\beta/\sqrt{r}\right)\right)^r\,,
\end{aligned}
\end{equation}
for any $r > 0$ and $\beta > 0$, where the last equality is from $h(t)=\exp\left(-\Theta((t/\beta)^2)\right)$. This implies that $L_a$ is quasi-local, and so is the operator $L_a^\dagger L_a$. 

Define $\widehat{j}(\nu)=-\frac{i}{2}\tanh\left(-\beta\nu/4\right)$ with Fourier transform $j(t)$, defined in distribution sense, which decays as $\exp(-\Theta(|t|))$ as $|t| \to \infty$ and has singularity $\Theta(1/t)$ as $t \to 0$. We recall the definition \eqref{eqn:G_KMS} of $G$ and find 
\begin{align*}
    & \bigg\| G - \sum_{a \in \mc{A}} \underbrace{\int_{-\infty}^\infty j(t) e^{i H_{B_{u_a}(r)} t} (L_a^\dag L_a) e^{ - i H_{B_{u_a}(r)} t} \ud t}_{\text{\rm quasi local}}  \bigg\| \\
    \le & \left(\sum_{a\in\mc{A}} \|L_a\|^2\right) \int_{-\infty}^{\infty} |j(t)|\min\left\{\frac{J^r|t|^r}{r!} ,2\right\}\mathrm{d}t \\
    \le & \left(\sum_{a\in\mc{A}} \|L_a\|^2\right) \left(\int_{|t|\leq \frac{r}{2eJ}}\frac{J^r|t|^r}{r!}|j(t)| \mathrm{d}t+2\int_{|t|>\frac{r}{2eJ}} |j(t)|\mathrm{d}t\right) \\ 
    = & \left(\sum_{a\in\mc{A}} \|L_a\|^2\right)\mc{O} \left(\int_{|t|\leq \frac{r}{2eJ}}\frac{J^r|t|^{r-1}}{r!}\mathrm{d}t - \log \tanh \left(\frac{\pi r}{2 e J \beta}\right)\right) \\
    = &\left(\sum_{a\in\mc{A}} \|L_a\|^2\right) \left(\mathcal{O}\left(\frac{1}{2^r}\right)+ e^{-\Theta(\frac{r}{J \beta})}\right)\,, \q r \ge 1\,.
\end{align*}
This implies that $G$ is the sum of quasi-local operators. 
\end{proof}

The above computation can be easily extended to the case where $q(\nu)=\widehat{h}(\nu)\omega(\nu/S)$ with $S \ge  4\|H\|$ and $|h(t)|\sim \exp(-\Theta(|t|^p))$ for $p\geq 1$ as $t \to \infty$. The quasi-locality does not affect the simulation complexity of the dynamics; however, it is a crucial property when analyzing the mixing time of the dynamics~\cite{rouz2024}. We briefly address this in the following remark.

\begin{rem}
To ensure that the constructed Lindblad dynamics \eqref{eq:lindbladkms} eventually relaxes to the desired Gibbs state, by \cref{lem:converg}, we should carefully choose the coupling operators $A^a$ and weighting functions $q^a$ such that the resulting $\{L_a\}$ and $G$ satisfy \cref{eq:irredu}. This is always possible, due to the finite dimensionality of the system. See also \cite[Theorem 8]{gilyen2024quantumgeneralizations} for related discussions. 
Moreover, it is known \cite[Theorem 9 and Section V]{TemmeKastoryanoRuskaiEtAl2010} that for a primitive KMS detailed balanced QMS, the mixing time \eqref{eqn:mixtime} can be estimated via the spectral gap of the Lindbladian. 

The recent work \cite{rouz2024} estimated the spectral gap of the efficient quantum Gibbs sampler in \cite{ChenKastoryanoGilyen2023} in the high-temperature regime, by mapping the Lindbladian $\mc{L}^\dag_\beta$ with KMS DBC to a Hamiltonian $\wt{\mc{L}}_\beta^\dag:= \si_\beta^{-1/4}\mc{L}^\dag_\beta (\si_\beta^{1/4}X \si_\beta^{1/4}) \si_\beta^{-1/4}$ and then analyzing its spectral properties by the stability of gapped Hamiltonians \cite{michalakis2013stability}. A key fact for using the stability result \cite{michalakis2013stability} is that the jump operators and the coherent term constructed in \cite{ChenKastoryanoGilyen2023} are quasi-local in the high-temperature regime. We prove in \cref{lem:quasi-local} that our $L_a$ and $G$ operators are also quasi-local at \emph{any} temperatures, paving the way for future spectral gap analyses of the proposed Gibbs sampler.
\end{rem}
}

\section{Recovery of the Gibbs sampler in \texorpdfstring{\cite{ChenKastoryanoGilyen2023}}{Lg}} \label{sec:recover}

In this section, we discuss the connections between our proposed family of efficient quantum Gibbs samplers with KMS DBC and those constructed in \cite{ChenKastoryanoBrandaoEtAl2023,ChenKastoryanoGilyen2023} and show that our framework can recover the one in \cite{ChenKastoryanoGilyen2023}.  

We have seen Davies generator \eqref{eq:davies} without the coherent term (i.e., Lindbladian with $\si_\beta$-GNS DBC) formally corresponds to the algorithmic Lindbladian \eqref{eqn:lindblad_filter} with $f\equiv 1$\footnote{\re{Rigorously  speaking, Davies generator \eqref{eq:davies} is the limit of Lindbladian \eqref{eqn:lindblad_filter} in \cite{ChenKastoryanoBrandaoEtAl2023}, where $|\widehat{f}(\omega)|^2$ approaches $\delta(\omega)$ due to the $L^2$-normalization $\int |f(t)|^2 = 1$.}}.  As emphasized in \cref{rem:ineffi}, such a choice of Dirac delta filtering function for the frequency makes it hard to approximate GNS detailed balanced Lindblad dynamic. Chen et al. \cite[Theorem I.3]{ChenKastoryanoBrandaoEtAl2023} introduced a Gaussian smoothed version by taking $f$ as 
\begin{equation} \label{eq:gaussianfilter}
    f(t)\propto \sqrt{\si_E} \exp(-t^2 \sigma^2_E)\,,
\end{equation}
with $\si_E$ of order $\eps/\beta$, which guarantees that the Gibbs state is an approximate fixed point. It follows that the parameter $\si_E$ has to be small enough so that $\widehat{f}(\ww) \propto \exp(-\ww^2/4 \sigma^2_E)/\sqrt{\si_E} \approx \d(\ww)/2\pi$, to prepare the Gibbs state accurately. Then \cite{ChenKastoryanoGilyen2023} carefully constructed coherent term $i [G, \dd]$ such that the resulting dynamics is $\si_\beta$-KMS detailed balanced and $\si_E$ could be a moderate constant, which reduced the computational cost significantly; see \cref{tab:comparison}. 


We next prove that our construction in \cref{sec:QGS} can include the one in \cite{ChenKastoryanoGilyen2023} as a special case. Let us first recall the construction by Chen, Kastoryano, and Gily\'en. Suppose that $\{A^a\}_{a\in \mathcal{A}}$ is a given set of operators satisfying $\{A^a\}_{a\in \mathcal{A}}=\{(A^a)^\dagger\}_{a\in \mathcal{A}}$. \cite[Corollary II.2, Proposition II.4]{ChenKastoryanoGilyen2023} defined the Lindbladian of the form \eqref{eqn:lindblad_filter}:
\begin{equation}\label{eqn:Lindblad_anthony}
\begin{aligned}
 \mc{L}^{\dag}(\rho) &=-i[G,\rho]+\sum_{a\in\mathcal{A}}\int^\infty_{-\infty}\gamma(\omega)\left(\widehat{A}_f^a(\omega)\rho \left(\widehat{A}_f^a(\omega)\right)^\dagger-\frac{1}{2}\left\{\left(\widehat{A}_f^a(\omega)\right)^\dagger \widehat{A}_f^a(\omega),\rho\right\}\right)\ud \omega \\ 
&=-i[G,\rho]+\sum_{a\in \mathcal{A}}\sum_{\nu,\nu'\in B_H}\alpha_{\nu,\nu'}\left(A^a_{\nu} \rho \left(A^a_{\nu'}\right)^\dagger-\frac{1}{2}\left\{\left(A^a_{\nu'}\right)^\dagger A^a_{\nu},\rho \right\}\right)\,,
\end{aligned}
\end{equation}
with the Gaussian filtering function \eqref{eq:gaussianfilter} for $\widehat{A}_f^a(\omega)$, the Gaussian-type transition weight function:
\begin{equation} \label{eq:gaussianweihgt}
\gamma^{(g)}(\ww) = \exp\left(- \frac{(\beta \ww + 1)^2}{2}\right)\,,
\end{equation}
or the Metropolis-type one: 
\begin{equation} \label{eq:metroweihgt}
    \gamma^{(m)}(\ww) = 
    \exp\left(- \beta \max\left(\ww + \frac{1}{2\beta}, 0\right)\right)\,,
\end{equation}
and the Hamiltonian 
\begin{equation}\label{eqn:G_anthony}
G:=\sum_{a\in \mc{A}}\sum_{\nu,\nu'\in B_H}\frac{\tanh(-\beta(\nu-\nu')/4)}{2i}\alpha_{\nu,\nu'}\left(A^a_{\nu'}\right)^\dagger A^a_{\nu}\,.
\end{equation}
Here, the coefficients $\alpha_{\nu,\nu'} \in \C$ are given by 
\begin{equation} \label{eq:alphanunu}
    \alpha_{\nu,\nu'} := (2\pi)^2\int^\infty_{-\infty} \gamma(\omega)\widehat{f}(\omega-\nu) \overline{\widehat{f}(\omega-\nu')} 
 \ud \omega\,.
\end{equation}
We then define the so-called Kossakowski matrix ${\bf C} := (\alpha_{\nu,\nu'})_{\nu,\nu' \in B_H}$ \cite{GoriniKossakowskiSudarshan1976}, which is a real and positive semidefinite matrix by choosing $\hat{f}(\omega)$, $\gamma(\omega)$ to be non-negative functions. Then \cite{ChenKastoryanoGilyen2023} showed that if $\si_E = 1/\beta$, there holds 
\begin{equation}\label{eqn:detail_alpha_anthony}
\alpha_{\nu,\nu'} e^{\beta (\nu + \nu')/4} = \alpha_{-\nu',-\nu} e^{-\beta (\nu + \nu')/4}\,,
\end{equation}
which implies the KMS detailed balance of $\mc{L}$ constructed in \eqref{eqn:Lindblad_anthony} \cite[Theorem I.1]{ChenKastoryanoGilyen2023}. 

To proceed, for notational simplicity, we assume $|\mathcal{A}|=1$, which means that $\{A^a\}_{a \in \mc{A}}$ is a single self-adjoint operator $A=A^\dagger$. We introduce a new coefficient matrix: 
\begin{equation} \label{auxeqq:c}
    \wt{{\bf C}} \in \mathbb{R}^{|B_H|\times |B_H|} \q \text{with}\q \wt{{\bf C}}_{\nu,\nu'}: = \alpha_{\nu,\nu'} e^{\beta (\nu + \nu')/4}\,,
\end{equation}
which is real and positive semidefinite and satisfy the centrosymmetry $\wt{{\bf C}}_{\nu,\nu'} = \wt{{\bf C}}_{-\nu',-\nu}$ by \eqref{eqn:detail_alpha_anthony}. We then consider the eigendecomposition of $\wt{{\bf C}}$: 
\begin{align} \label{auxeqq:c2}
    \wt{{\bf C}} = Q D Q^\dagger\,,
\end{align}
where $Q$ is real orthogonal (hencd $Q^{\dag}=Q^{\top}$) and $D$ is real diagonal with elements also indexed by $(\nu, \nu') \in B_H \times B_H$. Moreover, by \cite[Theorem 2]{cantoni1976properties}, each eigenvector $Q_{\dd, \nu'}$ is either symmetric (namely, $Q_{\nu, \nu'} = Q_{-\nu, \nu'}$) or skew-symmetric (namely, $Q_{\nu, \nu'} = -Q_{-\nu, \nu'}$). For $\nu' \in B_H$, we define
\begin{equation} \label{eqn:L_A}
L_{\nu'} = \begin{dcases}
    \sqrt{D_{\nu', \nu'}} \sum_{\nu \in B_H} Q_{\nu, \nu'} A_{\nu}\ e^{-\beta \nu/4}\,, &  \text{if $Q_{\dd,\nu'}$ is symmetric}\,, \\
    i \sqrt{D_{\nu', \nu'}} \sum_{\nu \in B_H} Q_{\nu,\nu'}A_{\nu}e^{-\beta \nu/4}\,, &  \text{if $Q_{\dd,\nu'}$ is skew-symmetric}\,.
\end{dcases} 
\end{equation}
Then, \cref{eqn:Lindblad_anthony} can be reformulated as 
\begin{align*}
    \mc{L}^\dag =-i[G,\rho]+ \sum_{\nu \in B_H} L_\nu \rho L^\dagger_\nu - \frac{1}{2}\left\{L^\dagger_\nu L_\nu, \rho\right\}\,,
\end{align*}
and one can verify that $G$ and $\{L_\nu\}_{\nu \in B_H}$ satisfy the requirements in \cref{prop:KMS_DBC}. 

For this, let us first consider $G$ defined in \eqref{eqn:G_anthony}. Noting from the construction \eqref{eqn:L_A} that 
\begin{align*}
\log\left(\Delta^{1/4}_{\sigma_\beta}\right)\left(\sum_{\nu \in B_H} L^\dagger_\nu L_\nu \right)=\sum_{\nu,\nu'\in B_H}-\frac{\beta(\nu-\nu')}{4}\alpha_{\nu,\nu'}\left(A^a_{\nu'}\right)^\dagger A^a_{\nu}\,,
\end{align*}
we compute, according to \eqref{eqn:G_anthony}, 
\begin{equation}\label{eqn:B_standard_anthony}
G = \sum_{\nu,\nu'\in B_H}\frac{\tanh(\beta(\nu'-\nu)/4)}{2i}\alpha_{\nu,\nu'}\left(A^a_{\nu'}\right)^\dagger A^a_{\nu}= -\frac{i}{2}\tanh\circ \log\left(\Delta^{1/4}_\beta\right)\left(\sum_{\nu \in B_H} L^\dagger_\nu L_\nu \right)\,,
\end{equation}
which matches the general form \eqref{eqn:G_formula_KMS}. We next check the condition \eqref{eq:adjointlj}. We start with the case where $Q_{\dd,\nu'}$ is symmetric. By \eqref{eq:rela2} and \eqref{eqn:L_A}, it holds that 
\begin{align*}
     \Delta^{-1/2}_{\sigma_\beta}(L_{\nu'}) = &  \sqrt{D_{\nu', \nu'}} \sum_{\nu \in B_H} Q_{\nu, \nu'} A_{\nu}\ e^{-\beta \nu/4} e^{\beta \nu/2} \\
    = & \sqrt{D_{\nu', \nu'}} \sum_{\nu \in B_H} Q_{\nu, \nu'} (A_{\nu})^{\dag}\ e^{-\beta \nu/4} = L^\dagger_{\nu'}\,,
\end{align*}
where we use $(A_\nu)^{\dag} = A_{-\nu}$ and $Q_{\nu,\nu'} = Q_{-\nu,\nu'}$ in the second equality. For the case where $Q_{\dd,\nu'}$ is skew-symmetric, a similar computation gives 
\begin{align*}
     \Delta^{-1/2}_{\sigma_\beta}(L_{\nu'}) = &  i\sqrt{D_{\nu', \nu'}} \sum_{\nu \in B_H} Q_{\nu, \nu'} A_{\nu}\ e^{-\beta \nu/4} e^{\beta \nu/2} \\
    = & i \sqrt{D_{\nu', \nu'}} \sum_{\nu \in B_H} Q_{-\nu, \nu'} A_{-\nu}\ e^{-\beta \nu/4} \\
     = - & i \sqrt{D_{\nu', \nu'}} \sum_{\nu \in B_H} Q_{\nu, \nu'} (A_{\nu})^{\dag}\ e^{-\beta \nu/4} = L^\dagger_{\nu'}\,,
\end{align*}
where the third equality is by $(A_\nu)^{\dag} = A_{-\nu}$ and $Q_{\nu,\nu'} = - Q_{-\nu,\nu'}$. 

Finally, to see that our construction in \cref{sec:gibbsampler} recovers the quantum Gibbs sampler in \cite{ChenKastoryanoGilyen2023}, it suffices to define the weighting function:
\begin{equation}  \label{auxeqq:q}
q_{\nu'}(\nu) = \begin{dcases}
    \sqrt{D_{\nu', \nu'}}  Q_{\nu, \nu'} \,, &  \text{if $Q_{\dd,\nu'}$ is symmetric}\,, \\
    i \sqrt{D_{\nu', \nu'}}  Q_{\nu,\nu'} \,, &  \text{if $Q_{\dd,\nu'}$ is skew-symmetric}\,,
\end{dcases} 
\end{equation}
and find $q_{\nu'}(-\nu) = \overline{q_{\nu'}(\nu)}$ as required in \eqref{eqn:q_sym} and that the jumps in \eqref{eqn:L_A} are the same as those in \eqref{eqn:L_a_formula} with $A$ and $q_{\nu'}(\nu)$ given above. 

\begin{rem} \label{rem:gnsdbc}
Letting $\widehat{f}(\ww) \propto \exp(-\ww^2/4 \sigma^2_E)$ be given as above, choosing $\gamma(\nu)$ to satisfy the KMS condition \eqref{eq:rela1}, in the limit $\si_E \to 0$, the matrix $\{\alpha_{\nu,\nu'}\}$ in \eqref{eq:alphanunu} reduces to $\{\gamma(\nu)\d_{\nu,\nu'}\}$ (up to some constant) and the associated Lindblad dynamic becomes GNS detailed balanced. One can similarly define the matrix $\w{{\bf C}}$ as in \eqref{auxeqq:c} with decomposition \eqref{auxeqq:c2}. It holds that for each $\nu \in B_H$, there exists a symmetric eigenvector $Q_{\nu,\nu'}=\frac{1}{\sqrt{2}}(\d_{\nu'} + \d_{-\nu'})$ and a skew-symmetric one $Q_{\nu,\nu'}=\frac{1}{\sqrt{2}}(\d_{\nu'} - \d_{-\nu'})$, which corresponds to the eigenvalue $D_{\nu, \nu} = D_{-\nu,-\nu} \propto \gamma(\nu)e^{\beta \nu/2}$. In this case, $q_{\nu'}(\nu)$ in \eqref{auxeqq:q} is either $\sqrt{D_{\nu',\nu'}/2}(\d_{\nu'} + \d_{-\nu'})$ or $i \sqrt{D_{\nu',\nu'}/2}(\d_{\nu'} - \d_{-\nu'})$, and the jumps defined in \eqref{eqn:L_A} are consistent with those in \eqref{auxeqq:Lj}. 
\end{rem}

We now estimate the element distribution of the Kossakowski matrix $(\alpha_{\nu,\nu'})_{\nu,\nu' \in B_H}$,
which in principle determines the KMS detailed balanced Lindbladian, in view of  \cref{eqn:Lindblad_anthony,eqn:G_anthony}. 
This would help us better understand the effects of the choice of $q$ on the energy transition and how our proposed Gibbs sampler relates to those in \cite{ChenKastoryanoBrandaoEtAl2023,ChenKastoryanoGilyen2023}. 
Recall the definition of $\alpha_{\nu,\nu'}$ in \cref{eq:alphanunu} and note that $\widehat{f}(\nu)$ was chosen as a Gaussian $\widehat{f}(\nu) \propto \sqrt{\beta}\exp(- (\beta\nu)^2/4)$. It follows that for any $L^\infty$-bounded $\gamma(\ww)$, 
\begin{align*}
    |\alpha_{\nu,\nu'}| \le C \norm{\gamma}_{L^\infty(\R)} \beta\int^\infty_{-\infty} 
    e^{-\frac{\beta^2 ((\ww - \nu)^2 + (\ww - \nu')^2) }{4}}\mathrm{d}\ww \le C \norm{\gamma}_{L^\infty(\R)} e^{-\frac{{\beta^2(\nu-\nu')^2}}{8}}\,,
\end{align*}
where $C$ is a uniform constant. This means that for any fixed $\beta > 0$, 
\begin{align} \label{eq:entryestimate}
    |\alpha_{\nu,\nu'}| = \Omega(1) \q \text{only if} \q |\nu - \nu'| = \Or\left(\frac{1}{\beta}\right)\,,
\end{align}
that is, $\alpha_{\nu,\nu'}$ is concentrated around the diagonal part, which is the case for the Metropolis-type transition weight \eqref{eq:metroweihgt}, due to $\gamma^{(m)} \equiv 1$ for $\ww \le 1/(2 \beta)$. When $\beta \to \infty$, this narrow strip shrinks rapidly and the matrix $\alpha_{\nu,\nu'}$ approximately reduces to a diagonal one so that the sampler becomes GNS detailed balanced (\cref{rem:gnsdbc}). For the case of Gaussian transition weight \eqref{eq:gaussianweihgt}, we can see that the $\alpha_{\nu,\nu'}$ is actually concentrated around the origin:
\begin{equation*}
    \alpha_{\nu,\nu'} \propto e^{-\frac{(\beta \nu + \beta \nu' + 2)^2}{16}} e^{-\frac{\beta^2(\nu-\nu')^2}{8}} = \Omega(1) \q \text{if and only if}\q |\nu|, |\nu'| = \Or\left(\frac{1}{\beta}\right)\,.
\end{equation*}
by the explicit computation in \cite[Proposition II.3]{ChenKastoryanoGilyen2023}. 
In contrast, for our Gibbs sampler constructed in \cref{sec:gibbsampler}, we have
\begin{equation*}
    \alpha_{\nu,\nu'} = e^{-\beta(\nu +\nu')/4} q^a(\nu)\overline{q^a(\nu')}\,,
\end{equation*}
which is always supported on $[-S,S]^2$ independent of $\beta$ by \cref{assumption:q}. We refer the readers to \cref{fig:matr_dist} below for an illustration of the pattern of $\alpha_{\nu,\nu'}$ for various Gibbs samplers. 


\begin{rem}\label{rem:difference} 
The discussion above shows that when $|\nu-\nu'|=\Omega(1/\beta)$, $\alpha_{\nu,\nu'}\ll 1$ (see \cref{eq:entryestimate}). Consequently, in \eqref{eqn:Lindblad_anthony}, the energy transition $\nu$ in $A_{\nu}$ is always $\mathcal{O}(1/\beta)$ close to the energy transition $\nu'$ in $(A^a_{\nu'})^{\dagger}$. 
This differs from our case, where $\{\alpha_{\nu,\nu'}\}$ always includes a $\Omega(1)$-sized principal submatrix and the dynamics allows different energy transition terms $A_{\nu}\rho (A^a_{\nu'})^{\dagger}$ even when $|\nu-\nu'|\gg \Omega(1/\beta)$. 
\end{rem}

\begin{figure}[!htbp]
     \subfloat{
         \centering
         \includegraphics[width=0.7\textwidth]{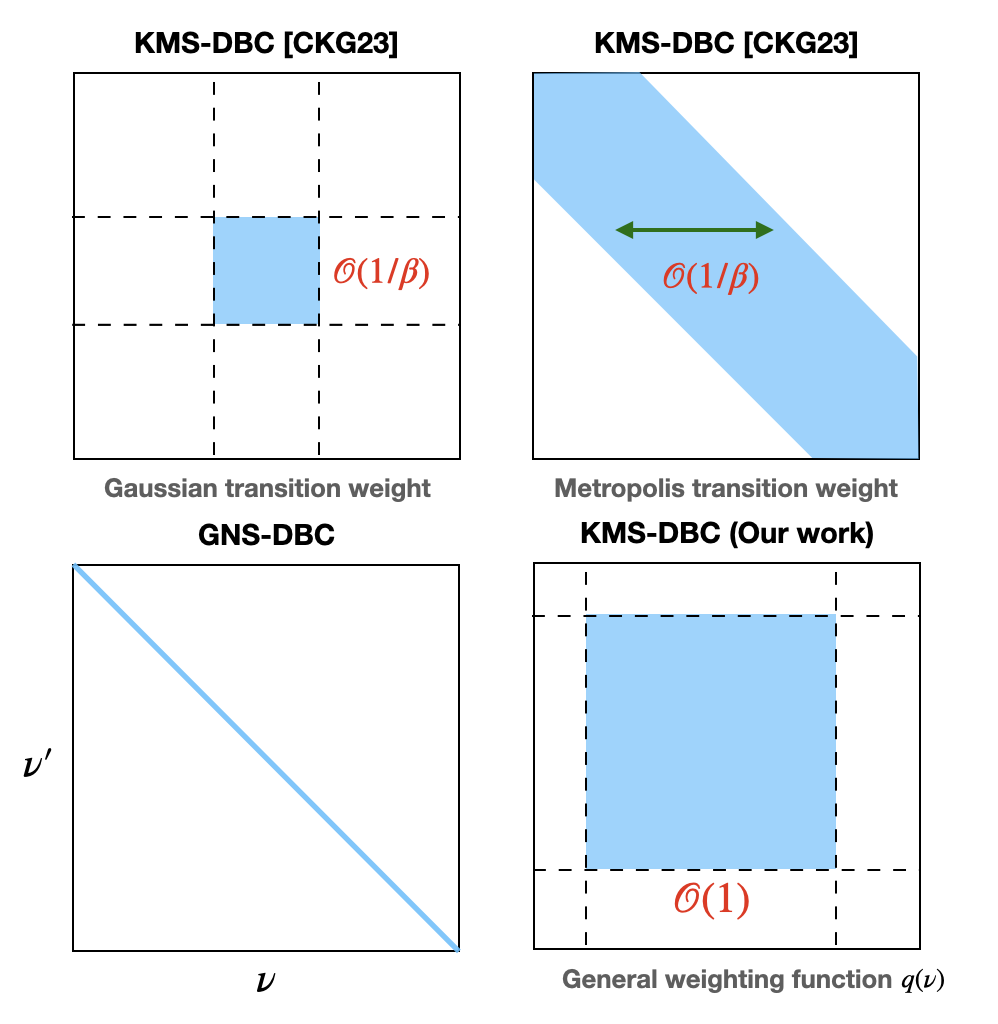}
     }
     \caption{The matrix element distribution of the Kossakowski matrix $(\alpha_{\nu,\nu'})_{\nu,\nu' \in B_H}$ associated with a coupling $A$ for various detailed balanced quantum Gibbs samplers. The blue shadow region indicates the dominant entries.}\label{fig:matr_dist}
\end{figure}

\subsection*{Conflict of interest statment}

Authors have no conflict of interest to declare.

\newcommand{\etalchar}[1]{$^{#1}$}

\appendix

\section{Block encoding}\label{app:block_encoding}

Block encoding (see \cite{LowChuang2019,GilyenSuLowEtAl2019}) provides a general framework for encoding a non-unitary matrix using unitary matrices, which can be implemented on quantum devices. 

\begin{defn}[Block encoding] Given a matrix $A\in\CC^{2^n\times 2^n}$, if we can find $\alpha, \epsilon \in \mathbb{R}_+$, and a unitary matrix $U_A\in\CC^{2^{n+m}\times 2^{n+m}}$ so that 
\begin{equation}
\Vert A - \alpha \left(\langle 0^m | \otimes I_n\right) U_A \left( | 0^m \rangle \otimes I_n \right) \Vert \leq \epsilon,
\end{equation}
then $U_A$ is called an $(\alpha, m, \epsilon)$-block-encoding of $A$. The parameter $\alpha$ is referred to as the block encoding factor, or the subnormalization factor. 
\label{def:blockencode}
\end{defn}

Intuitively, the block encoding matrix $U_A$ encodes the rescaled matrix $A/\alpha$ in its upper left block:
\begin{equation*}
    U_A \approx \left( \begin{array}{cc}
        A/\alpha & * \\
        * & *
    \end{array} \right). 
\end{equation*}
There have been substantial efforts on implementing the block encoding of certain structured matrices of practical interest~\cite{GilyenSuLowEtAl2019,NguyenKianiLloyd2022,CampsLinVanBeeumenEtAl2024,sunderhauf2023blockencoding}. 
In this work, we assume that the query access to the block encoding of relevant matrices is available. We also assume that there is no error in the block encodings of the input matrices $H$, $A^a$, etc. If such errors are present, their impact should be treated using perturbation theories. Meanwhile, we need to carefully keep track of the error of the block encodings derived from the input matrices, such as the jump operators $L_a$.

\section{Linear combination of unitaries}\label{app:LCU}

The linear combination of unitaries (LCU)~\cite{ChildsWiebe2012} is an important quantum primitive, which allows matrices expressed as a superposition of unitary matrices to be coherently implemented using block encoding. Here we follow~\cite{GilyenSuLowEtAl2019} and present a general version of the LCU that is applicable to possibly complex coefficients. 

LCU implements a block encoding of $\sum_{j=0}^{J-1} c_j U_j$, where $U_j$ are unitary operators and $c_j$ are complex numbers. 
For the coefficients, we assume access to a pair of \emph{state preparation oracles} (also called prepare oracles for short) $(\textbf{Prep}_{\overline{\gamma_l}}, \textbf{Prep}_{\gamma_r})$ acting as (assume $J=2^\ell$)
\begin{align*}
    \mathbf{Prep}_{\overline{\gamma_l}}&: \ket{0^{\ell}} \rightarrow \frac{1}{\norm{\gamma_l}_2} \sum_{j=0}^{J-1} \overline{\gamma_{l,j}} \ket{j}, \\
    \mathbf{Prep}_{\gamma_r}&: \ket{0^{\ell}} \rightarrow \frac{1}{\norm{\gamma_r}_2} \sum_{j=0}^{J-1} \gamma_{r,j} \ket{j}. 
\end{align*}
Here the coefficients should satisfy $\gamma_{l,j}\gamma_{r,j}=c_j,0\le j\le J-1$.
To minimize the block encoding factor, the optimal choice is $\abs{\gamma_{l,j}}=\abs{\gamma_{r,j}}=\abs{\sqrt{c_j}}$, and where $\sqrt{z}$ refers to the principal value of the square root of $z$. In this case, $\norm{\gamma_l}_2=\norm{\gamma_r}_2=\sqrt{\norm{c}_1}$, where $\norm{c}_1=\sum_{j} \abs{c_j}$ is the $1$-norm of the vector $c$.

The unitaries $U_j$ needs to be accessed via a \emph{select oracle} $\textbf{Select}$ as 
\begin{equation*}
    \mathbf{Select} = \sum_{j=0}^{J-1} \ket{j}\bra{j} \otimes U_j,
\end{equation*}
which can be constructed using controlled versions of the block encoding matrices $U_j$. 

\begin{lem}[LCU]
    Assume $J=2^\ell$, $\norm{\gamma_l}_2=\norm{\gamma_r}_2=\sqrt{\norm{c}_1}$, then the matrix $$W=(\mathbf{Prep}_{\overline{\gamma_l}}^{\dagger} \otimes I) \;\mathbf{Select}\; (\mathbf{Prep}_{\gamma_r}\otimes I)$$ is a $(\norm{c}_1,\ell,0)$-block-encoding of the linear combination of unitaries $\sum_{j=0}^{J-1} c_j U_j$. 
\end{lem}

\section{Gevrey functions}\label{sec:gevrey}

In this section, we collect a few results for Gevrey functions that are useful for this
work. While similar findings have been previously demonstrated in the
literature, notably in \cite{Adwan_2017} and \cite{Gus_2019}, we
provide self-contained proofs here with explicit expressions for the constants involved. 


\begin{lem}[Product of Gevrey functions]\label{lem:product_gevrey} 
Given $h\in\mathcal{G}^{s}_{C_1,C_2}(\mathbb{R}^d)$ and $h'\in\mathcal{G}^{s'}_{C'_1,C'_2}(\mathbb{R}^d)$, then
\[
h\cdot h'\in \mathcal{G}^{\max\{s,s'\}}_{C_1C'_1,C_2+C'_2}(\mathbb{R}^d)\,.
\]
\end{lem}

\begin{proof} For a index vector $\alpha \in \mathbb{N}^d$, a direct application of Leibniz rule gives 
\begin{align*}
\left\|\pa^\alpha \left(h\cdot h'\right)\right\|_{L^\infty(\mathbb{R}^d)} &\le \sum^{|\alpha|}_{j=0}\sum_{\beta\leq \alpha,|\beta|=j}\binom{\alpha}{\beta}\left\|\pa^\beta h \right\|_{L^\infty(\mathbb{R}^d)}\left\|\pa^{\alpha- \beta} h'\right\|_{L^\infty(\mathbb{R}^d)} \\
& \leq C_1C'_1 \sum^{|\alpha|}_{j=0}\left( (C_2)^{j} j^{js}(C'_2)^{|\alpha|-j}(|\alpha|-j)^{(|\alpha|-j)s'}\left(\sum_{\beta\leq \alpha,|\beta|=j}\binom{\alpha}{\beta}\right)\right) \\
& \leq C_1C'_1 \sum^{|\alpha|}_{j=0}\left(\binom{|\alpha|}{j}(C_2)^{j} j^{js}(C'_2)^{|\alpha|-j}(|\alpha|-j)^{(|\alpha|-j)s'}\right)\\
& \leq C_1C'_1 \left(C_2+C'_2\right)^{|\alpha|}|\alpha|^{|\alpha|\max\{s,s'\}}\,,
\end{align*}
where the third inequality is by  $\sum_{\beta\leq \alpha,|\beta|=j}\binom{\alpha}{\beta} = \binom{|\alpha|}{j}$.  
\end{proof}

We next show that $\exp(-\frac{\sqrt{1+x^2}+x}{4})$
is a Gevrey function and its derivative is $L^1$-integrable. We first recall the Fa\`a di Bruno's formula and the partial Bell polynomial for the chain rule of high-order derivatives \cite[p.139]{comtet1974advanced}. 

\begin{lem}[Fa\`a di Bruno's formula and partial Bell polynomial]
\label{lem:fa}Let $h,g$ be smooth function from $\mathbb{C}$ to $\mathbb{C}$, and $f(x) := h(g(x))$. Then the $k$-th order derivative of $f$ is given by
\[
f^{(k)}(x)=\sum\frac{k!}{q_1!(1!)^{q_1}q_2!(2!)^{q_2}\cdots q_k!(k!)^{q_k}} h^{(\sum q_i)}(g(x))\prod^k_{i=1}\left(g^{(i)}(x)\right)^{q_i}
\]
where the sum is over all $k$-tuples of nonnegative integers $(q_1,q_2,\cdots, q_k)$ satisfying $\sum iq_i=k$. 
The above formula can also be rewritten as
\[
f^{(k)}(x)=\sum^k_{j=1}h^{(j)}(g(x))B_{k,j}\left(g^{(1)}(x),g^{(2)}(x),\dots,g^{(k-j+1)}(s)\right)\,,
\]
where $B_{k,j}$ is the partial Bell polynomial: 
\re{\[
B_{k,j}(z_1,z_2,\dots,z_{k-j+1})= \sum_{\substack{1\leq i\leq k\,,\, q_i\in\mathbb{N}\\\sum^k_{i=1}iq_i=k\\\sum^k_{i=1}q_i=j}} \frac{k!}{q_1!q_2!\cdots q_{k-j+1}!}\prod^{k-j+1}_{i=1}\left(\frac{z_i}{i!}\right)^{q_i}\,.
\]}
\end{lem}

Using Fa\`a di Bruno's formula, we can calculate the high order derivatives of $\exp(-\frac{\sqrt{1+x^2}+x}{4})$ and verify that it belongs to a Gevrey class.

\begin{lem}\label{prop:exp_gevrey} 
It holds that 
\begin{equation*}
    e^{-\frac{\sqrt{1+x^2}+x}{4}} \in \mathcal{G}^1_{1,\frac{7}{2}}\,,\q   \left(e^{-\frac{\sqrt{1+x^2}+x}{4}}\right)^{(1)} \in L^1(\R)\,.
\end{equation*}
\end{lem}
\begin{proof}

We first use the second formula of \cref{lem:fa} and some properties of the partial Bell polynomials $B_{n,k}$ to calculate $k$-th order derivative for $\sqrt{1+x^2}$. For $k\geq 2$, we have (define $\binom{j}{k-j}=0$ if $k>2j$)
\small
\[
\begin{aligned}
&\left(\sqrt{1+x^2}\right)^{(k)}\\
=&\sum^k_{j=1}\frac{(-1)^{j+1}(2j-3)!!}{2^{j}(1+x^2)^{j-1/2}}B_{k,j}(2x,2,0,\dots,0)=\sum^k_{j=1}\frac{(-1)^{j+1}2^{j}(2j-3)!!}{2^{j}(1+x^2)^{j-1/2}}B_{k,j}(x,1,0,\dots,0)\\
=&\sum^k_{j=1}\frac{(-1)^{j+1}2^{j}(2j-3)!!}{2^{j}(1+x^2)^{j-1/2}} \frac{1}{2^{k-j}}\frac{k!}{j!}\binom{j}{k-j}x^{2j-k}=\frac{k!}{2^k}\sum^k_{j=1}\frac{(-1)^{j+1}2^j(2j-3)!!}{j!}\binom{j}{k-j}\frac{x^{2j-k}}{(1+x^2)^{j-1/2}}\,.
\end{aligned}
\]
\normalsize
Using the fact that $\left|\frac{x^{2j-k}}{(1+x^2)^{j-1/2}}\right|\leq 1$, $(2j-3)!!\leq 2^j j!$, and $(1+2^2)^k=\sum^k_{j=0}\binom{k}{k-j} 2^{2j}$, we have 

\[
\left|\left(\sqrt{1+x^2}\right)^{(k)}\right|\leq \frac{k!}{2^k}\sum^k_{j=0}\frac{2^j(2j-3)!!}{j!}\binom{j}{k-j}\leq \frac{k!}{2^k}\sum^k_{j=0} 4^j\binom{k}{k-j}\leq \left(5/2\right)^kk!\,.
\]

Next, by Fa\`a di Bruno's formula for $\exp(-\frac{\sqrt{1+x^2}+x}{4})$, we obtain
\begin{equation}\label{eq:derivative}
\begin{aligned}
      &\left(e^{-\frac{\sqrt{1+x^2}+x}{4}}\right)^{(k)} \\
= &\sum\frac{k!}{q_1!(1!)^{q_1}q_2!(2!)^{q_2}\cdots q_k!(k!)^{q_k}}\left(-\frac{1}{4}\right)^ke^{-\frac{\sqrt{1+x^2}+x}{4}}\prod^k_{j=1}\left(\left(\sqrt{1+x^2}\right)^{(j)}+x^{(j)}\right)^{q_j}\,.
\end{aligned}
\end{equation}
 Plugging the upper bound $\left|\left(\sqrt{1+x^2}\right)^{(j)}\right|$ gives
\begin{align*}
    &\left|\left(e^{-\frac{\sqrt{1+x^2}+x}{4}}\right)^{(k)}\right| < \sum\frac{k!}{q_1!(1!)^{q_1}q_2!(2!)^{q_2}\cdots q_k!(k!)^{q_k}}\left(\frac{1}{4}\right)^k\prod^k_{j=1} (7/2)^{jq_j}(j!)^{q_j}\\
=&\sum\frac{k!}{q_1!q_2!\cdots q_k!}\left(\frac{1}{4}\right)^k\prod^k_{j=1}(7/2)^{jq_j}=\sum\frac{(7/8)^kk!}{q_1!q_2!\cdots q_k!}\leq (7/8)^kk!\sum 1\leq (7/2)^kk!\,,
\end{align*}
where we use the fact that the number of $k$-tuples of nonnegative integers $(q_1,q_2,\cdots, q_k)$ satisfying $\sum jq_j=k$ is less than $\binom{2k}{k}\leq 2^{2k}$. This concludes $\exp(-\frac{\sqrt{1+x^2}+x}{4}) \in \mc{G}_{1,\frac{7}{2}}^1$. The $L^1$-integrablity of $(\exp(-\frac{\sqrt{1+x^2}+x}{4}))^{(1)}$ is a simple consequence of \cref{eq:derivative}.  
\end{proof}

Finally, we show that the Fourier transform of the Gevrey class with compact support decays rapidly, by a Paley-Wiener type estimate.
\begin{lem}\label{lem:hat_Gev_decay}
Given $h\in\mathcal{G}^s_{C_1,C_2}(\mathbb{R}^d)$ with compact support $\Omega=\mathrm{supp}(h)$ and $s\geq 1$, define 
\begin{equation*}
    H(y)=\frac{1}{(2\pi)^d}\int_{\mathbb{R}^d}h(x)e^{-ix\cdot y} \ud x\,.
\end{equation*}
Then, for any $y\in\mathbb{R}^d$, there holds
\[
\left|H(y)\right|\leq \frac{C_1|\Omega|}{(2\pi)^d} \, e^{\frac{esd}{2}-\frac{s}{C^{1/s}_2e} \norm{y}_2^{1/s}}\,,
\]
where $|\Omega|=\int_{\Omega}1\,\mathrm{d}x$ is the volume of $\Omega$ and $\norm{y}_2$ is the $2$-norm of the vector $y$. 
\end{lem}

\begin{proof} 
By \cref{def:gevrey}, we have, for every $d$-tuple of nonnegative integers $\alpha$ with $|\alpha|=\sum_i \abs{\alpha_i}$,
\[
\left\| \pa^\alpha h \right\|_{L^\infty(\R^d)}\leq C_1C^{|\alpha|}_2|\alpha|^{|\alpha|s}\,.
\]
It follows that 
\[
\left|y^{\alpha}\right|\left|H(y)\right|=\left|y^{\alpha}H(y)\right|=\left|\frac{1}{(2\pi)^d}\int_{\Omega}\pa^\alpha h(x) e^{-ix\cdot y} \ud x\right|\leq \frac{C_1|\Omega|}{(2\pi)^d}C^{|\alpha|}_2|\alpha|^{|\alpha|s}\,,
\]
where $y^{\alpha} := \Pi^d_{i=1}y^{\alpha_i}_i$.  Recall from \cite[Proposition 3.1]{Adwan_2017} that 
\[
\inf_{m\in \mathbb{Z}_{\geq 0}}\left\{\left(\frac{s}{ae}\right)^{ms}\frac{m^{ms}}{|t|^m}\right\}\leq e^{es/2}e^{-a|t|^{1/s}}\,, \q \text{for any $a, s, t > 0$}\,.
\]
Letting $a := s/(C^{1/s}_2e)$ and using the above inequality, we obtain
\begin{align*}
    \left|H(y)\right|\leq \frac{C_1|\Omega|}{(2\pi)^d}\inf_{\alpha\in \mathbb{N}^d}\frac{C^{|\alpha|}_2|\alpha|^{|\alpha|s}}{|y^\alpha|}\leq &\frac{C_1|\Omega|}{(2\pi)^d}\prod_{i}\inf_{\alpha_i\in\mathbb{Z}_{\geq0}}\left|\frac{C^{\alpha_i}_2|\alpha_i|^{\alpha_is}}{|y_i|^{\alpha_i}}\right|\\
\leq &\frac{C_1|\Omega|}{(2\pi)^d}e^{esd/2}e^{-a\sum^d_{i=1}|y_i|^{1/s}} \\
\leq &\frac{C_1|\Omega|}{(2\pi)^d}e^{esd/2}e^{-a \norm{y}_2^{1/s}}\,. \qedhere 
\end{align*}
\end{proof}

\section{Error analysis and complexity}\label{sec:discretization}

In this section, for notational simplicity, we often absorb the generic constant $C_s$ depending on the parameter $s$ of the weighing function $q(\nu)$
in \cref{eq:relaq2} into $\Or$. 
We first study the decaying and integrable properties of functions \re{$f$ and $g$ involved in the jumps and the coherent term, respectively.} 


\begin{lem}\label{lem:property_f_g} 
Let \re{$f(t)$ and $g(t)$} be the functions defined in \eqref{eqn:f_a} and \eqref{eqn:gk}-\eqref{eqn:g_hat} with \re{$q(\nu)$ and $\kappa(\nu)$ satisfying \cref{assumption:q}} and \cref{assp:kappa}. Then, it holds that
\begin{equation} \label{eq:decayf}
    |f(t)|\leq \frac{\xi^2_qS}{\pi}\exp\left(\frac{es}{2}-\frac{s}{e}\left(\frac{|t|}{\beta \xi_u+\xi_w/S}\right)^{1/s}\right)\,,
\end{equation}
and
\begin{equation} \label{eq:decayg}
   \left|g(t)\right|\leq \frac{2 \xi_q S}{\pi}\exp\left(\frac{es}{2}-\frac{s}{e}\left(\frac{|t|}{\beta+\xi_w/4S}\right)^{1/s}\right)\,.
\end{equation}
The integrals of $f$ and $g$ have the following asymptotics: 
\begin{equation} \label{eq:inttf}
\int_\R |f(t)| \ud t  = \mathcal{O}\left( (\xi_q C_{1,u} + \xi^2_q \xi_w) \log ((\beta \xi_u  + \xi_w/S) \max\{S,1\})\right) \,,
\end{equation}
and
\begin{equation} \label{eq:inttg}
\int_\R |g(t)| \ud t = \Or\left(\xi_q(1+\log((\beta +\xi_w/S)\max\{S,1\}))\right)\,. 
\end{equation}
\end{lem}

\begin{proof}  We recall the definitions of \re{$f(t)$ and $g(t)$}:
\begin{equation} \label{def:ft}
f(t)= \frac{1}{2\pi} \int_\R  \re{u(\beta \nu)w(\nu/S)e^{-\beta \nu/4}} e^{-it\nu}\, \mathrm{d} \nu\,,    
\end{equation}
and
\begin{equation}\label{def:gtt}
  \re{g(t)
= \frac{1}{2\pi}\int_\R -\frac{i}{2} \tanh\left(- \beta\nu/4\right) w(\nu/4S) e^{-i\nu t}\, \mathrm{d}\nu \, \mathrm{d}\nu'\,.}
\end{equation}

\medskip 

\noindent {\bf Exponential decay of $f$.}
Thanks to \cref{lem:product_gevrey} with \cref{assumption:q}, we have 
\begin{equation*}
\widehat{f}(\nu) = u(\beta\nu)w(\nu/S)e^{-\beta\nu/4}\in \re{\mathcal{G}^s_{\xi^2_q,\beta \xi_u+\xi_w/S}(\mathbb{R})} \q \text{with}\q \mathrm{supp}\big(\widehat{f}\big) \in [-  S, S ]\,.   
\end{equation*}
Then, \cref{lem:hat_Gev_decay} with $d=1$ and $|\Omega|=2 S$ yields the estimate for $f$: 
\begin{equation}\label{eq:expf}
\left|f(t)\right| 
\leq \frac{\xi^2_qS}{\pi}\exp\left(\frac{es}{2}-\frac{s}{e}\left(\frac{|t|}{\re{\beta \xi_u+\xi_w/S}}\right)^{1/s}\right)\,.     
\end{equation}


\medskip 

\noindent {\bf Exponential decay of $g$.} 
By using  \cite[Eq.\,(3.3)]{tanh_2007}, we first have
\begin{equation*}
\left\| \tanh^{(N)}(x)\right\|_{L^\infty(\mathbb{R})}\leq 2^N\sum^N_{k=0}k!\binom{N}{k}\leq 4^{N}N!\,,
\end{equation*}
which implies 
\begin{equation} \label{auxeq:tanh}
\re{\tanh\left(-\nu\beta/4\right)\in \mathcal{G}^1_{1,\beta}(\mathbb{R})\,.} 
\end{equation}
It follows from \cref{lem:product_gevrey} that
\begin{align} \label{auxeqq1}
\re{(2i)\widehat{g}(\nu) := \tanh\left(-\nu\beta/4\right)w\left(\nu/4S\right)\in \mathcal{G}^{s}_{\xi_q,\beta+\xi_w/4 S}\left(\mathbb{R}\right),\quad \mathrm{supp}(\widehat{g})=[-4S , 4S]\,.}
\end{align}
Again, applying \cref{lem:hat_Gev_decay} with $d=1$ and $|\Omega|=8S$ gives
\begin{equation}\label{eq:expg}
\re{\left|g(t)\right|\leq \frac{2 \xi_q S}{\pi}\exp\left(\frac{es}{2}-\frac{s}{e}\left(\frac{|t|}{\beta+\xi_w/4S}\right)^{1/s}\right)\,.}
\end{equation}
which concludes the proof for the decay of $g$. 


\medskip 

\noindent {\bf Estimate the integral of $f$.} 
We note from \cref{assumption:q} that the derivative of $u(\beta \nu)w(\nu/S)e^{-\beta \nu/4}$ is $L^1$-integrable. It follows that
\begin{equation*} 
itf(t)= \frac{1}{2\pi} \int_\R \frac{\rd}{\rd \nu} \left(u(\beta \nu)w(\nu/S)e^{-\beta \nu/4}\right) e^{-it\nu} \ud \nu\,,    
\end{equation*}
and then
\begin{equation}\label{eqn:tft_bound}
\begin{aligned}
   2 \pi \norm{t f(t)}_{L^\infty(\R)} & \le  \|(u(\beta \nu)e^{-\beta \nu/4})^{(1)}\|_{L^1(\R)} \|w(\nu/S)\|_{L^\infty(\R)} \\ &\qquad \qquad \qquad   + \|u(\beta \nu)e^{-\beta \nu/4}|_{L^\infty(\R)} \|(w(\nu/S))^{(1)}\|_{L^1(\R)} \\
             & \le \xi_q C_{1,u} + 2 \xi^2_q \xi_w\,,
\end{aligned}
\end{equation}
by the invariance of $\|(u(\beta \nu)e^{-\beta \nu/4})^{(1)}\|_{L^1(\R)}$ in $\beta$. Here, we use $w\in \mathcal{G}^s_{\xi_q,\xi_w}$ and $\mathrm{supp}(w)=[-1,1]$ to obtain $\|(w(\nu/S))^{(1)}\|_{L^1(\R)}\leq 2\xi_q\xi_w$. In addition, because $\|u(\beta \nu)e^{-\beta \nu/4}\|_{L^\infty(\R)}\leq \xi_q$, there holds
\begin{align*}
    \norm{f}_{L^\infty(\R)} \le \frac{1}{2\pi} \int_{-S}^S  \left|u(\beta \nu)w(\nu/S)e^{-\beta \nu/4}\right|\, \mathrm{d} \nu \leq \frac{1}{\pi} \xi_q^2 S\,. 
\end{align*}
Thus, we readily have, \re{for any $T > \min\{1, 1/S\}$},
\begin{equation} \label{eq:estint}
 \re{\begin{aligned}
          \int_{-T}^T |f(t)| \ud t & \le \int_{|t|\leq \min\{1, 1/S\}} |f(t)|\ud t + \int_{\min\{1, 1/S\}\leq |t|\leq T} |f(t)| \ud t \\
    & \le \frac{2}{\pi} \xi_q^2 \min\{S, 1\} + 2 \frac{\xi_q C_{1,u} + 2 \xi^2_q \xi_w}{2 \pi} \log(T \max\{S,1\}) \\ &=\mathcal{O}\left((\xi_q C_{1,u} + \xi^2_q \xi_w) \log(T \max\{S,1\})\right).
    \end{aligned}}
\end{equation}

Next, by \cref{eq:decayf}, 
a direct computation via change of variable $t^{1/s} = u$ gives 
\begin{equation} \label{auxeq:intf}
  \int_{T}^\infty \left|f(t)\right| \ud t 
=\mathcal{O}\left(\xi^2_q S  \int_{T^{1/s}}^\infty u^{s-1} \exp\left(- \frac{s}{e(\re{\beta \xi_u+\xi_w/S)^{1/s}}}u\right) \ud u\right)\,, \q T > 0\,.
\end{equation}
We define the constant $T_f$ by 
\begin{align} \label{eq:tauaw}
    T_f: = \inf \left \{T > 0\,;\ u^{s-1} \le \exp\left(\frac{s}{\re{2e(\beta \xi_u + \xi_w/S)^{1/s}}} u \right)\ \text{for any $u \ge T^{1/s}$} \right\},
\end{align}
which satisfies the following asymptotics: 
\begin{equation*} 
     T_f/\log T_f = \Theta\left(\beta \xi_u + \xi_w/S\right)\,.
\end{equation*}
Note that $x/\log x$ is decreasing on $(1,e]$ and increasing on $[e,+\infty)$ with global minimum $e$ at $x = e$. For any $y \ge e$, the equation $x/\log x = y$ has a unique solution $y \le x \le y^2$, which readily gives $ y \log y \le x = y \log x \le 2 y \log y$ and thus, by \eqref{eq:tauaw},  
\begin{align} \label{auxeq:tfest}
   \re{T_f = \Theta\left((\beta \xi_u + \xi_w/S) \log\left(\beta \xi_u + \xi_w/S\right) \right)}\,.
\end{align} 
We can then estimate the integral by \cref{eq:estint,auxeq:intf}, as well as \cref{auxeq:tfest}:\re{
\small
\begin{align} \label{auxeq:inttf}
          \int_\R \left|f(t)\right| \ud t & = \int_{-T}^{T} |f(t)| \ud t + \mathcal{O}\left(\xi^2_q S   \int_{T^{1/s}}^\infty  \exp\left(- \frac{s}{2 e (\beta \xi_u+\xi_w/S)^{1/s}}u\right) \ud u\right)\notag \\
     & =\mathcal{O}\left((\xi_q C_{1,u} + \xi^2_q \xi_w) \log(T \max\{S,1\}) + \xi^2_q \underbrace{S (\beta \xi_u+\xi_w/S)^{1/s} \exp\left(- \frac{s T^{1/s}}{2e(\beta \xi_u+\xi_w/S)^{1/s}}\right)}_{=:{\bf II}}\right) \\
     & = \mathcal{O}\left( (\xi_q C_{1,u} + \xi^2_q \xi_w) \log ((\beta \xi_u  + \xi_w/S) \max\{S,1\})\right), \notag
    \end{align}
\normalsize
where $T =\Theta\left((\beta \xi_u + \xi_w/S)\log^s\left(\left(\beta \xi_u + \xi_w/S\right)\max\{S,1\}\right) \right) > T_f$ is chosen such that 
\begin{align*}
    {\bf II} & =  S(\beta \xi_u+\xi_w/S)^{1/s} \exp\left(- \Theta(\log(\left(\beta \xi_u + \xi_w/S\right)\max\{S,1\})) \right) \\
    & = \frac{1}{(\beta \xi_u+\xi_w/S)^{\Theta(1)-1/s}} \frac{S}{\max\{S,1\}^{\Theta(1)}} \to 0 \q \text{as}\q \beta \xi_u + \xi_w/S \to \infty\,.
\end{align*}
We have proved \cref{eq:inttf}.}

\medskip

\noindent {\bf Estimate the integral of $g$.} Similarly, by \cref{def:gtt} with \cref{assumption:q} and $|\tanh(\nu/4)|\leq 1$, we have \re{
$
    \norm{g(t)}_{L^\infty(\R^2)} \le \frac{2 \xi_q S}{\pi}
$.}
A direct computation gives
\begin{equation*}
    \frac{\rd}{\rd \nu} \tanh(\nu/4) = \frac{1}{4} (\tanh^2(\nu/4) - 1) \in L^1(\R) \bigcap L^\infty(\R)\,.
\end{equation*}
Similar to \eqref{eqn:tft_bound}, this gives
\begin{align*}
    2\pi|tg(t)| &\leq \int_\R \left|\widehat{g}'(\nu)\right|\mathrm{d}\nu\\
    &\leq \frac{1}{2} \underbrace{\|\tanh(-\beta \nu/4)^{(1)}\|_{L^1(\R)}}_{=\mc{O}(1)} \underbrace{\|w(\nu/4S)\|_{L^\infty(\R)}}_{=\mc{O}(\xi_q)} + \frac{1}{2} \underbrace{\|\tanh(-\beta \nu/4)\|_{L^\infty(\R)}}_{=\mc{O}(1)} \underbrace{\|w(\nu/4S)^{(1)}\|_{L^1(\R)}}_{= \mc{O}(1)}\,,
\end{align*}
namely,
\begin{equation*}
    |g(t)| \le \mc{O}(\xi_q)/t\,.
\end{equation*}
In the same manner as \cref{eq:estint}, we have, for $T > \min\{1,1/S\}$ 
\begin{equation} \re{\label{auxeq:finiteg}
    \begin{aligned}
        \int_{-T}^T |g(t)| \ud t & =\int_{|t|\leq \min\{1,1/S\}} |g(t)| \ud t+\int_{\min\{1,1/S\} \leq |t|\leq T} |g(t)| \ud t \\ &=  \Or(\xi_q(1+\log(T \max\{S,1\}))\,,
    \end{aligned}}
\end{equation}

We now consider the truncation time as in \eqref{eq:tauaw}:
\begin{align}  \label{eq:tautgg}
    T_g: = \inf \left \{T > 0\,;\ u^{s-1} \le \exp\left(\frac{s}{2e(\beta + \xi_w/4S)^{1/s}} u \right)\ \text{for any $u \ge T^{1/s}$} \right\},
\end{align} 
which satisfy 
\begin{align*}
    T_g = \Theta\left((\beta+\xi_w/4S)\log\left(\beta+\xi_w/4S\right)\right).
\end{align*}
\re{Let $T = \Theta\left((\beta+\xi_w/4S)\log^s(\left(\beta+\xi_w/4S\right)\max\{S,1\})\right) > T_g$. It follows from \cref{auxeq:finiteg,eq:decayg} and some similar estimates as in \cref{auxeq:inttf} that 
\begin{align*}
    \int_\R |g(t)| \ud t &= \int_{-T}^{T}|g(t)| \ud t+ \mathcal{O}\left(\int^\infty_{T^{1/s}}\xi_q S \exp\left(-\frac{s}{e\left(\beta+\xi_w/4S\right)^{1/s}}u\right)u^{s-1}\rd u\right)\\
    & =\int_{-T}^{T}|g(t)| \ud t+\mathcal{O}\left(\int^\infty_{T^{1/s}}\xi_q S \exp\left(-\frac{s}{2e\left(\beta + \xi_w/4S\right)^{1/s}}u\right)\rd u\right)\\
    & = \Or\left(\xi_q(1+\log((\beta +\xi_w/S)\max\{S,1\}))\right)\,.
\end{align*}
The proof is complete.}
\end{proof}

We next prove \cref{lem:discretization_L_G}. \rre{We first restate it with more explicit parameter dependence.} 

\begin{prop}[restatement of \cref{lem:discretization_L_G}]\label{lem:discretization_L_Gre}
Under \cref{assumption:q}, we assume $\beta>0$, $\|A^a\|\leq 1$ for any $a\in\mc{A}$. When \re{$\tau\le \frac{\pi}{\norm{H} + 2S}$} and  
\begin{equation*} 
  (M-1)\tau = \Omega \left((\beta \xi_u + \xi_w/S) \log\left(\beta \xi_u + \xi_w/S\right) \right)\,,
\end{equation*}
it holds that 
\begin{equation}\label{eqn:approx_L}
\left\|L_a-\sum^{2 M - 1}_{m = 0}f^a(t_m) A^a(t_m) \tau\right\|\leq C_fS\exp\left(-\frac{s ((M-1)\tau)^{1/s}}{2(\beta \xi_u + \xi_w/S)^{1/s}e} \right)\,,
\end{equation}
with
\begin{align*} 
   C_f = \mathcal{O}\left(\xi_q^2(\beta \xi_u + \xi_w/S)^{1/s}\right)\,,
\end{align*}
and 
\begin{equation*}
\begin{aligned}
&\left\|G-\sum^{2 M - 1}_{m = 0}g(t_m) H_L(t_m) \tau\right\| \le C_g S |\mathcal{A}|\exp\left(-\frac{s((M-1)\tau)^{1/s}}{2e(\beta+\xi_w/4S)^{1/s}}\right)\,,
\end{aligned}
\end{equation*}
with 
\begin{equation*} 
      C_g = {\mathcal{O}}\left(\xi_q(\xi_q C_{1,u} + \xi^2_q \xi_w)^2\left(\beta+\xi_w/4S\right)^{1/s} \log^2((\beta \xi_u + \xi_w/S)\max\{S,1\})\right)\,.
\end{equation*}
Here $\Omega(\cdot)$, $\mathcal{O}(\cdot)$ absorbs some constant depending on $s$. 
\end{prop}

We shall employ the Poisson summation formula in the following lemma~\cite[Theorem 4.4.2]{pinsky2008introduction}. 

\begin{lem}[Poisson summation formula]\label{lem:poisson}
Given any $h\in L^1(\mathbb{R}^d)$ with inverse Fourier transform: 
\begin{equation*}
    \widehat{h}(y) = \int_{\mathbb{R}^d}h(x)e^{i x\cdot y}\, \mathrm{d}x\,,
\end{equation*}
for any $y\in\mathbb{R}^d$ and $\tau>0$, there holds
\begin{align*}
    \sum_{\textbf{n}\in\mathbb{Z}^d}\widehat{h}\left(y+\frac{2\pi\textbf{n}}{\tau}\right)=\sum_{\textbf{n}\in\mathbb{Z}^d}\tau^d h\left(\textbf{n}\tau\right)e^{iy\cdot \textbf{n}\tau}\,.
\end{align*}
\end{lem}

We are now ready to show \cref{lem:discretization_L_Gre}.

\begin{proof}[Proof of \cref{lem:discretization_L_Gre}] 
In the proof, we shall consider an infinite time grid defined by $\{m\tau\}_{m \in \mathbb{Z}}$. The grid $\{t_m\}_{m = 0}^{2M-1}$ is given by the subset $\{m\tau\}_{m = - M}^{M-1}$ as in the statement of \cref{lem:discretization_L_Gre}.  
We omit the upper index $a$ in $f^a$. 

\medskip
\noindent {\bf Estimate for the jump $L_a$.}
Thanks to \cref{lem:property_f_g}, there holds 
\begin{equation} \label{auxeq:La1}
\begin{aligned}
     &\left\|\sum^\infty_{m=-\infty}f(m\tau) A^a(m \tau) \tau - \sum^{2M-1}_{m=0}f(t_m) A^a(t_m) \tau \right\| \le \|A^a\|\sum_{|m|\ge M}|f(m\tau)|\tau \\
   \le & \re{\frac{\xi^2_qS}{\pi}\exp\left(\frac{es}{2}-\frac{s}{e}\left(\frac{|m \tau|}{\beta \xi_u+\xi_w/S}\right)^{1/s}\right) \tau}\,.
\end{aligned}
\end{equation}
By the monotonicity of the exponential function, we have 
\begin{equation} \label{auxeq:La2}
    \begin{aligned}
 & \sum_{|m|\ge M} \exp\left(-\frac{s}{(\beta \xi_u+\xi_w/S)^{1/s}e}\left| m\tau \right|^{1/s}\right)\tau \\   
\le & 2 \int^\infty_{(M-1)\tau}\exp\left(-\frac{s}{(\beta \xi_u+\xi_w/S)^{1/s}e}\left|t\right|^{1/s}\right)\mathrm{d}t \\
\le & 2 s \int^\infty_{((M-1)\tau)^{1/s}} u^{s-1}\exp\left(-\frac{s}{(\beta \xi_u + \xi_w/S)^{1/s}e} u \right)\mathrm{d}u\,.
    \end{aligned}
\end{equation}
Let $T_f$ be defined as in \cref{eq:tauaw}. 
Then, when $(M-1)\tau \ge T_f$, 
we can compute, by \eqref{auxeq:La2}, 
\begin{align*}
    \sum_{|m|\ge M} \exp\left(-\frac{s}{(\beta \xi_u+\xi_w/S)^{1/s}e}\left|m\tau\right|^{1/s}\right)\tau &\le 2 s  \int^\infty_{((M-1)\tau)^{1/s}} \exp\left(-\frac{s}{2(\beta \xi_u + \xi_w/S)^{1/s}e} u \right)\mathrm{d}u \\
    & \le 4 s \frac{(\beta \xi_u + \xi_w/S)^{1/s}e}{s} \exp\left(-\frac{s ((M-1)\tau)^{1/s}}{2(\beta \xi_u + \xi_w/S)^{1/s}e} \right)\,.
\end{align*}
Combining this with \cref{auxeq:La1}, we find 
\begin{equation} \label{auxeq:La}
\begin{aligned}
     &\left\|\sum^\infty_{m=-\infty}f(m\tau) A^a(m \tau) \tau - \sum^{2M-1}_{m=0}f(t_m) A^a(t_m) \tau \right\|  \\
   \le & \sum_{|m|\ge M}\frac{\xi^2_qS}{\pi}\exp\left(\frac{es}{2}-\frac{s}{(\beta \xi_u+\xi_w/S)^{1/s}e}\left|m\tau\right|^{1/s}\right)\tau \\
   \le & C_s \xi_q^2 S (\beta \xi_u + \xi_w/S)^{1/s} \exp\left(-\frac{s ((M-1)\tau)^{1/s}}{2(\beta \xi_u + \xi_w/S)^{1/s}e} \right)\,.
\end{aligned}
\end{equation} 
Then, by the triangle inequality, we readily have 
\small
\begin{equation}\label{eqn:f_simulation_error}
\begin{aligned}
&\left\|L_a-\sum^{2M-1}_{m=0}f(t_m) A^a(t_m) \tau \right\|\\
\leq &\left\|\int_{-\infty}^\infty f(t) A^a(t) \ud t-\sum^\infty_{m=-\infty}f(m\tau) A^a(m \tau) \tau \right\| + \left\|\sum^\infty_{m=-\infty}f(m\tau) A^a(m \tau) \tau  -\sum^{2M-1}_{m=0}f(t_m) A^a(t_m) \tau \right\|\\
\leq&\left\|\int_{-\infty}^\infty f(t) A^a(t) \ud t-\sum^\infty_{m=-\infty}f(m\tau) A^a(m \tau) \tau \right\| 
+ \re{C_f S\exp\left(-\frac{s ((M-1)\tau)^{1/s}}{2(\beta \xi_u + \xi_w/S)^{1/s}e} \right)}\,,
\end{aligned}
\end{equation}
\normalsize
with constant 
\begin{align*}
  \re{C_f = \mathcal{O}\left(\xi_q^2(\beta \xi_u + \xi_w/S)^{1/s}\right)\,.}
\end{align*}
Note from \eqref{eqn:L_a_formula} that $\int^\infty_{-\infty}f(t) A^a(t) \ud t=\sum_{i,j}\widehat{f}(\lambda_i-\lambda_j) P_i A^a P_j$, and from \cref{assumption:q} that 
\begin{equation*}
    \text{$|\lambda_i-\lambda_j|\leq 2\|H\|$ for any $\lad_i, \lad_j \in {\rm Spec}(H)$} \q \text{and} \q \text{ $\mathrm{supp}(\widehat{f})\subset [-S,S]$}\,.
\end{equation*}
For $\tau<\frac{2\pi}{2\norm{H} + S}$, by Poisson summation formula in \cref{lem:poisson}, we have 
\[
\widehat{f}(\lambda_i-\lambda_j)=\sum^\infty_{m=-\infty}\widehat{f}\left(\lambda_i-\lambda_j+\frac{2\pi m}{\tau}\right)=\sum^\infty_{m=-\infty}f(m\tau)e^{i(\lambda_i-\lambda_j)m\tau}\tau\,.
\]
Plugging the above formula into the first term of \eqref{eqn:f_simulation_error}, there holds 
\[
\begin{aligned}
&\int^\infty_{-\infty}f(t) A^a(t) \ud t-\sum^\infty_{m=-\infty}f(m\tau)A^a(m\tau) \tau\\
=&\sum_{i,j}\left(\widehat{f}(\lambda_i-\lambda_j)-\sum^\infty_{m=-\infty}f(m\tau)e^{i(\lambda_i-\lambda_j)m\tau}\tau\right)P_i A^a P_j = 0\,,
\end{aligned}
\]
which then implies 
\[
\begin{aligned}
&\left\|L_a-\sum^{2 M -1}_{m = 0}f(t_m) A^a(t_m) \tau\right\|
\le \re{C_fS\exp\left(-\frac{s ((M-1)\tau)^{1/s}}{2(\beta \xi_u + \xi_w/S)^{1/s}e} \right)}\,.
\end{aligned} 
\]
and concludes the proof of \eqref{eqn:approx_L}.

\medskip
\noindent {\bf Estimate for the coherent term $G$.} 
\re{By the inequality \eqref{eq:inttf} in \cref{lem:property_f_g}, we obtain
\begin{equation*}
\|L_a\|\leq \int_{\mathbb{R}}|f(t)|\|A^a\|\ud t = \mathcal{O}\left( (\xi_q C_{1,u} + \xi^2_q \xi_w) \log ((\beta \xi_u  + \xi_w/S) \max\{S,1\})\right)\,,    
\end{equation*}
which implies, recalling \cref{def:hll},
\begin{align*}
    \sup_{t\in\mathbb{R}}\|H_L(t)\| \le \sum_{a\in\mathcal{A}}\|L_a\|^2 = \mathcal{O}\left(|\mc{A}| (\xi_q C_{1,u} + \xi^2_q \xi_w)^2 \log^2((\beta \xi_u  + \xi_w/S) \max\{S,1\})\right)\,.
\end{align*}
Again by \cref{lem:property_f_g} with similar estimates as in \eqref{auxeq:La} for $L_a$, when $(M-1)\tau \ge T_g$ with $T_g$ defined in \eqref{eq:tautgg}, there holds 
\begin{equation*}
    \begin{aligned}
&\left\|\sum^\infty_{m=-\infty}g(m\tau) H_L(m \tau) \tau - \sum^{2M-1}_{m=0}g(t_m) H_L(t_m) \tau \right\|\\
\leq & \sup_{t\in\mathbb{R}}\|H_L(t)\|\sum_{
|m|\ge M}\left|g(m\tau)\right|\tau 
\le C_g S |\mathcal{A}|\exp\left(-\frac{s((M-1)\tau)^{1/s}}{2e(\beta+\xi_w/4S)^{1/s}}\right)\,,
\end{aligned} 
\end{equation*}
where 
\begin{equation*}
    C_g = {\mathcal{O}}\left(\xi_q(\xi_q C_{1,u} + \xi^2_q \xi_w)^2\left(\beta+\xi_w/4S\right)^{1/s} \log^2((\beta \xi_u  + \xi_w/S) \max\{S,1\})\right)\,.
\end{equation*}
It follows that}
\small
\begin{equation}\label{eqn:g_simulation_error}
\begin{aligned}
&\left\|G-\sum^{2M-1}_{m=0}g(t_m) H_L(t_m) \tau \right\|\\
\leq &\left\|\int_{-\infty}^\infty g(t) H_L(t) \ud t-\sum^\infty_{m=-\infty}g(t_m) H_L(t_m) \tau  \right\| + \left\|\sum^\infty_{m=-\infty}g(t_m) H_L(t_m) \tau   -\sum^{2M-1}_{m=0}g(t_m) H_L(t_m) \tau  \right\|\\
\leq&\left\|\int_{-\infty}^\infty g(t) H_L(t) \ud t-\sum^\infty_{m=-\infty}g(t_m) H_L(t_m) \tau  \right\| 
+ \re{C_g S |\mathcal{A}|\exp\left(-\frac{s((M-1)\tau)^{1/s}}{2e(\beta+\xi_w/4S)^{1/s}}\right)}\,.
\end{aligned}
\end{equation}
\normalsize
For  $\tau<\frac{2\pi}{2\norm{H} + 4S}$, again by Poisson summation formula in \cref{lem:poisson}, the first term of \eqref{eqn:g_simulation_error} is zero. By definition \eqref{eqn:G_KMS} of $G$ and above estimates, it holds that 
\begin{align*}
\left\|G-\sum^{2M-1}_{m=0}g(t_m) H_L(t_m) \tau \right\|\leq \re{C_g S |\mathcal{A}|\exp\left(-\frac{s((M-1)\tau)^{1/s}}{2e(\beta+\xi_w/4S)^{1/s}}\right)}\,.
\end{align*}
The proof is complete.      
\end{proof}

Finally, the simulation of the algorithm requires the preparation of oracles~\eqref{eqn:prep_f}-\eqref{eqn:prep_g_bar}, where the normalization factors $Z_f$ and $Z_g$ affect the algorithm complexity. In the following theorem, we demonstrate that the discretization normalization constant can be bounded by the $L^1$ norm of $f$ and $g$ when the discretization step $\tau$ is sufficiently small.
\begin{lem} \label{lem:zfzg}
Under \cref{assumption:q}, for any given $T > 0$, there exists small \re{$\tau = \Theta\left(1/\xi_q^2 S^2 T\right)$}  such that for any integer $M$ with $M \tau \le T$, 
\begin{equation} \label{eq:estzf}
    \sum^M_{m=-M}|f(m\tau)| \, \tau\leq \|f\|_{L^1(\R)}+1\,,
\end{equation}
and
\begin{equation} \label{eq:estzg}
\sum^M_{m=-M}|g(m\tau)|\,\tau\leq \|g\|_{L^1(\R^2)}+1\,.   
\end{equation}
\end{lem}
\begin{proof}
We note
\begin{equation*}
    \left|f'(t)\right|=\left|\frac{1}{2\pi} \int_\R  u(\beta \nu)w(\nu/S)e^{-\beta \nu/4}\nu e^{-it\nu}\, \mathrm{d} \nu\right|\leq \frac{\xi_q}{2\pi}\int_{\mathbb{R}}|\nu w(\nu/S)| \ud \nu =\mathcal{O}(\xi_q^2 S^2)\,.
\end{equation*}
Similarly, one can obtain 
\begin{align*}
    \re{\left|g'(t)\right|=\mathcal{O}( \xi_q S^2)}\,.
\end{align*}
Then, it follows from the mean-value theorem that 
\begin{equation*}
    \left|\sum^M_{m=-M}|f(m\tau)|\,\tau-\int^{(M+1)\tau}_{-M\tau}\left|f(t)\right|\,\mathrm{d}t\right|=\mathcal{O}\left(  \xi_q^2 S^2 M \tau^2  \right)\,,
\end{equation*}
which implies
\begin{align*}
    \sum^M_{m=-M}|f(m\tau)|\, \tau\leq \|f\|_{L^1(\R)} + \mathcal{O}\left(  \xi_q^2 S^2 M \tau^2  \right)\,.
\end{align*}
Thus \cref{eq:estzf} follows. For the estimate of $g$, similarly, by the mean-value theorem, we find
\begin{equation*}
    \left|\sum^M_{m=-M}|g(m\tau)|\,\tau-\int^{(M+1)\tau}_{-M\tau}\left|g(t)\right|\,\mathrm{d}t\right| = \re{\mathcal{O}\left(\xi_q S^2 M\tau^2 \right)}\,,
\end{equation*}
which means 
\begin{equation*}
    \re{\sum^M_{m=-M}|g(m\tau)|\,\tau \leq \|g\|_{L^1(\R^2)}+ \mathcal{O}\left(\xi_q S^2 M\tau^2\right)\,.}
\end{equation*}
This concludes the proof of \cref{eq:estzg}. 
\end{proof}

\rre{We finally prove our main result \cref{thm:simulation} restated below with explicit parameter dependence.} 

\begin{thm}[restatement of \cref{thm:simulation}] Assume access to weighting functions $\{q^a\}$ satisfying \cref{assumption:q} with any $s>1$, block encodings $U_{\mathcal{A}}$ in \cref{eqn:A_block_encoding}, controlled Hamiltonian simulation $U_H$ in \cref{eqn:U_H}, and prepare oracles for filtering functions $\{f^a\}$ and $g$ in \eqref{eqn:prep_f}--\eqref{eqn:prep_g_bar}. The Lindbladian evolution \eqref{eqn:Lindblad_master_equation} \rre{with $\beta > 0$} can be simulated up to time $t_{\rm mix}$ with an $\eps$-diamond distance, and the total Hamiltonian simulation time is
\begin{equation*}
  \widetilde{\mathcal{O}}\left(C_qt_{\rm mix} \rre{(\beta + 1)}|\mc{A}|^2  \log^{1+s}\left(1/\epsilon\right)\right)\,,
\end{equation*}
where the constant $C_q$ is defined as follows:
\begin{align*}
\re{C_q:= {\rm C}_{u,w} (\xi_u+ \xi_w/S) (\log^{2 + s}(\max\{S,1\}) + 1) \q \text{with}\q {\rm C}_{u,w} =(\xi_qC_{1,u}+\xi^2_q\xi_w)^2\,.}
\end{align*}
In addition, the algorithm requires 
\begin{equation*}
\widetilde{\mathcal{O}}\left(\log(\xi_q \max\{S,1\} (S + \norm{H})) + \log^2(t_{\rm mix} |\mc{A}|/\epsilon) + \log^2 {\rm C}_{u,w}   + \log(\beta \xi_u + \xi_w/S)\right)\,.
\end{equation*}
number of additional ancilla qubits for the prepare oracles and simulation. The $\widetilde{\mathcal{O}}$ absorbs a constant only depending on $s$ and subdominant polylogarithmic dependencies on parameters $t_{\rm mix}$, $\|H\|$, $|\mathcal{A}|$, $S$, $\xi_q$, $\xi_u$, $\xi_w$, and $\beta$. 
\end{thm}

\begin{proof} 
By \cref{thm:simulation_LW} with estimates in \cref{eq:zf,eq:zg}, we have 
\begin{equation} \label{eq:normbe}
\re{\begin{aligned}
   \|\mathcal{L}\|_{\rm be} &= \frac{1}{2}Z_g |\mathcal{A}| + \frac{1}{2}Z^2_f|\mathcal{A}| \\
   & =  \mathcal{O}\left({\rm C}_{u,w}\log^2((\beta \xi_u + \xi_w/S)\max\{S,1\})|\mathcal{A}|\right)\,.
\end{aligned}}
\end{equation} 
Recalling \cref{lem:discretization_L_G} and \cref{lem:zfzg}, we set a truncation time 
\begin{equation}  \label{eq:estmtau}
   \re{T = \Theta \left((\beta \xi_u + \xi_w/S)\log^s\left( (C_f + C_g) \max\{S,1\} |\mc{A}|  \|\mathcal{L}\|_{\rm be}t_{\rm mix} \eps^{-1} \right)\right)\,,}
\end{equation}
and the step size
\begin{equation} \label{eq:steptime}
   \re{\tau= \Theta \left(\min\left\{ \frac{1}{\xi_q^2S^2 T},\frac{1}{S+\|H\|}\right\}\right)\,,}
\end{equation}
and then choose $M = 2^{\mathfrak{m} - 1}$ such that $(M-1) \tau \le T \le M \tau$.
This allows us to control the block-encoding error as follows:
\begin{equation*}
\re{(C_f + C_g) S |\mc{A}| \exp\left(-\frac{s((M-1)\tau)^{1/s}}{2e(\beta \xi_u+\xi_w/S)^{1/s}}\right) \le \epsilon/({t_{\rm mix}}\|\mathcal{L}\|_{\rm be})\,.}
\end{equation*}
Note that one query to the block encoding of $L_a$ requires one query to $U_L$, by \cref{thm:simulation_LW} and \cref{eq:normbe}, the simulation \re{requires
\begin{align*}
   & \mathcal{O}\left(t_{\rm mix}\|\mathcal{L}\|_{\rm be}\log\left(t_{\rm mix}\|\mathcal{L}\|_{\rm be}/\epsilon\right) |\mc{A}|\right) = \wt{\Or}\left(t_{\rm mix}\|\mathcal{L}\|_{\rm be}\log\left(1/\epsilon\right)|\mc{A}|\right) \\ &= \widetilde{\mathcal{O}}\left(t_{\rm mix}{\rm C}_{u,w} \log^2((\beta \xi_u + \xi_w/S)\max\{S,1\}) |\mathcal{A}|^2 
   \log\left(1/\epsilon\right)\right)
\end{align*}
queries to $U_{L}$. Similarly, we need 
\begin{equation*}
  \widetilde{\mathcal{O}}\left(t_{\rm mix} 
    {\rm C}_{u,w} \log^2((\beta \xi_u + \xi_w/S)\max\{S,1\})
    |\mc{A}|  \log\left(1/\epsilon\right)\right)
\end{equation*}
queries to $U_G$.} Combining this with the cost $\mathcal{O}(M\tau) = \mathcal{O}(T)$ of one query to $U_G$ and $U_{L}$, we can estimate the  total Hamiltonian simulation time as follows: 
\begin{equation*}
\re{\begin{aligned}
     &\widetilde{\mathcal{O}}\left(t_{\rm mix} {\rm C}_{u,w} (\beta \xi_u+ \xi_w/S) (\log^{2 + s}(\max\{S,1\}) + 1)|\mc{A}|^2  \log^{1+s}\left(1/\epsilon\right)\right) \\ =\, & \widetilde{\mathcal{O}}\left(C_q t_{\rm mix}\rre{(\beta + 1)} |\mc{A}|^2  \log^{1+s}\left(1/\epsilon\right)\right)\,.    
    \end{aligned}}
\end{equation*}

Next, we consider the number of extra ancilla qubits, again by \cref{thm:simulation_LW}. We first note 
\begin{equation*}
 \log(t_{\rm mix}\norm{\mc{L}_{\rm be}}/\epsilon) =   \widetilde{\mathcal{O}}\left(\log(t_{\rm mix} |\mc{A}|/\epsilon) + \log {\rm C}_{u,w} + \log \log((\beta \xi_u + \xi_w/S)\max\{S,1\}) \right),
\end{equation*}
and hence there holds 
\begin{align*}
    & \Or\left(\log\left(t_{\rm mix}\|\mathcal{L}\|_{\rm be}/\epsilon\right)\left(\log |\mathcal{A}|+\log\left(t_{\rm mix}\|\mathcal{L}\|_{\rm be}/\epsilon\right)\right)\right) \\
    =\, & \widetilde{\mathcal{O}}\left(  \log^2(t_{\rm mix} |\mc{A}|/\epsilon) + (\log {\rm C}_{u,w} + \log \log((\beta \xi_u + \xi_w/S)\max\{S,1\}))^2 \right).
\end{align*}
For the preparation oracles, by \cref{eq:steptime,eq:estmtau}, we find 
\begin{align*}
    M & = \Theta(\max\{\xi_q^2 S^2 T^2, T(S + \norm{H})\} = \Or(\xi_q^2 \max\{S,1\} (S + \norm{H}) T^2) \\
    & = \Or(\xi_q^2 \max\{S,1\} (S + \norm{H}) (\beta \xi_u + \xi_w/S)^2\log^{2s}\left((C_f + C_g) \max\{S,1\} |\mc{A}|  \|\mathcal{L}\|_{\rm be}t_{\rm mix} \eps^{-1} \right))\,.
\end{align*}
This implies that we need 
\begin{multline*}
      \mathfrak{m} =  \Or \big(\log(\xi_q \max\{S,1\} (S + \norm{H})) + \log(\beta \xi_u + \xi_w/S)  \\ +  \log\log\left((C_f + C_g) \max\{S,1\} |\mc{A}|  \|\mathcal{L}\|_{\rm be}t_{\rm mix} \eps^{-1} \right)\big).
\end{multline*}
additional ancilla qubits for the preparation of $\textbf{Prep}_{f}$, $\textbf{Prep}_{\overline{f}}$, $\textbf{Prep}_{g}$, and $\textbf{Prep}_{\overline{g}}$. Finally, $\textbf{Prep}_{\mathcal{A}}$ requires $\log(|\mathcal{A}|)$ ancilla qubits. Adding these quantities together concludes the proof.
\end{proof}

\end{document}

%% file: preamble.tex
\usepackage{color, xcolor, colortbl}
\usepackage{graphicx}
\usepackage[margin=1.2in]{geometry}
\usepackage{amsmath,amssymb,amsthm,amsfonts}
\usepackage{algorithm}
\usepackage{algorithmic}
\usepackage{dcolumn}
\usepackage{epstopdf}
\usepackage{bm}
\usepackage[caption=false]{subfig}
\usepackage{appendix}
\usepackage{multirow}
\usepackage{braket}
\usepackage[english]{babel}
\usepackage[colorlinks]{hyperref}
\usepackage[capitalize]{cleveref}
\usepackage[T1]{fontenc}
\usepackage{xpatch}
\usepackage{tikz}
\usepackage{adjustbox}
\usepackage{xspace}
\usepackage[roman]{complexity}
\usepackage{xparse}

\newcommand{\mc}[1]{\mathcal{#1}}

\newcommand{\wt}[1]{\widetilde{#1}}

\newcommand{\abs}[1]{\left\lvert#1\right\rvert}
\newcommand{\norm}[1]{\left\lVert#1\right\rVert}

\newcommand{\ud}{\,\mathrm{d}}
\newcommand{\Or}{\mathcal{O}}

\newcommand{\RR}{\mathbb{R}}
\newcommand{\CC}{\mathbb{C}}

\newtheorem{thm}{\protect\theoremname}
\theoremstyle{plain}

\theoremstyle{plain}
\newtheorem{lem}[thm]{\protect\lemmaname}
\theoremstyle{plain}
\newtheorem{rem}[thm]{\protect\remarkname}
\theoremstyle{plain}
\newtheorem*{lem*}{\protect\lemmaname}
\theoremstyle{plain}
\newtheorem{prop}[thm]{\protect\propositionname}
\theoremstyle{plain}
\newtheorem{cor}[thm]{\protect\corollaryname}
\newtheorem{assumption}[thm]{\protect\assumptionname}
\newtheorem{defn}[thm]{\protect\definitionname}

\theoremstyle{definition}

\providecommand{\definitionname}{Definition}
\providecommand{\assumptionname}{Assumption}
\providecommand{\corollaryname}{Corollary}
\providecommand{\lemmaname}{Lemma}
\providecommand{\propositionname}{Proposition}
\providecommand{\remarkname}{Remark}
\providecommand{\examplename}{Example}
\providecommand{\theoremname}{Theorem}
\providecommand{\conjecturename}{Conjecture}

\NewDocumentCommand{\ketbra}{mG{#1}}{\mathinner{|{#1}\rangle\!\langle{#2}|}}

\usetikzlibrary{quantikz2}

\usetikzlibrary{fit}
\tikzset{%
  highlight/.style={rectangle,rounded corners,fill=blue!15,draw,fill opacity=0.3,thick,inner sep=0pt}
}

